\documentclass[11pt,onecolumn,draftcls]{IEEEtran}

\usepackage{amssymb}
\usepackage{amsthm}
\usepackage{amsmath}   % From the American Mathematical Society
\usepackage{pstricks,pst-plot,pst-node,setspace}
\usepackage{graphicx}  % Written by David Carlisle and Sebastian Rahtz
\usepackage{algorithmic}
\usepackage{algorithm}

\usepackage{ifthen}
\newboolean{DoubleColumn}
\setboolean{DoubleColumn}{true} % boolvar=true or false
\graphicspath{{./images/}}
\usepackage{subfigure,comment}
\usepackage{url}
\usepackage{psfrag}    % Written by Craig Barratt, Michael C. Grant
\usepackage{multirow}
\newcommand{\vect}[1]{ \boldsymbol{#1} }
\newcommand{\matr}[1]{ \boldsymbol{#1} }
\newcommand{\tvect}[1]{{\overline{\boldsymbol{#1}}} }

\newcommand{\hvect}[1]{\widehat{\boldsymbol{#1}}}

\newcommand{\e}[1]{ \mathrm{e}^{#1} }

\DeclareMathOperator*{\diag}{diag}

\DeclareMathOperator*{\argmax}{argmax}

\DeclareMathOperator*{\E}{\mathbb{E}} % expectation operator

\newtheorem{prop}{Proposition}

\newcommand{\CGaussPDF}{\mathrm{CN}}

\newcommand{\GammaPDF}{\mathrm{Ga}}
\newcommand{\EXP}[1]{\E_{q\left(#1\right)}}
\newcommand{\AlphaInf}{\widehat{\alpha}_{l}^{[\infty]}}

\newcommand{\notl}{-l}
\title{Joint Detection and Super-Resolution Estimation of Multipath Signal Parameter Using Incremental Automatic Relevance Determination}
\author{Dmitriy Shutin and Nicolas Schneckenburger} %

\begin{document}
\maketitle
%\begin{frontmatter}



\begin{abstract}
The presented work investigates a sparse Bayesian incremental automatic relevance determination (IARD) algorithm in the context of multipath parameter estimation in a super-resolution regime.
The corresponding estimation problem is highly nonlinear and, in general, requires an estimation of the number of multipath components.
In the IARD approach individual multipath components are processed sequentially, which permits a tractable convergence analysis of the corresponding inference expressions.
This leads to a simple condition, termed here a pruning condition, that determines if a multipath component is ``sparsified'' or retained in the model, thus realizing  a sparse estimator and permitting a fast and adaptive realization of the estimation algorithm.
Yet previous experiments demonstrated that IARD fails to select the correct number of components when the parameters entering nonlinearly the multipath model are also estimated.
To understand this effect, an analysis of the statistical structure of the pruning condition from the perspective of statistical hypothesis testing is proposed.
It is shown that the corresponding test statistic in the pruning condition follows an extreme value distribution.
As a result, when applied to the problem of multipath estimation, the standard IARD algorithm implements a statistical test with a very high probability of false alarm.
This leads to insertion of estimation artifacts and underestimation of signal sparsity.
Moreover, the probability of false alarm worsens as the number of measured signal samples grows.
Based on the developed statistical interpretation of the IARD, an optimal adjustment of the pruning condition is proposed.
This permits a reliable and efficient removal of estimation artifacts and joint estimation of signal parameters, as well as optimal model order selection within a sparse Bayesian learning framework.
The presented experiments demonstrate the effectiveness of this approach.
\end{abstract}

\begin{keywords}
Super-resolution channel estimation, model order selection, sparse Bayesian learning.
\end{keywords}

%\end{frontmatter}

% \linenumbers

%%%%%%%%%%%%%%%%%%%%%%%%%%%%%%%%%%%%%%%%%%%%%%%%%%%%%%%%%%%%%%%%%%%%%%%%%%%%%%%%%%
\section{Introduction}\label{sec:Intro}
Multipath propagation is known to have a significant impact on the performance of wireless communication or localization systems.
However, when the multipath channel structure is known, it can offer a key to a reliable high-rate data communication or accurate localisation.

Typically, a multipath wireless channel is assumed to consist of a linear combination of a finite number of $L$ discrete propagation paths, which we term multipath components, embedded in a white additive ambient noise and a non-white random process that represents diffuse propagation.
While multipath components can be deterministically described by a set of parameters \--- dispersion parameters that characterize specular waves propagating from the transmitter site to the receiver site, such as a propagation delay, direction of departure, direction of arrival, and a Doppler frequency \--- diffuse components are of a random nature and are characterized statistically \cite{Salmi2009,Richter2005,Jost2012a}.
In this work we are concerned with an estimation of the discrete multipath components as they are a very sought-after characteristic of a wireless propagation channel due to their direct relationship to the geometry of the propagation environment.

Historically, the problem of multipath component parameter estimation has been solved using a combination of two techniques: super-resolution (SR) parameter estimation algorithms (see e.g., \cite{Richter2005,KrimViberg20yearsArray, Fleury_SAGE99} and references therein) and model order selection \cite{Stoica2004,lanterman00schwarz, WaxKailath85}.
Parameter estimation algorithms are used to find the parameters of multipath components given measurement data and a model of a multipath  channel with a known number of superimposed components.
SR property of the estimation algorithm is essential, as an accurate estimation of component parameters beyond bandwidth resolution is often required.
%requires resolving individual propagation paths beyond the resolution imposed by the bandwidth of the measurement equipment.
Expectation-Maximization (EM) type of algorithms \cite{Richter2005,KrimViberg20yearsArray,FederWeinsteinEM88, Fleury_SAGE99}
%BessonStoica00
are often used for this purpose.
They allow simplifying the numerical optimization of the objective function with respect to the dispersion parameters that enter the channel model nonlinearly.
Unfortunately, these techniques are applicable only when the order of the model, i.e., the number of specular components is known \--- a requirement that is rarely satisfied in practice.
This has motivated the use of model order selection techniques, such as Bayesian Information criterion or Minimum description length and similar \cite{Jay2005, lanterman00schwarz,Stoica2004, WaxKailath85} to determine the number of components in the model.
These methods select the model order by balancing the model complexity, i.e., a total number of parameters to be estimated, with a norm of the residual error.
Yet for the considered problem these algorithms become computationally very demanding: in order to find the optimal model order, the parameters of models with different number of components have to be estimated first, and then compared using selected criterion.
In practice, the number of components can range from a only a few to several tens of components, making separate parameter estimation and model order selection very inefficient, especially in time-varying scenarios, where the number of components can change \cite{Jost2012,Gentner2013}.

To make estimation more efficient, we propose a variational Bayesian wireless channel estimator that combines model order selection and parameter estimation within a single framework.
The proposed solution is based on merging a variational Bayesian parameter estimation \cite{Beal2003, BishopPR}, which generalizes classical EM-based SR parameter estimation algorithms, and sparse Bayesian learning (SBL) techniques \cite{Tipping2001,WipfRao04,Tzikas2008}.
Sparse reconstruction of a multipath channel can effectively solve the model order selection problem, since irrelevant multipath components will be ``sparsified'' by the algorithm; sparsity, thus, effectively controls the complexity of the estimated models.

Such multipath estimation approaches have been to some extent explored in \cite{ShutinVBSAGE} and \cite{Shut1307:Incremental}.
In \cite{ShutinVBSAGE} the authors casted the Space Alternating Generalized Expectation-Maximization (SAGE) algorithm for multipath  parameter estimation\footnote{See \cite{FesslerSAGE94} and \cite{Fleury_SAGE99} for the details on the SAGE algorithm.} in a variational Bayesian framework.
The new algorithm, termed variational Bayesian SAGE (VB-SAGE), introduces sparsity priors to jointly estimate model order via sparsity penalization and estimate the parameters of multipath components.
The VB-SAGE algorithm makes a typical assumption on the independence of individual components.
In \cite{Shut1307:Incremental} this assumption is relaxed by considering correlations between the gains of propagation paths.
By adopting a special class of SBL algorithms, known as incremental Automatic Relevance Determination (IARD) \cite{Tipping2003, ShutinFastRVM,Shutin_IARD,ShutinBuchgraber2012}, a new algorithm is proposed that, as we will show here, generalizes the VB-SAGE algorithm.
A key feature of both VB-SAGE and IARD algorithms is the structure of variational inference expressions that leads to a simple numerical condition for removing or keeping a component in the model. 
It is this condition that eventually leads to sparse estimate. 
Further in the text we refer to this condition as a pruning condition.
The pruning condition permits the reduction of the model complexity ``on the fly'', while the components are updated.
In this way model order selection and parameter estimation are realized jointly.

It has been observed, however, that some of the estimated multipath components have small, yet non-zero weights \cite{ShutinVBSAGE}.
In other words, the IARD and VB-SAGE estimators compress the measured signal, but overestimate the model order.
To cancel erroneous components an empirical threshold was adopted in \cite{ShutinBuchgraber2012,ShutinVBSAGE,ShutinFastRVM}.
The selection of the threshold exploits the link between the pruning condition and an estimate of the per-component signal-to-noise ratio (SNR).
Yet it remains unclear whether a particular choice of the threshold can be motivated more formally.
A better understanding of these aspects can be exploited not only for improving performance of IARD schemes in the presence of noise and better understanding of the IARD performance in general, but for an accurate and fast extraction of specular multipath components, as we argue in this paper.

Thus, our goals in this work can be formulated as follows: we aim to further the theoretical understanding of IARD within the context of sparse estimation of multipath component and present a more detailed analysis of the pruning condition used in the IARD algorithms.
Specifically, we show that the IARD algorithm generalizes VB-SAGE. 
Also, we demonstrate that the pruning condition used in IARD is equivalent to a statistical hypothesis test applied to a specific multipath component under the assumption that the other multipath components are fixed.
With this new interpretation it becomes possible to show that (i) within the IARD scheme the presence of a component in the model can be determined using a statistical hypothesis test of a desired test size, (ii) the test is a uniformly most powerful (UMP), (iii) probability of false alarm for this test (i.e., the probability of falsely accepting a component in the model) is upper-bounded, with the standard IARD algorithm implementing the test with the highest probability of false alarm. 

Throughout this paper we shall make use of the following notation.
Vectors are represented as boldface lowercase letters, e.g., $\vect{x}$, and matrices as boldface uppercase letters, e.g., $\matr{X}$.
For vectors and matrices $(\cdot)^H$ denotes the Hermitian transpose.
We write $[\vect{X}]_{k,l}$ to denote an element of the matrix $\vect{X}$ at the $k$th row and $l$th column.
The expression $\diag(\vect{x})$ stands for a diagonal matrix with the elements of $\vect{x}$ on the main diagonal.
For some positive-semidefinite matrix $\vect{A}$, notation $\|\vect{x}\|_{\vect{A}}=\sqrt{ \vect{x}^H\vect{A}\vect{x}}$ denotes a weighted $\ell_2$ norm of a vector $\vect{x}$.
We write $\EXP{x} f(x)$ to denote the expectation of the function $f(x)$ under the probability density function $q(x)$.
%; $\diag(\vect{X})$ stands for a vector of diagonal entries of a square matrix $\vect{X}$; $\tr(\matr{X})$ denotes the trace of the matrix $\matr{X}$.
%$[\matr{B}]_{\notl\overline{k}}$ denotes a matrix obtained by deleting the $l$th row and  $k$th column from the matrix $\matr{B}$;
%similarly, $[\vect{b}]_{\notl}$ denotes a vector obtained by deleting the $l$th element from the vector $\vect{b}$.
%We will use $\vect{e}_l=[0,\ldots,0,1,0,\ldots,0]^T$ to denote a vector of all zeros with $1$ at the $l$th position.
Finally, for a random vector $\vect{x}$, $\CGaussPDF(\vect{x}|\vect{a},\matr{B})$ denotes a circular complex multivariate Gaussian pdf with mean $\vect{a}$ and covariance matrix $\matr{B}$; similarly, for a random variable $x$, $\GammaPDF(x|a,b)=\frac{b^{a}}{\Gamma(a)}x^{a-1}\exp(-bx)$ denotes a gamma pdf with parameters $a$ and $b$.

%%%%%%%%%%%%%%%%%%%%%%%%%%%%%%%%%%%%%%%%%%%%%%%%%%%%%%%%%%%%%%%%%%%%%%%%%%%
\section{Signal model}
In the following sections we outline the used signal model. 
Also, the corresponding probabilistic formulation of the inference problem that builds the foundation for the variational Bayesian parameter estimation adopted here is presented.

%%%%%%%%%%%%%%%%%%%%%%%%%%%%%%%%%%%%%%%%%%%%%%%%%%%%%%%%%%%%%%%%%%%%%%%%%%%
\subsection{Multipath channel model}
Consider for simplicity a single-input\---single-output (SISO) wireless channel\footnote{The proposed method can also be extended to MIMO time-variant channels with stationary propagation constellation.
This will, however, lead to a more complicated signal model with additional dispersion parameters, while not adding any new aspect relevant to the understanding of the proposed methods.}.
%The scenario considering a SISO channel seems a sensible compromise between complexity of the model underlying the theoretical analyses and an interesting application in which the proposed method can be demonstrated.
The received signal $y(t)$ can be represented as a superposition of an unknown number $L$ of specular multipath components $w_l s(t;\vect{\theta}_l)$ contaminated by additive noise $\xi(t)$ (see e.g., \cite{Richter2005, ShutinVBSAGE, WirelessCommRapp}):
\begin{equation}
	y(t)=\sum_{l=1}^{L} w_l s(t;\vect{\theta}_l) + \xi(t).
	   \label{equ:ModelCom}
\end{equation}
In \eqref{equ:ModelCom} $w_l$ is a complex-valued multipath gain and $s(t;\vect{\theta}_l)$ is an altered version of some transmitted signal $x(t)$.
The alteration process is described by a (non-linear) mapping $x(t)\mapsto s(t;\vect{\theta}_l)$, where $\vect{\theta}_l$ is the vector of dispersion parameters, e.g., relative delay, Doppler shift, etc.
For a SISO channel, $s(t;\vect{\theta}_l)$ can be represented as $s(t;\vect{\theta}_l)\equiv s(t;\tau_l,\nu_l)=\mathrm{e}^{j2\pi\nu_l t}x(t-\tau_l)$, where $\vect{\theta}_l=[\tau_l,\nu_l]^T$, $\tau_l$ is a delay of the $l$th multipath component and $\nu_l$ is its Doppler shift.
In general, the nonlinear mapping  $x(t)\mapsto s(t,\vect{\theta}_l)$ also includes the measurement system effects, e.g., signal distortions at the transmitter and the receiver due to analog filtering, RF components, etc.
Additive noise $\xi(t)$ is assumed to be a zero-mean wide-sense stationary Gaussian process.
%, i.e., $E\{\xi(t)\xi^*(t+t')\}=R_\xi(t')$.
In addition to white noise, this term will also include effects due to diffuse scattering \cite{Richter2005, Jost2012a}.

In practice the signal $y(t)$ is sampled with the sampling period $T_s$, resulting in $N$ discrete measurement samples.
By stacking the samples in a vector $\vect{y}=[y(0),\ldots,y((N-1)T_s)]^T$, model \eqref{equ:ModelCom} can be rewritten in a more convenient matrix form as
\begin{equation}
        \vect{y}=\sum_{l=1}^{L}w_l\vect{s}(\vect{\theta}_l)+\vect{\xi}=\vect{S}(\vect{\Theta})\vect{w}+\vect{\xi},
        \label{equ:ModelComMatrix}
\end{equation}
where we define
$\vect{s}(\vect{\theta}_l)=[s(0;\vect{\theta}_l),\ldots, s((N-1)T_s;\vect{\theta}_l)]^T$, $\vect{\Theta}=[\vect{\theta}_1,\ldots,\vect{\theta}_L]$, $\vect{w}=[w_1,\ldots,w_L]^T$, and $\vect{S}(\vect{\Theta})=[\vect{s}(\vect{\theta}_1),\ldots,\vect{s}(\vect{\theta}_L)]$.
The term $\vect{\xi}=[\xi(0),\ldots,\xi((N-1)T_s)]^T$ is the additive noise vector that follows a circular complex normal distribution with covariance matrix $\mathbb{E}\left\{\vect{\xi}\vect{\xi}^H\right\}=\vect{\Lambda}^{-1}$.
In the following we will assume that $\vect{\Lambda}$ is known or has been estimated; the estimation of diffuse scattering statistics and white noise statistics we will leave outside the scope of this work.
%We will merely assume that $R_\xi(t')$ is known or has been estimated.

%%%%%%%%%%%%%%%%%%%%%%%%%%%%%%%%%%%%%%%%%%%%%%%%%%%%%%%%%%%%%%%%%%%%%%%%%%%
\subsection{Probabilistic structure of the multipath channel model}
Expression \eqref{equ:ModelComMatrix} is the starting point for the multipath parameter estimation algorithms.
Given \eqref{equ:ModelComMatrix}, the joint model order selection and parameter estimation aims at determining the values of  $L$, $\vect{w}$, and $\vect{\Theta}$.
For fixed $L$ both $\vect{w}$ and $\vect{\Theta}$ can be found using classical maximum a posteriori (or maximum likelihood) approach, which amounts to a numerical maximization of the corresponding probability density function (pdf) $p(\vect{w},\vect{\Theta}|\vect{y})\propto p(\vect{y}|\vect{w},\vect{\Theta})p(\vect{w},\vect{\Theta})$,
where $p(\vect{y}|\vect{w},\vect{\Theta})=\CGaussPDF(\vect{y}|\vect{S}(\vect{\Theta})\vect{w},\vect{\Lambda}^{-1})$ following \eqref{equ:ModelComMatrix}.
Unfortunately, in majority of practical cases the number of multipath components $L$ is not known.
A possible approach to circumvent an explicit specification of the model order consists of imposing sparsity constraints on $\vect{w}$.
The advantage of such approach is a joint model order selection and parameter estimation within a Bayesian inference framework, as will be outlined below.

A classical SBL approach \cite{Tipping2001, WipfRao04, Tzikas2008} assumes a hierarchical factorable prior
%\footnote{It is also possible to extend the SBL prior formulation to priors involving three layers of hierarchy (see e.g. \cite{Pedersen2011} and \cite{Lee2010}).}
$p(\vect{w}|\vect{\alpha})p(\vect{\alpha})=\prod_{l=1}^L p(w_l|\alpha_l)p(\alpha_l)$ for the weights $\vect{w}$, where
$p(w_l|\alpha_l)=\CGaussPDF(w_l|0,\alpha_l^{-1})$.
Parameters $\alpha_l$, also called sparsity parameters, regulate the width of this pdf and must be estimated along with the other model parameters \--- an approach referred to as \emph{empirical Bayes}.

In IARD version of SBL two techniques are combined.
First, the hyperprior $p(\vect{\alpha})$ is assumed to be non-informative by selecting $p(\vect{\alpha})\propto \prod_{l=1}^L \alpha_l^{-1}$. Such choice is known as automatic relevance determination (ARD). 
The resulting inference scheme is then similar to a weighted version of minimum $\ell_1$-norm regression and basis pursuit denoising (see \cite{Candes2008, Wipf2007,CPA:CPA20124}) \---  more traditional ``non-Bayesian'' methods for learning sparse representations.
Second, in the incremental inference approach to the SBL the corresponding objective function is optimized with respect to the parameters of one component per single algorithm iteration.
Such incremental optimization permits a fast estimation of sparsity parameters \cite{Tipping2003, ShutinFastRVM, Shutin_IARD}.
Moreover, it also underlies the EM-based multipath estimation schemes, since it simplifies nonlinear optimizations with respect to dispersion parameters $\vect{\Theta}$.
This motivates a combination of IARD and multipath inference schemes in a single framework.

The joint multipath parameter estimation and model order selection within IARD amounts to inference of the joint posterior pdf
\begin{equation}
    p(\vect{w},\vect{\Theta},\vect{\alpha}|\vect{y})\propto p(\vect{y}|\vect{w},\vect{\Theta})p(\vect{w}|\vect{\alpha})p(\vect{\alpha})p(\vect{\Theta}),
    \label{equ:PosteriorPDF}
\end{equation}
where we explicitly assume that $p(\vect{w},\vect{\Theta},\vect{\alpha})=p(\vect{w}|\vect{\alpha})p(\vect{\alpha})p(\vect{\Theta})$.
Unfortunately, \eqref{equ:PosteriorPDF} cannot be evaluated in closed form, but can be approximated using, e.g., variational Bayesian techniques \cite{BishopPR,Beal2003}.
The latter aims at estimating an approximating pdf $q(\vect{w}, \vect{\Theta}, \vect{\alpha})$ by maximizing the lower bound of the log-evidence $\log p(\vect{y})$ :
\begin{equation}
	\log p(\vect{y})\ge \EXP{\vect{w},\vect{\Theta},\vect{\alpha}}
	\log \frac{p(\vect{w},\vect{\Theta},\vect{\alpha},\vect{y})}{q(\vect{w}, \vect{\Theta}, \vect{\alpha})},
	\label{equ:LowerBoundVB}
\end{equation}
which is equivalent to minimizing the Kullback-Leibler divergence between $q(\vect{w}, \vect{\Theta}, \vect{\alpha})$ and the intractable $p(\vect{w},\vect{\Theta},\vect{\alpha}|\vect{y})$.
The complexity of the inference depends on the choice of $q(\vect{w}, \vect{\Theta}, \vect{\alpha})$.
Here we will assume that
\begin{equation}
    q(\vect{w}, \vect{\Theta}, \vect{\alpha})=q(\vect{w})\prod_{k=1}^L q(\vect{\theta}_k)q(\alpha_k).
    \label{equ:ProxyPDF}
\end{equation}
Let us now specify each factor in \eqref{equ:ProxyPDF}.
First, we will select $q(\vect{\theta}_l)=\delta(\vect{\theta}_l-\hvect{\theta}_l)$.
This assumption results in a point estimate of the dispersion parameters.
This choice simplifies the numerical optimization of the right-hand side of \eqref{equ:LowerBoundVB}.
For the factor $q(\vect{w})$ we will consider two assumptions:
\begin{align}
(\mathrm{A1}):\quad\quad&    q(\vect{w})=\prod_{l=1}^L q(w_l)=\prod_{l=1}^L\CGaussPDF(w_l|\widehat{w}_l,\widehat{\Phi}_l),\label{equ:FactorizationVBSAGE}\\
(\mathrm{A2}):\quad\quad&    q(\vect{w})=\CGaussPDF(\vect{w}|\hvect{w},\hvect{\Phi}).\label{equ:FactorizationIARD}
\end{align}
$\mathrm{A1}$ explicitly enforces a statistical independence between individual multipath components; this assumption underlies the SAGE \cite{Fleury_SAGE99} and the VB-SAGE algorithms \cite{ShutinVBSAGE} for multipath parameter estimation.
Under the assumption $\mathrm{A2}$ the gains of the components are assumed to be correlated.
This formulation is used in a classical SBL and in the IARD algorithm for multipath estimation in \cite{Shut1307:Incremental}.
In the following we will consider both assumptions and investigate their impact on multipath estimation and detection.
Let us mention here that $\mathrm{A1}$ can be obtained as a special case of $\mathrm{A2}$ by constraining $\hvect{\Phi}$ to a diagonal matrix.
The form of the factor $q(\vect{\alpha})$ can be obtained analytically as a maximizer of \eqref{equ:LowerBoundVB} for the chosen form of $q(\vect{\theta}_l)$ and $q(\vect{w})$.
For the IARD case it can be shown \cite{ShutinVBSAGE,ShutinFastRVM} that
$$
    q(\vect{\alpha})=\prod_{l=1}^L q(\alpha_l)=\prod_{l=1}^L \GammaPDF(\alpha_l; 1, \widehat{\alpha}_l^{-1}),
$$
i.e., $q(\alpha_l)$ is parameterized  by a single coefficient $\widehat{\alpha}_l$.

The maximization of the bound in \eqref{equ:LowerBoundVB} then reduces to the estimation of the parameters $\hvect{w}$, $\hvect{\Phi}$, $\widehat{\alpha}_l$, and $\hvect{\theta}_l$, $l=1,\ldots,L$ that parameterize \eqref{equ:ProxyPDF}.
In what follows we describe this in more details.

%%%%%%%%%%%%%%%%%%%%%%%%%%%%%%%%%%%%%%%%%%%%%%%%%%%%%%%%%%%%%%%%%%%%%%%%%%
\subsection{Incremental variational inference of model parameters}\label{sec:IARD}
The IARD algorithm optimizes \eqref {equ:LowerBoundVB} with respect to the parameters of one component per iteration, cycling through the components in a round-robin fashion.
Consider now the variational inference steps for a single component $l$.
We will begin with the estimation of $q(\vect{\theta}_l)$.
To this end we define $\vect{\Theta}_{\notl}=\big[\vect{\theta}_1,\ldots,\vect{\theta}_{l-1},\vect{\theta}_{l+1},\ldots,\vect{\theta}_{L}\big]$ as a set of dispersion parameters obtained by removing $\vect{\theta}_l$ from $\vect{\Theta}$, and assume that the pdfs $q(\vect{w})$, $q(\vect{\alpha})$, and $q(\vect{\Theta}_{\notl})$ are available.\footnote{In other words, we assume that the parameters of the corresponding pdfs are known.}
%Following the standard variation inference steps (see \cite{BishopPR}),
The bound in \eqref{equ:LowerBoundVB} on $\log p(\vect{y})$ with respect to $q(\vect{\theta}_l)$  can then be expressed as $\log p(\vect{y})\ge\EXP{\vect{\theta}_l} \log \big(\tilde{p}(\vect{\theta}_l)/q(\vect{\theta}_l)\big)$, where
\begin{equation}
    \tilde{p}(\vect{\theta}_l)\propto \mathrm{exp}{\left(\EXP{\vect{w},\vect{\Theta}_{\notl}}\log p(\vect{y}|\vect{w},\vect{\Theta})p(\vect{\Theta})\right)}.
    \label{equ:PTheta}
\end{equation}
This bound is maximized when the Kullback-Leibler divergence between $q(\vect{\theta}_l)$ and $\tilde{p}(\vect{\theta}_l)$ is minimal.
Due to the assumed form of $q(\vect{\theta}_l)$, this is achieved when $\hvect{\theta}_l$ is aligned with the mode of $\tilde{p}(\vect{\theta}_l)$.
By computing the expectation in \eqref{equ:PTheta} it can be shown that
\begin{equation}
\begin{split}
    \hvect{\theta}_l &= \argmax_{\vect{\theta}_l} \Big\{ \log p(\vect{\theta}_l) -\|\vect{r}_l-\widehat{w}_l\vect{s}(\vect{\theta}_l)\|_{\vect{\Lambda}}^2\\
			-&\sum_{k=1,k\neq l}2\Re\left\{ [\hvect{\Phi}]_{k,l}\vect{s}(\hvect{\theta}_k)^H\vect{\Lambda}\vect{s}(\vect{\theta}_l)\right\}	
			-[\hvect{\Phi}]_{l,l}\|\vect{s}(\vect{\theta}_l)\|_{\vect{\Lambda}}^2 \Big\},
\end{split}
	\label{equ:OptimTheta}
\end{equation}
where $\Re\left\{\cdot\right\}$ denotes the real part operator and
\begin{equation}
\begin{split}
    \vect{r}_l = \vect{y} -\sum_{k=1,k\neq l}^{L} \widehat{w}_k\vect{s}(\hvect{\theta}_k)
	\label{equ:InterferenceCancel}
\end{split}
\end{equation}
is a residual signal that cancels the contribution of the other $L-1$ components.
Solving \eqref{equ:OptimTheta} requires in general a numerical optimization.
Let us point out that the last two terms in \eqref{equ:OptimTheta} account for correlations between the elements of $\vect{w}$, acting as penalty factors in the estimator of $\vect{\theta}_l$.
Also, note that under the assumption $\mathrm{A1}$ \eqref{equ:OptimTheta} coincides with the estimation expression used in the VB-SAGE algorithm \cite{ShutinVBSAGE}.

Now, let us consider the  estimation of $q(\alpha_l)$.
The bound in \eqref{equ:LowerBoundVB} with respect to $q(\alpha_l)$ can be expressed as $\log p(\vect{y})\ge\EXP{\alpha_l} \log \tilde{p}(\alpha_l)/q(\alpha_l)$, where
\[
    \tilde{p}(\alpha_l)\propto \mathrm{exp}{\left(\EXP{\vect{w}}\log p(w_l|\alpha_l)p(\alpha_l)\right)}.
\]
%The latter expression depends only on $q(\vect{w})$.
It has been demonstrated in  \cite{ShutinVBSAGE} (for the assumption $\mathrm{A1}$) and in \cite{ShutinFastRVM} (for the assumption $\mathrm{A2}$) that the sequence of estimates $\left\{q^{[0]}(\alpha_l), q^{[1]}(\alpha_l), q^{[2]}(\alpha_l),\ldots\right\}$,  obtained by repeated maximization of the right-hand side of \eqref{equ:LowerBoundVB} with respect to the pdfs $q(w_l)$ and $q(\alpha_l)$ (for $\mathrm{A1}$), or $q(\vect{w})$ and $q(\alpha_l)$ (for $\mathrm{A2}$), converges to the pdf $q^{[\infty]}(\alpha_l)=\GammaPDF(\alpha_l|1,(\AlphaInf)^{-1})$ with
\begin{equation}
		\AlphaInf=\left\{		
		\begin{array}{ll}
		(|\mu_l|^2-\varsigma_l)^{-1}, & \frac{|\mu_l|^2}{\varsigma_l}>1 \\
		\infty, & \frac{|\mu_l|^2}{\varsigma_l}\le 1.
		\end{array}
		\right.
		\label{equ:FixedPoint}	
\end{equation}
The parameters $\varsigma_l$ and $\mu_l$ in \eqref{equ:FixedPoint} are computed as follows.
For the assumption $\mathrm{A1}$:
\begin{equation}
\begin{split}
    (\mathrm{A1}):\quad& \varsigma_l=1\big/\|\vect{s}(\hvect{\theta}_l)\|^{2}_{\vect{\Lambda}},\quad
                  \mu_l=\varsigma_l\vect{s}(\hvect{\theta}_l)^H\vect{\Lambda}\tvect{r}_l^{[\mathrm{A1}]},\\
            &\tvect{r}_l^{[\mathrm{A1}]} = \vect{y} -\sum_{k=1,k\neq l}^{L} \widehat{w}_k\vect{s}(\hvect{\theta}_k).
    \label{equ:TestParametersA1}
\end{split}
\end{equation}
For the assumption $\mathrm{A2}$, we first define
\begin{equation}
\begin{split}
    &\hvect{A}_{\notl}=\diag([\widehat{\alpha}_1,\ldots,\widehat{\alpha}_{l-1},\widehat{\alpha}_{l+1},\ldots,\widehat{\alpha}_{L}]), \\
    &\vect{S}(\hvect{\Theta}_{\notl})=[\vect{s}(\hvect{\theta}_1),\ldots,\vect{s}(\hvect{\theta}_{l-1}),\vect{s}(\hvect{\theta}_{l+1}),\ldots,\vect{s}(\hvect{\theta}_{L})],\\ &\hvect{\Phi}_{\notl}=\left(\vect{S}(\hvect{\Theta}_{\notl})^H\vect{\Lambda}\vect{S}(\hvect{\Theta}_{\notl})+\hvect{A}_{\notl}\right)^{-1},\\
  	&\hvect{w}_{\notl}=\hvect{\Phi}_{\notl}\vect{S}(\hvect{\Theta}_{\notl})^H\vect{\Lambda}\vect{y},\, \text{and}\quad \\
    &\tvect{r}_l^{[\mathrm{A2}]}=\vect{y}-\vect{S}(\hvect{\Theta}_{\notl})\hvect{w}_{\notl}.
    \label{equ:Substitutions}
\end{split}
\end{equation}
%\[
%     \tvect{r}_l^{[\mathrm{A2}]}=\vect{y}-\vect{S}(\hvect{\Theta}_{\notl})\hvect{w}_{\notl}.
%\]
Then, $\varsigma_l$ and $\mu_l$ for this assumption are evaluated as follows
%\begin{figure*}[!t]
%\normalsize
%%\setcounter{MYtempeqncnt}{\value{equation}}
%% Set the equation number to one less than the one
%% desired for the first equation here.
%% The value here will have to changed if equations
%% are added or removed prior to the place these
%% equations are referenced in the main text.
%%\setcounter{equation}{5}
%\begin{equation}
%\begin{split}
%    &\hvect{A}_{\notl}=\diag([\widehat{\alpha}_1,\ldots,\widehat{\alpha}_{l-1},\widehat{\alpha}_{l+1},\ldots,\widehat{\alpha}_{L}]),\quad
%    %\vect{S}(\hvect{\Theta}_{\notl})=[\vect{s}(\hvect{\theta}_1),\ldots,\vect{s}(\hvect{\theta}_{l-1}),\vect{s}(\hvect{\theta}_{l+1}),\ldots,\vect{s}(\hvect{\theta}_{L})],\\ %&\hvect{\Phi}_{\notl}=\left(\vect{S}(\hvect{\Theta}_{\notl})^H\vect{\Lambda}\vect{S}(\hvect{\Theta}_{\notl})+\hvect{A}_{\notl}\right)^{-1},\quad
%  	\hvect{w}_{\notl}=\hvect{\Phi}_{\notl}\vect{S}(\hvect{\Theta}_{\notl})^H\vect{\Lambda}\vect{y},\, \text{and}\quad
%    %&\tvect{r}_l^{[\mathrm{A2}]}=\vect{y}-\vect{S}(\hvect{\Theta}_{\notl})\hvect{w}_{\notl}.
%    \label{equ:Substitutions}
%\end{split}
%\end{equation}
%% Restore the current equation number.
%%\setcounter{equation}{\value{MYtempeqncnt}}
%% IEEE uses as a separator
%\hrulefill
%% The spacer can be tweaked to stop underfull vboxes.
%\vspace*{4pt}
%\end{figure*}

\begin{gather}
\begin{split}
	(\mathrm{A2}):
    \varsigma_l&=\Big(\vect{s}(\hvect{\theta}_l)^H\vect{\Lambda}\vect{s}(\hvect{\theta}_l)-\\
            &\vect{s}(\hvect{\theta}_l)^H\vect{\Lambda}\vect{S}(\hvect{\Theta}_{\notl})\hvect{\Phi}_{\notl}\vect{S}(\hvect{\Theta}_{\notl})^H\vect{\Lambda}\vect{s}(\hvect{\theta}_l)\Big)^{-1},\\ \mu_l=&\varsigma_l\vect{s}(\hvect{\theta}_l)^H\vect{\Lambda}\vect{y}-\\
            &\varsigma_l\vect{s}(\hvect{\theta}_l)^H\vect{\Lambda}\vect{S}(\hvect{\Theta}_{\notl})\hvect{\Phi}_{\notl}\vect{S}(\hvect{\Theta}_{\notl})^H\vect{\Lambda}\vect{y}\\
            =&\varsigma_l\vect{s}(\hvect{\theta}_l)^H\vect{\Lambda}(\vect{y}
            -\vect{S}(\hvect{\Theta}_{\notl})\hvect{w}_{\notl})\\
            =&\varsigma_l\vect{s}(\hvect{\theta}_l)^H\vect{\Lambda}\tvect{r}_l^{[\mathrm{A2}]}.
            %&\quad\tvect{r}_l^{[\mathrm{A2}]}=\vect{y}-\vect{S}(\hvect{\Theta}_{\notl})\hvect{w}_{\notl}.
            \label{equ:TestParametersA2}
\end{split}
\end{gather}
Let us point out that for both $\mathrm{A1}$ and $\mathrm{A2}$ cases, the weight $\mu_l$ is a projection of $\vect{s}(\hvect{\theta}_l)$ on the corresponding residual signal $\tvect{r}_l^{[\mathrm{A1}]}$ or $\tvect{r}_l^{[\mathrm{A2}]}$, respectively.
The latter are computed by canceling (subtracting) the contribution of the other $L-1$ components.
Note that $\tvect{r}_l^{[\mathrm{A1}]}$ coincides with \eqref{equ:InterferenceCancel}; also, $\tvect{r}_l^{[\mathrm{A2}]}$ and $\tvect{r}_l^{[\mathrm{A1}]}$ are equal when $\hvect{\Phi}_{\notl}$ is diagonal, i.e., for uncorrelated components.
This will be a valid assumption for components that are physically well separated, i.e., when $\vect{s}(\vect{\theta}_l)^H\vect{\Lambda}\vect{s}(\vect{\theta}_k)\approx 0$, $k\neq l$.
Thus, for uncorrelated components the IARD and the VB-SAGE algorithms will lead to the same estimation results. 
Also, when assumption $\mathcal{A1}$ is used with IARD, an instance of the VB-SAGE algorithm is obtained. 
Yet IARD does not require an introduction of any latent variables, as it was done in the VB-SAGE algorithm. 

Finally, we estimate $q(w_l)$ and $q(\vect{w})$.
For the assumption $\mathrm{A1}$ the parameters of $q(w_l)$ are computed as
%Given a finite estimate of $q(\alpha_l)$ and $\AlphaInf$ from \eqref{equ:FixedPoint}, the pdf $q(w_l)=\GaussPDF(\vect{w}|\widehat{w},\widehat{\Phi}_l)$ can then be computed as follows:
\begin{equation}	
\begin{split}
	%(\mathrm{A1}):\quad\quad
    \widehat{\Phi}_l&=\left(\|\vect{s}(\hvect{\theta}_l)\|_{\vect{\Lambda}}^2+\AlphaInf\right)^{-1},
    \widehat{w}_l=\widehat{\Phi}_l\vect{s}(\hvect{\theta}_l)^H\vect{\Lambda}\tvect{r}_l^{[\mathrm{A1}]}.
	\label{equ:PhiDefVBSAGE}
\end{split}
\end{equation}
Similarly, for the assumption $\mathrm{A2}$ we compute
\begin{equation}	
\begin{split}
    %(\mathrm{A2}):\quad\quad
    \hvect{\Phi}=&\left(\vect{S}(\hvect{\Theta})^H\vect{\Lambda}\vect{S}(\hvect{\Theta})+\diag(\hvect{\alpha})\right)^{-1},\\
    \hvect{w}=&\hvect{\Phi}\vect{S}(\hvect{\Theta})^H\vect{\Lambda}\vect{y},
   	\label{equ:PhiDef}
\end{split}
\end{equation}
where $\hvect{\alpha}=[\widehat{\alpha}_1,\ldots,\widehat{\alpha}_{l-1},\AlphaInf,\widehat{\alpha}_{l+1},\widehat{\alpha}_{L}]^T$.

The key advantages of such incremental component-wise estimation scheme are the expressions \eqref{equ:OptimTheta} and \eqref{equ:FixedPoint}.
The former permits a simpler numerical optimization of the dispersion parameters as the dimensionality of the resulting objective function equals to the dimensionality of $\vect{\theta}_l$, rather than that of $\vect{\Theta}$.
Result \eqref{equ:FixedPoint} gives a simple criterion for model order selection: when $|\mu_l|^2\le \varsigma_l$, we get $\AlphaInf=\infty$, i.e., $w_l\rightarrow 0$ and the component is removed.
This implements an automatic model order selection.
Moreover, the signal model can be constructed from bottom up, i.e., starting with an empty model $\vect{S}(\hvect{\Theta})\hvect{w}=\vect{0}$, and initializing the first component using ``incoherent'' initialization as described in the Algorithm \eqref{alg:AlgorithmInit}.%
\begin{algorithm}[H]
	\caption{Component initialization }
	\label{alg:AlgorithmInit}
	\begin{algorithmic}
	\STATE Compute $\vect{r}_l\gets \vect{y}-\vect{S}(\hvect{\Theta})\hvect{w}$; estimate $\vect{\theta}_l$ using
        \begin{equation}
        \begin{split}
        \hvect{\theta}_l &= \argmax_{\vect{\theta}_l} \Big\{ \log p(\vect{\theta}_l) -\|\vect{r}_l^H\vect{s}(\vect{\theta}_l)\|_{\vect{\Lambda}}^2\Big\},
        \end{split}
	   \label{equ:IncoherentOptimTheta}
    \end{equation}
	\STATE Compute $q(\alpha_l)$ using \eqref{equ:FixedPoint}
    \IF {$\AlphaInf$ is finite}
        \STATE Compute $q(w_l)$ using \eqref{equ:PhiDefVBSAGE} (or $q(\vect{w})$ using \eqref{equ:PhiDef})
    \ELSE
        \STATE Discard the component and abort initialization
    \ENDIF
	\end{algorithmic}
\end{algorithm}
\noindent
If during the initialization the test \eqref{equ:FixedPoint} results in a finite sparsity parameter $\AlphaInf$, a new component is accepted in the model.
The parameters of the components are then updated following the Algorithm \ref{alg:Update}.
\begin{algorithm}[H]
	\caption{Parameter update}
	\label{alg:Update}
	\begin{algorithmic}
    \WHILE{ Not converged }
    \FOR{$l \in \{1,\ldots,L\}$}
    	\STATE Update $q(\vect{\theta}_l)$ from \eqref{equ:OptimTheta} and $q(\alpha_l)$ using \eqref{equ:FixedPoint}
		\IF {$\AlphaInf$ is finite}
            \STATE Update $q(w_l)$ using \eqref{equ:PhiDefVBSAGE} (or $q(\vect{w})$ using \eqref{equ:PhiDef})
		\ELSE
            \STATE Remove the $l$th component from the model
        \ENDIF
        \ENDFOR
    \ENDWHILE
	\end{algorithmic}
\end{algorithm}
After update, the initialization can be repeated again for an updated residual signal.
The algorithm is interrupted when no new components can be added to the model.
Let us also mention at this stage that $\hvect{\Phi}_{\notl}$ can be efficiently computed using rank-one updates (see \cite{ShutinVBSAGE} for more details). 
Thus, $q(\vect{w})$ can be efficiently updated even for large $L$. 

The condition $|\mu_l|^2>\varsigma_l$ in \eqref{equ:FixedPoint} we term a pruning condition since it determines if $\AlphaInf$ is finite.
It forms a basis for a multipath component detector.
In fact, the sparsity of the estimated model is governed by this condition.
To better understand its properties and limitations we consider this condition in more details in the following section.

%%%%%%%%%%%%%%%%%%%%%%%%%%%%%%%%%%%%%%%%%%%%%%%%%%%%%%%%%%%%%%
\section{Analysis of the pruning condition}\label{sec:AdjustedPruning}

Let us now investigate this pruning condition in greater detail for both $\mathrm{A1}$ and $\mathrm{A2}$ assumptions.
To this end we define $\rho_l=|\mu_l|^2/\varsigma_l$.
A closer look at \eqref{equ:TestParametersA1} and \eqref{equ:TestParametersA2} reveals that the parameters $\mu_l$ and $\varsigma_l$ correspond, respectively, to the posterior estimate of the $l$th path weight $w_l$ and its variance when $\widehat{\alpha}_l=0$.
Thus, we can interpret $\rho_l$ as an estimate of the $l$th component SNR after the processing.\footnote{This can also be interpreted as the component SNR after a matched filter processing, with $\vect{s}(\hvect{\theta}_l)$ playing the role of a matched filter.}
Specifically, the pruning condition
\begin{equation}
	\rho_l>1,
	\label{equ:SNRPrune}
\end{equation}
states that an estimate of the approximating pdf $q(\alpha_l)$ has a finite mean if, and only if, an estimate of the $l$th component SNR after subtracting the interference of the other $L-1$ components exceeds $1$ (or equivalently $0$~dB).

Yet in many practical applications a $0$~dB threshold might not represent the desired level of confidence in the estimated component. Moreover, we have empirically observed the condition \eqref{equ:SNRPrune} generally overestimates the model order: some of the detected components were falsely introduced into the model, with the estimated weights having small, yet non-zero weights and the corresponding parameters $\rho_l$  exceeding a $0$~dB threshold.
Empirical adjustment of the threshold to some level $\kappa_l\ge 1$ improves the model order estimate \cite{ShutinBuchgraber2012, ShutinVBSAGE, ShutinFastRVM, Shutin_IARD}.
In what follows we explain why signal sparsity is overestimated with the condition \eqref{equ:SNRPrune} and how to select the threshold $\kappa_l$ such that the conditions $\rho_l>\kappa_l$ is more robust against estimation artifacts.
For this purpose we will explore a connection between the statistical structure of \eqref{equ:SNRPrune} and hypothesis testing.

%Let us assume that the algorithm iterations have converged.
Consider a single component $l$, and assume that the parameters of the other components are fixed.
Define now two hypotheses $H_0$ and $H_1$ for the ``true'' weight $w_l$ of the $l$th multipath component as follows:
\begin{equation}
	\left\{
		\begin{array}{ll}
			H_0:& w_l=0 \\
	   		H_1:& w_l\neq 0.
		\end{array}
	\right.
	\label{equ:HypothesisTesting}
\end{equation}
Our goal here is to understand how statistics of $\rho_l$ can be utilized to choose between these two hypotheses in the Neyman-Pearson sense.
To this end we will consider the distribution of $\rho_l$ under $H_0$ and $H_1$ hypotheses for both $\mathrm{A1}$ and $\mathrm{A2}$ assumptions.

%%%%%%%%%%%%%%%%%%%%%%%%%%%%%%%%%%%%%%%%%%%%%%%%%%%%%%%%%%%%%%%%%%%%%%%%%%%%%%%%%
\subsection{Assumption $\mathrm{A1}$: independent multipath components}
We will begin our analysis with the following proposition:
\begin{prop}\label{thm:1}
Assume that $\hvect{\theta}_l$ is found from \eqref{equ:OptimTheta} and that other factors in \eqref{equ:FactorizationVBSAGE} are fixed.
Then, under hypothesis $H_0$ the statistic $\rho_l$ will follow an extreme value distribution \cite{FisherTippet1928} with the following pdf:
\begin{equation}
    p(\rho_l|H_0)=\left\{
        \begin{array}{lc}
        (\mathrm{e}^{-N/\mathrm{e}})\delta(\rho_l)& 0\le\rho_l\le 1\\
        (1-\mathrm{e}^{-N/\mathrm{e}})\tilde{p}(\rho_l|H_0) & \rho_l> 1
        \end{array}
        \right.
	\label{equ:Theorem1}
\end{equation}
where $\delta(\rho_l)$ is a Dirac delta distribution and
\begin{equation}
\begin{split}
    \tilde{p}(\rho_l|H_0)=\mathrm{e}^{\left(-\rho_l+\log(N)-\mathrm{e}^{-\rho_l+\log(N)}\right)},\quad \rho_l\ge0,
	\label{equ:Theorem1Max}
\end{split}
\end{equation}
is a pdf of the Gumbel distribution \cite{Gumbel1954}.
\end{prop}

\begin{proof}
Consider the distribution of $\rho_l$ under the hypothesis $H_0$ for some arbitrary value of $\vect{\theta}_l$ and known noise statistics.
Due to the efficiency of maximum likelihood estimators for linear models \cite{Kay1993}, it is straightforward to show that $\mu_l\thicksim \mathrm{CN}(\mu_l|0,\varsigma_l)$.
Recall now that $\rho_l=|\mu_l|^2/\varsigma_l$.
It is known that the square of a normally distributed zero mean random variable normalized by its variance will follow a $\chi^2$ distribution.
Since the variance of real and imaginary parts of $\mu_l$ is $\varsigma_l/2$, then $\rho_l$ will follow a scaled\footnote{The scaling factor in this case is $1/2$ to compensate for the reduced variance of real and imaginary parts.} $\chi^2$ distribution with two degrees of freedom.
In our case it is an exponential distribution with the pdf
\begin{equation}
    p(\rho_l)=\mathrm{e}^{-\rho_l},\quad \rho_l\ge0.
    \label{equ:H0Independent}
\end{equation}
This distribution arises when for a fixed $\vect{\theta}_l$ different realizations of the residual signal $\tvect{r}_l^{[\mathrm{A1}]}$ are generated.
Alternatively, $\tvect{r}_l^{[\mathrm{A1}]}$ can be fixed and $\vect{\theta}_l$ then drawn at random.
Note that under $H_0$ the residual $\tvect{r}_l^{[\mathrm{A1}]}$ is a realization of an $N$-dimensional Gaussian noise vector.
However, due to maximization \eqref{equ:OptimTheta} we select the ``best'' dispersion parameter $\hvect{\theta}_l$ out of $N$ independent possibilities.\footnote{Note that possible correlations in the residual signal due to diffuse multipath are ``whitened'' by the matrix $\vect{\Lambda}^{-1}$.}
As a results an observed value of $\rho_l$ under $H_0$ will follow the distribution of a maximum out of $N$ values drawn from \eqref{equ:H0Independent}.
Such type of distributions are known as extreme value distributions.

To derive the distribution function $F_{\mathrm{max}}(\rho_l)$ of the corresponding extreme value distribution, we apply the Fisher-Tippett-Gnedenko theorem \cite{FisherTippet1928} to the distribution function $F(\rho_l)=1-\mathrm{e}^{-\rho_l}$ of the exponential pdf \eqref{equ:H0Independent}.
By the theorem, $F_{\mathrm{max}}(\rho_l)$ can be computed as the limit of appropriately shifted and scaled variable $\rho_l$: $F_{\mathrm{max}}(\rho_l)=\lim_{n\rightarrow \infty} \left(F(\frac{\rho_l-b_n}{a_n})\right)^n$ for some real sequences $a_n>0$ and $b_n>0$ that are independent of $\rho_l$.
In our case, it can be demonstrated that for $a_n=1$ and $b_n=log(n)$, the maximum of out of $N$ exponentially distributed values will follow a Gumbel distribution \cite{Gumbel1954} $F_{\mathrm{max}}(\rho_l)$ with the distribution function
\[
    F_{\mathrm{max}}(\rho_l)=\exp\left(-\mathrm{e}^{-(\rho_l-\log(N))}\right)
\]
and the corresponding pdf
\begin{equation}
    \tilde{p}(\rho_l|H_0)=\mathrm{e}^{\left(-\rho_l+\log(N)-\mathrm{e}^{-\rho_l+\log(N)}\right)},\quad \rho_l\ge0.
    \label{equ:H0Gumbel}
\end{equation}
Note, however, that for $\rho_l\le 1$ the sparsity parameter $\AlphaInf=\infty$.
In this case the hypothesis $H_0$ is automatically accepted.
Taking this into consideration, the pdf $p(\rho_l|H_0)$ can be specified as
\begin{equation}
    p(\rho_l|H_0)=\left\{
        \begin{array}{lc}
        F_{\mathrm{max}}(1)\delta(\rho_l)& 0\le\rho_l\le 1\\
        (1-F_{\mathrm{max}}(1))\tilde{p}(\rho_l|H_0) & \rho_l> 1,
        \end{array}
        \right.
\end{equation}
which completes the proof.
\end{proof}

The next proposition defines the distribution of $\rho_l$ under hypothesis $H_1$.
\begin{prop}\label{thm:1}
Under hypothesis $H_1$ the statistic $\rho_l$ will follow a scaled non-central chi-square distribution
\begin{equation}
    p(\rho_l|H_1)=\left\{
        \begin{array}{lc}
        0& 0\le\rho_l\le 1\\
       \frac{1}{Z} \tilde{p}(\rho_l|H_0) & \rho_l> 1
        \end{array}
        \right.
	\label{equ:Theorem2}
\end{equation}
where
\begin{equation}
\begin{split}
    \tilde{p}_{H_1}(\rho_l)=&\mathrm{e}^{-\left(\rho_l+\frac{\eta_l}{2}\right)}
    \mathrm{I}_{0}\left(\sqrt{2\eta_l\rho_l}\right),\\
    \text{and}\quad  Z=&\int_{1}^{\infty}\tilde{p}_{H_1}(\rho_l)d\rho_l.
	\label{equ:Theorem2Max}
\end{split}
\end{equation}
\end{prop}
\begin{proof}
The distribution of $\rho_l$ under hypothesis $H_1$ can be studied in a similar fashion.
The weight $\mu_l$ will follow a Gaussian distribution with the true (unknown) mean $w_l\neq 0$ and a variance $\varsigma_l$.
Following the same line of arguments as for the $H_0$ case, it can be shown that $\rho_l$ will follow a scaled non-central chi-square distribution ${\chi'}^{2}_2(\eta_l)$ with two degree of freedom and a non-centrality parameter $\eta_l=2|w_l|^2/\varsigma_l$:
\begin{equation}
    \tilde{p}(\rho_l|H_1)=\mathrm{e}^{-\left(\rho_l+\frac{\eta_l}{2}\right)}
    \mathrm{I}_{0}\left(\sqrt{2\eta_l\rho_l}\right),
    \quad \rho_l>0,
    \label{equ:H1Independent}
\end{equation}
where $\mathrm{I}_{0}(x)$ is a modified Bessel function of the first kind.
Since for $\rho_l\le 1$ the $H_1$ hypothesis is automatically rejected, the support of $p(\rho_l|H_1)$ is restricted to the interval $(1,\infty)$.
Taking this into account leads to result \eqref{equ:Theorem2}, which finalizes the proof.
\end{proof}
Let us note that, strictly speaking, \eqref{equ:Theorem2} will hold for components with a sufficiently high ``true'' SNR $|w_l|^2/\varsigma_l$.
In high SNR regime optimization \eqref{equ:OptimTheta} will consistently result in the same value of $\hvect{\theta}_l$.
Yet as $w_l$ decreases, the corresponding residual signal $\tvect{r}_l^{[\mathrm{A1}]}$ becomes dominated by the additive noise $\vect{\xi}$ and a mixture of \eqref{equ:Theorem2} and \eqref{equ:H0Gumbel} will be observed.
%Thus, in low SNR regime $p_{H_1}(\rho_l)$ will deviate from \eqref{equ:H1Independent} and as $w_l\rightarrow 0$, $p_{H_1}(\rho_l)$ will converge to \eqref{equ:Theorem1}.

Now, we can select between $H_0$ and $H_1$ using the following test function $T\left(\rho_l\right)$:
\begin{equation}
\begin{split}
	T\left(\rho_l\right)&=\left\{
		\begin{array}{ll}
			0, & \rho_l\le \kappa_l \\
	   		1, & \rho_l>\kappa_l,
		\end{array}\right.,\\
	\emph{\text{s.t.}}~~\kappa_l>0&,\quad \E_{p(\rho_l|H_0)}\left\{T(\rho_l)\right\}=\epsilon_l,
	\label{equ:TestFunction}
\end{split}
\end{equation}
where $\epsilon_l$ is the size of the test.
Let us now indicate some important properties of $T\left(\rho_l\right)$.
\begin{enumerate}
    \item The test \eqref{equ:TestFunction} is uniformly most powerful (UMP) test of size $\epsilon_l$ to choose between $H_0$ and $H_1$ specified by pdfs \eqref{equ:Theorem1} and \eqref{equ:Theorem2}, respectively.
        This follows from the fact that the rejection region of the test function $T(\rho_l)$, given by the interval $[\log(1/\log(1-\epsilon_l)^{-1/N}),\, \infty]$ is independent of $\eta_l$ \cite{PoorDetectionEstimation}.
    \item Under assumption $\mathrm{A1}$ the standard IARD algorithm implements the test \eqref{equ:TestFunction} with $\kappa_l=1$, as seen from \eqref{equ:FixedPoint}.
    \item Since for $\rho_l\le 1$ the corresponding component is automatically removed, the size $\epsilon_l$ of the test \eqref{equ:TestFunction} must be upper bounded. The upper bound is given by by $(1-F_{\mathrm{max}}(1))$.
\end{enumerate}
It is important to stress that for a standard threshold $\kappa_l=1$, the size of the hypothesis test $\epsilon_l$ will be quite large for typical values of $N$ (see Fig. \ref{fig:HypothesisSize}).
\begin{figure}[!htb]
\includegraphics[width=0.95\linewidth]{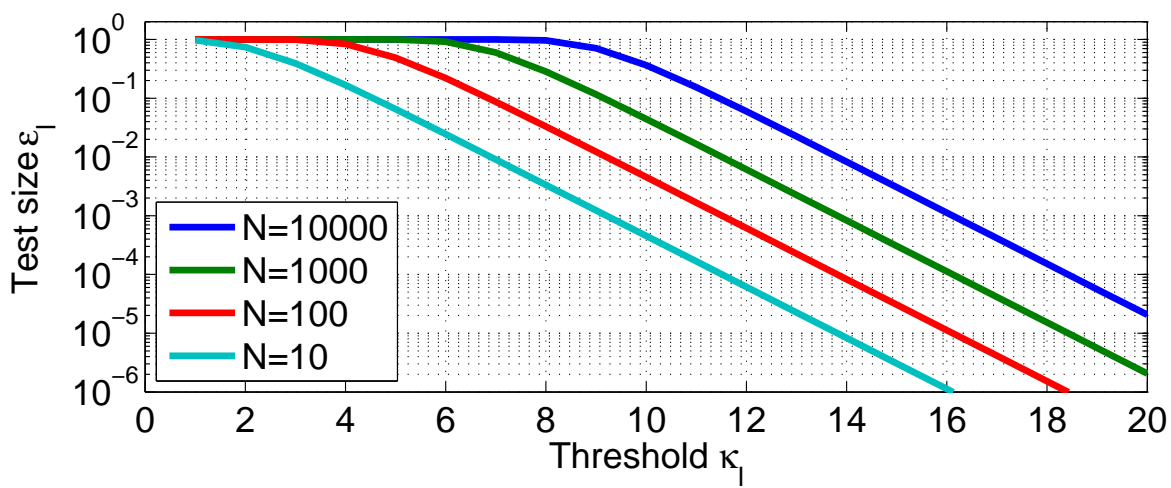}
\caption{Size $\epsilon_l$ of the test \eqref{equ:TestFunction} versus threshold $\kappa_l$ for different values of $N$.}
\label{fig:HypothesisSize}
\end{figure}
In other words the standard IARD will implement the test \eqref{equ:TestFunction} with a very high probability of false alarm.
As a result, $H_0$ will be falsely accepted more often, leading to estimation artifacts.
Moreover, as the number of samples $N$ increases, the probability of generating artifacts grows as well, making it more difficult to distinguish ``true'' components from noise.
The reason for this is the optimization \eqref{equ:OptimTheta}, which leads to the emergence of the extreme value distribution \eqref{equ:Theorem1}.
As $N$ increases, this distribution shifts further away from the standard threshold $\kappa_l=1$, making the correct rejection of artifacts less probable.
%The reason for this behaviour is optimization \eqref{equ:OptimTheta}, which ``blindly'' selects the best component parameters $\hvect{\theta}_l$.
Naturally, by increasing the threshold $\kappa_l$ we can control the probability of false detection at some desired level $\epsilon_l$.
%The plot in Fig. \ref{fig:HypothesisSize} demonstrates the dependency between $\epsilon_l$ and the corresponding threshold $\kappa_l$ for different values of $N$.

%%%%%%%%%%%%%%%%%%%%%%%%%%%%%%%%%%%%%%%%%%%%%%%%%%%%%%%%%%%%%%%%%%%%%%%%%%%%%%%%%
\subsection{Assumption $\mathrm{A2}$: correlated multipath components}
Under the assumption $\mathrm{A2}$ the pruning condition \eqref{equ:SNRPrune} has a similar interpretation.
However, due to the correlations between the elements of $\vect{w}$, the corresponding analysis becomes significantly more involved.
Let us begin by considering the marginal posterior of $w_l$ for the case when $\alpha_l=0$.
This is again a Gaussian pdf with the mean $\mu_l$ and the variance $\varsigma_l$ given by \eqref{equ:TestParametersA2}.
Consider now the expectation $\E\{\mu_l\}=\E\{\varsigma_l\vect{s}_l(\hvect{\theta}_l)^H\vect{\Lambda}\tvect{r}_l^{[A2]}\}$ in \eqref{equ:TestParametersA2}.
It can be shown that
\begin{equation}
\begin{split}
    \E\{\mu_l\}=&\varsigma_l\vect{s}_l(\vect{\theta}_l)^H\left(\vect{\Lambda}^{-1}+\vect{S}(\hvect {\Theta}_{\notl})\hvect{A}_{\notl}^{-1}\vect{S}(\hvect {\Theta}_{\notl})^H\right)^{-1}\\
    \times&\vect{S}(\hvect {\Theta}_{\notl})\vect{w}_{\notl} + w_l.
              %{\hvect{s}_l^H\left(\hvect{\Lambda}^{-1}+\hvect{S}_{\notl}\hvect{A}_{\notl}^{-1}\hvect{S}_{\notl}^H\right)^{-1}\hvect{s}_l}\vect{w}_{\notl}
    \label{equ:MeanCorrelatedComp}
\end{split}
\end{equation}
where we re-used definitions \eqref{equ:Substitutions} to simplify notation.
By inspecting \eqref{equ:MeanCorrelatedComp} we see that the bias $\E\{\mu_l\}$ does not vanish under hypothesis $H_0$, i.e., when $w_l=0$.
Due to the correlations between the components, this bias is proportional to the ``true'' weights $\vect{w}_{\notl}$, which are generally unknown.
In other words, in order to decide between $H_0$ and $H_1$ within the incremental estimation apporach, i.e., for a particular component $l$, we need to known the weights of the other multipath component.
This in general prohibits a computation of the pdf $p_{H_0}(\rho_l)$ or $p_{H_0}(\rho_l)$ for the case $\mathrm{A2}$ unless some assumptions about the true weights $\vect{w}_{\notl}$ can be made.

Nonetheless, our simulations show that the test \eqref{equ:TestFunction} applied to the case $\mathrm{A2}$ performs quite well.

\section{Simulations results}

In the following we will investigate the performance of the proposed joint estimator and component detector for synthetic channels.
% as well as for measured channel responses.

\subsection{One component in noise}
We will begin with a single synthetic multipath component in white noise, i.e., $L=1$.
For that we generate a channel response according to \eqref{equ:ModelComMatrix} with the following assumptions.
We restrict the set of dispersion parameters to a single delay $\tau$, so that $\vect{s}(\vect{\theta})\equiv\vect{s}(\tau)$.
The vector $\vect{s}(\tau)$ is constructed as $\vect{s}(\tau)=\big[s[-\tau/T_s],\ldots,s[(N-1)-\tau/T_s]\big]^T$, where $N=128$ and $T_s=1$s.
The signal $s[n]$ is an OFDM signal with $K$ subcarriers located at discrete frequencies $2\pi k/K$, $k=0,\ldots,K-1$.
Each subcarrier is generated with a constant unit magnitude and random phase uniformly drawn from the interval $[0,2\pi]$.
The delay $\tau$ of the synthetic component is set to $\tau=0$.
The weight $w$ has a unit magnitude and a random phase drawn from the interval $[0,2\pi]$.

Our goal in this experiment is to validate the derived distributions of the decision statistic $\rho_l$ for both $H_0$ and $H_1$ hypothesis.
To this end we restrict the values of estimated component delays to the sampling instances.
The estimation algorithm is then initialized with only $2$ components: one with the delay set to the true delay $\tau$ to approximate the $H_1$ hypothesis, and the other one set to the neighboring sampling instance to approximate the hypothesis $H_0$.
To collect the corresponding statistics, we run the algorithm and collect the values of $\rho_1$ and $\rho_2$ over  $10000$ independent runs of the algorithm.
The obtained empirical distributions of both statistics are then compared to the derived theoretical distributions $p(\rho_l|H_0)$ and $p(\rho_l|H_1)$.
For both components a pruning threshold of $\kappa_1=\kappa_2=1$ is used, which corresponds to the standard IARD pruning condition.
The analysis is performed for different input SNRs that we compute as $10\log_{10}\|w_1\vect{s}(\tau_1)\|^2/\|\vect{\xi}\|^2+10\log_{10}(N)$; here, $10\log_{10}(N)$ is the processing gain of the estimator.

We begin our tests for the assumption $\mathrm{A1}$.
For that we use $K=N$, which corresponds to the correlation coefficient of $0.007$ between the components with delays located at two neighboring sampling instances.
In Fig. \ref{fig:EmpiricalVSTheory} we plot the resulting distributions for $9$dB, $13$dB, $17$dB, and $21$dB SNR.
As we see, there is a very good fit between the empirical and theoretical distributions under the $H_0$ hypothesis.
\begin{figure}[!htb]
\psfrag{ph0rho}[bl][bl][0.9][0]{$p(\rho_l|H_0)$}
\psfrag{ph1rho}[bl][bl][0.9][0]{$p(\rho_l|H_1)$}
\psfrag{EmpiricalpH0rho}[bl][bl][0.9][0]{Est. $p(\rho_l|H_0)$}
\psfrag{EmpiricalpH1rho}[bl][bl][0.9][0]{Est. $p(\rho_l|H_1)$}

\centerline{
\subfigure[\label{fig:EmpiricalVSTheory-12dB}]{\includegraphics[width=0.5\linewidth]{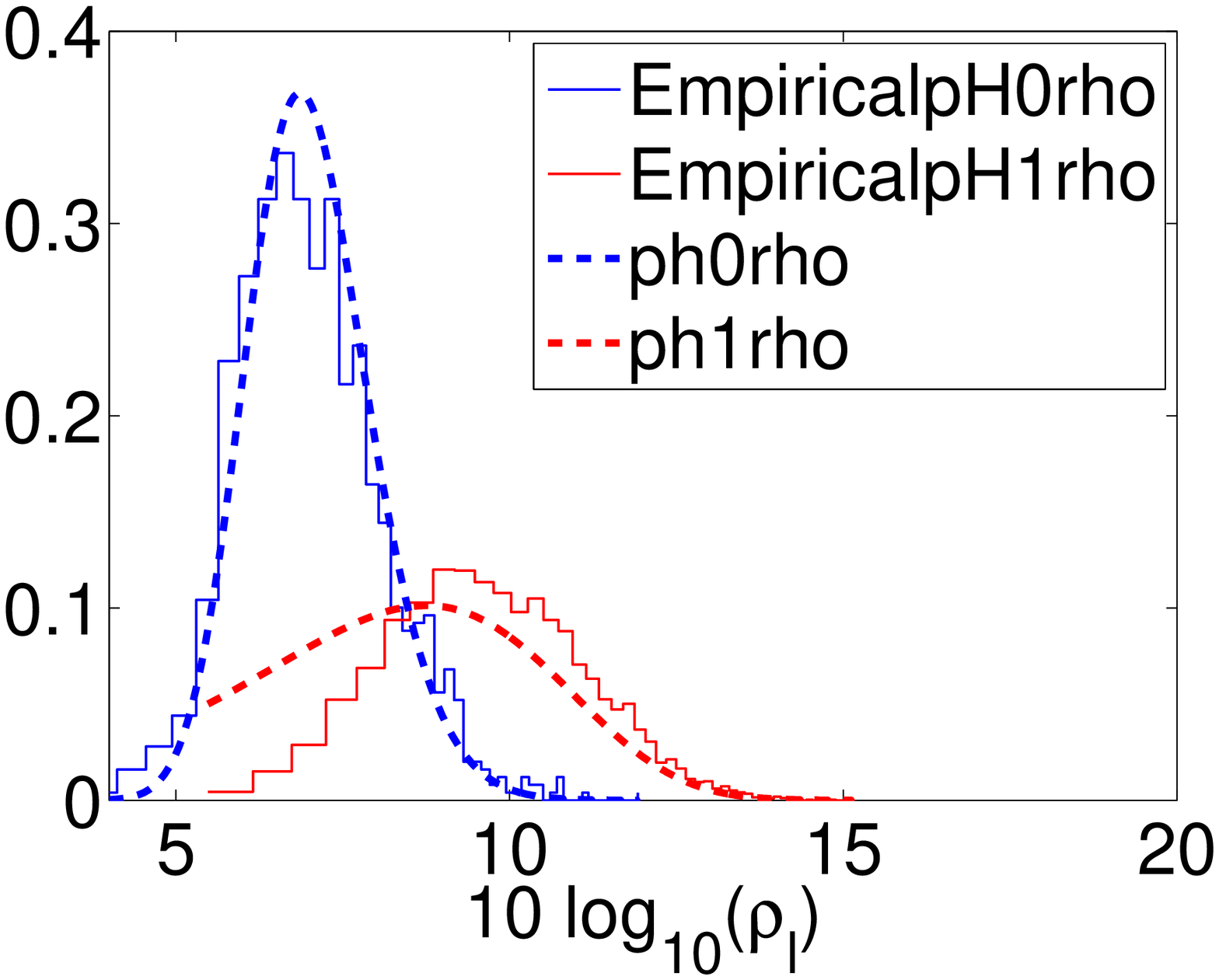}}
\subfigure[\label{fig:EmpiricalVSTheory-8dB}]{\includegraphics[width=0.5\linewidth]{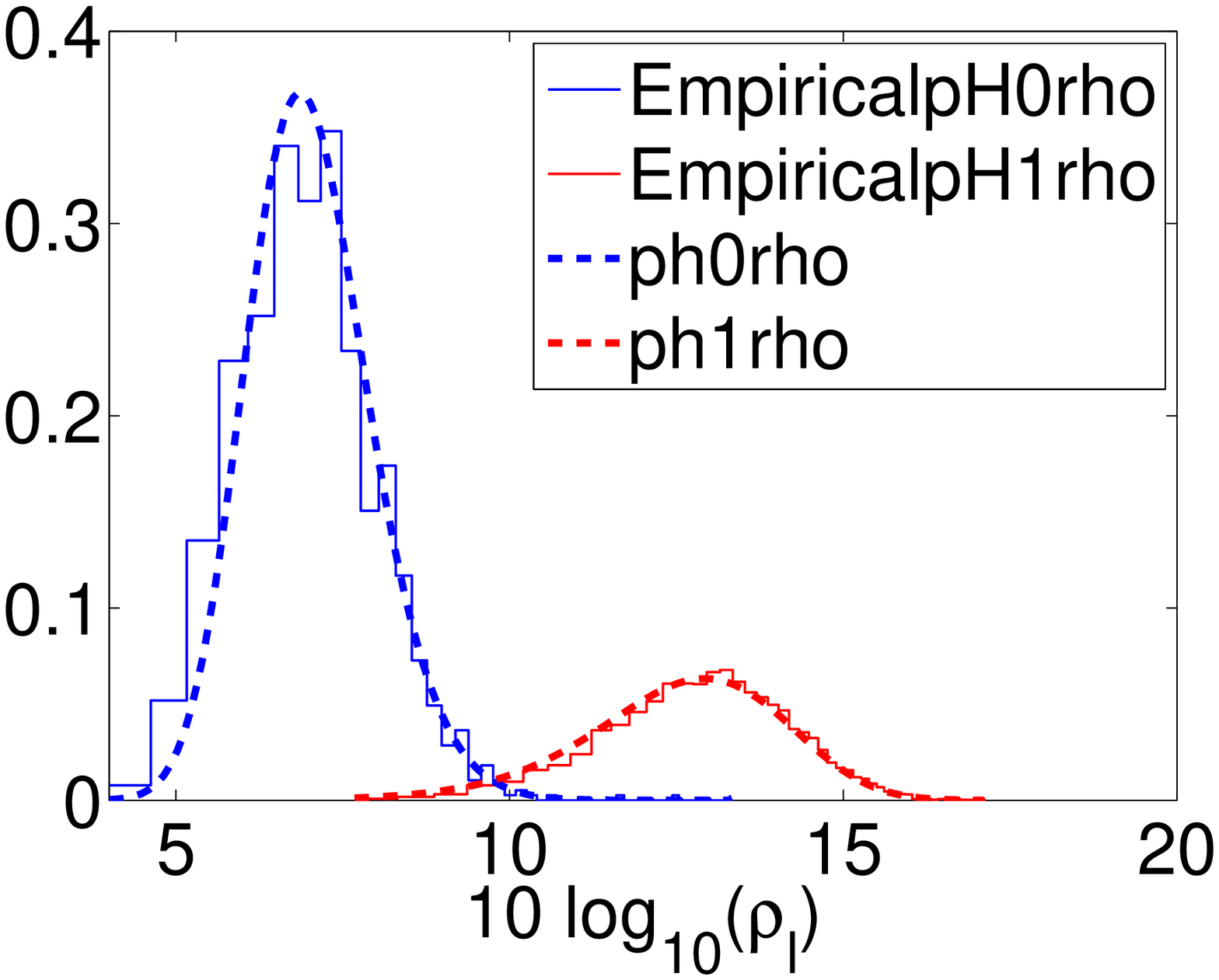}}
}
\centerline{
\subfigure[\label{fig:EmpiricalVSTheory-4dB}]{\includegraphics[width=0.5\linewidth]{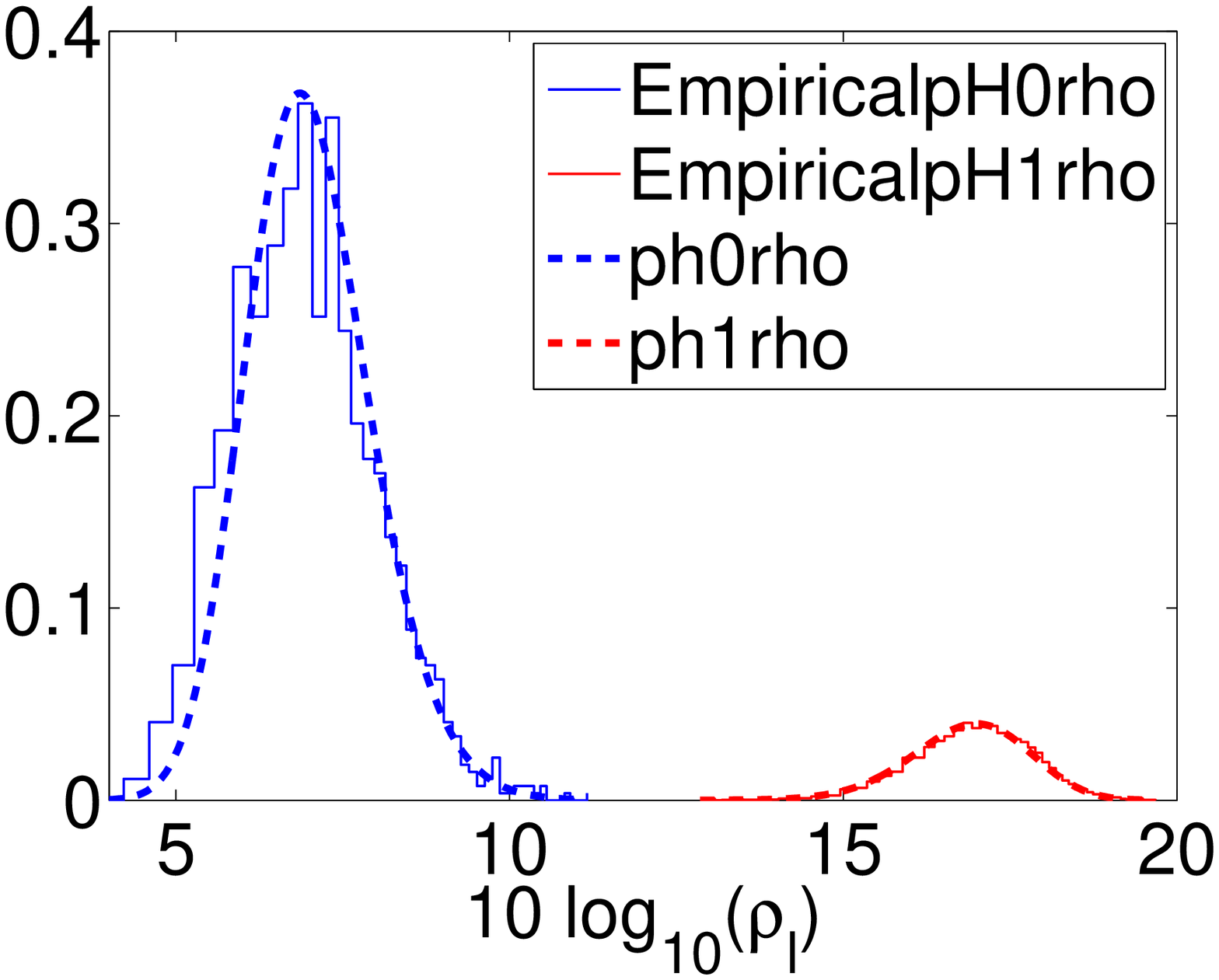}}
\subfigure[\label{fig:EmpiricalVSTheory0dB}]{\includegraphics[width=0.5\linewidth]{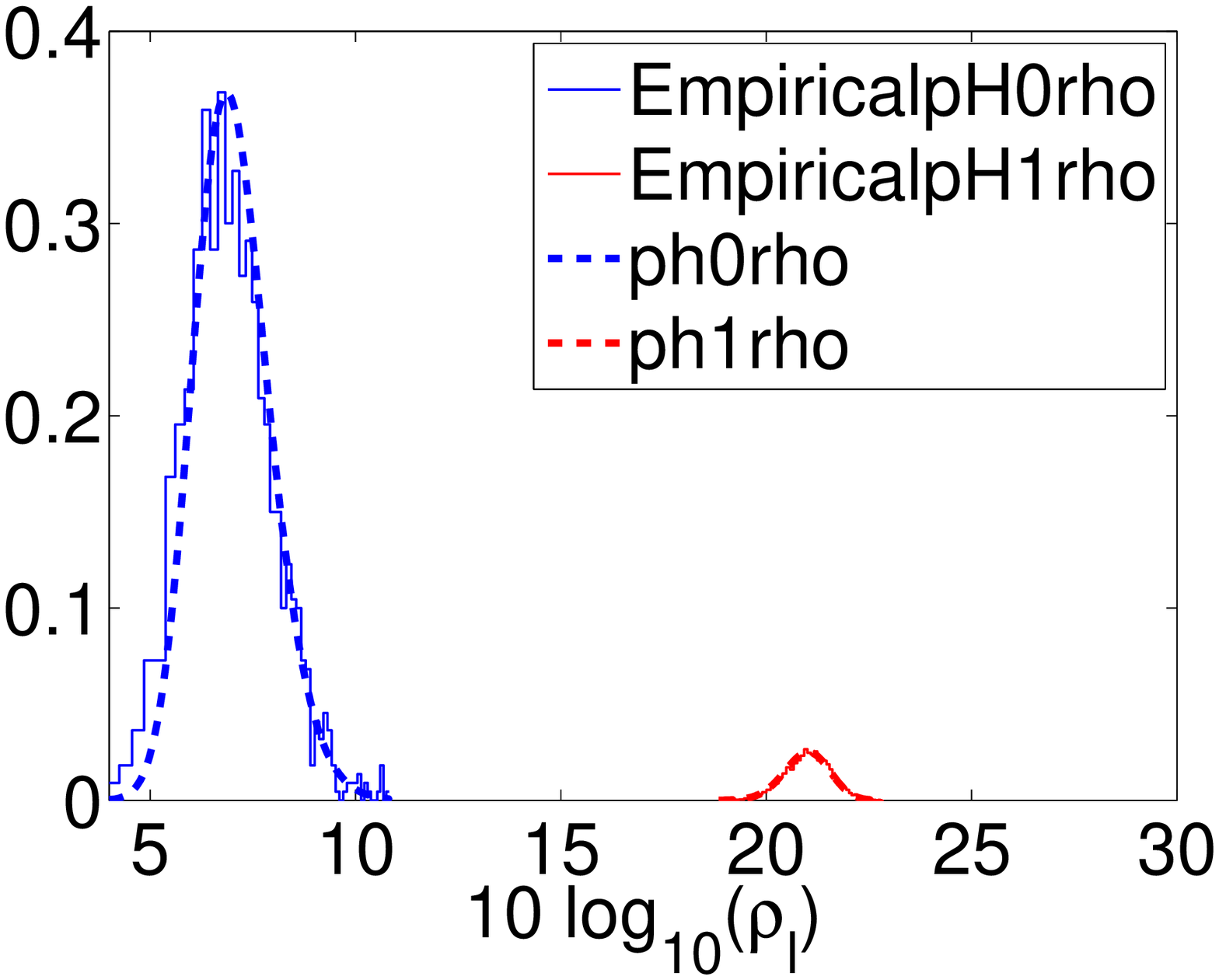}}
}

\caption{Comparison of empirical and derived distributions of the test statistic $\rho_l$ under $H_0$ and $H_1$ hypotheses for a) $\mathrm{SNR}=9$dB, b)$\mathrm{SNR}=13$dB, c)$\mathrm{SNR}=17$dB, and d) $\mathrm{SNR}=21$dB.}
\label{fig:EmpiricalVSTheory}
\end{figure}
Also, as expected, for low SNR the derived pdf $p(\rho_l|H_1)$ deviates slightly from the observed empirical distribution.

Now, let us consider the same scenario, yet for correlated components.
To increase the correlation between the components we select $K=N/2$, $K=N/4$, $K=N/8$, and $K=N/16$, which is equivalent to keeping the sampling rate fixed while reducing the bandwidth of the signals.
This leads to increased correlation between closely spaced components.
The correlation coefficients between two signals located at two consecutive delays for the above chosen values of $K$ are $0.62$, $0.89$, $0.97$, and $0.99$, respectively.

In Fig. \ref{fig:EmpiricalVSTheoryCorrelated} we show the empirical distributions of the decision statistic for $17$dB SNR and the corresponding pdfs $p(\rho_l|H_0)$ and $p(\rho_l|H_1)$.
Note that the latter are computed under the assumption $A1$.
\begin{figure}[!htb]
%\psfrag{ph0rho}[bl][bl][1.0][0]{$p_{H_0}(\rho_l)$}
%\psfrag{ph1rho}[bl][bl][1.0][0]{$p_{H_1}(\rho_l)$}
%\psfrag{EmpiricalpH0rho}[bl][bl][1.0][0]{Est. $p_{H_0}(\rho_l)$}
%psfrag{EmpiricalpH1rho}[bl][bl][1.0][0]{Est. $p_{H_1}(\rho_l)$}
\psfrag{ph0rho}[bl][bl][0.9][0]{$p(\rho_l|H_0)$}
\psfrag{ph1rho}[bl][bl][0.9][0]{$p(\rho_l|H_1)$}
\psfrag{EmpiricalpH0rho}[bl][bl][0.9][0]{Est. $p(\rho_l|H_0)$}
\psfrag{EmpiricalpH1rho}[bl][bl][0.9][0]{Est. $p(\rho_l|H_1)$}

\centerline{
\subfigure[\label{fig:EmpiricalVSTheoryCorr1}]{\includegraphics[width=0.5\linewidth]{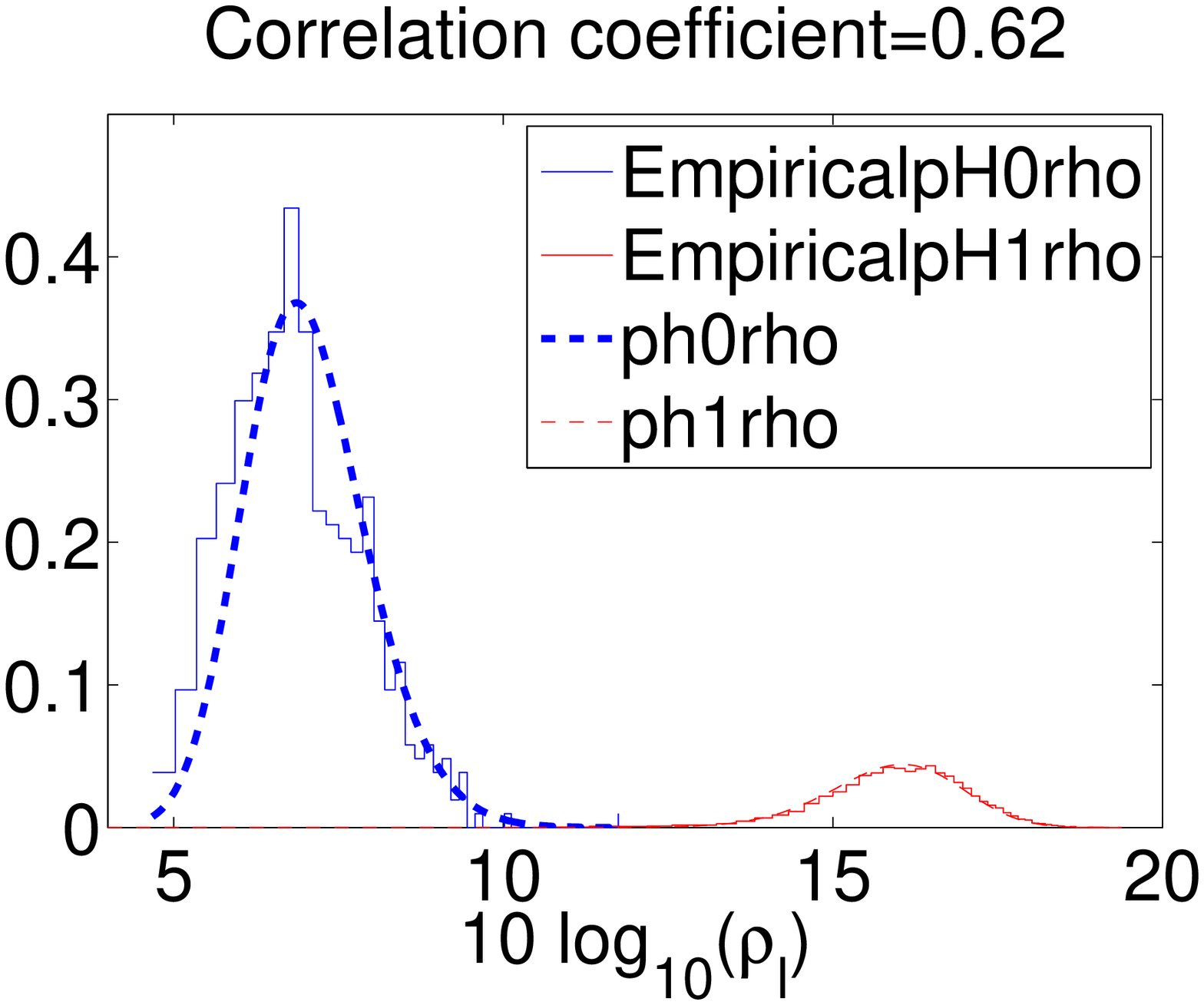}}
\subfigure[\label{fig:EmpiricalVSTheoryCorr2}]{\includegraphics[width=0.5\linewidth]{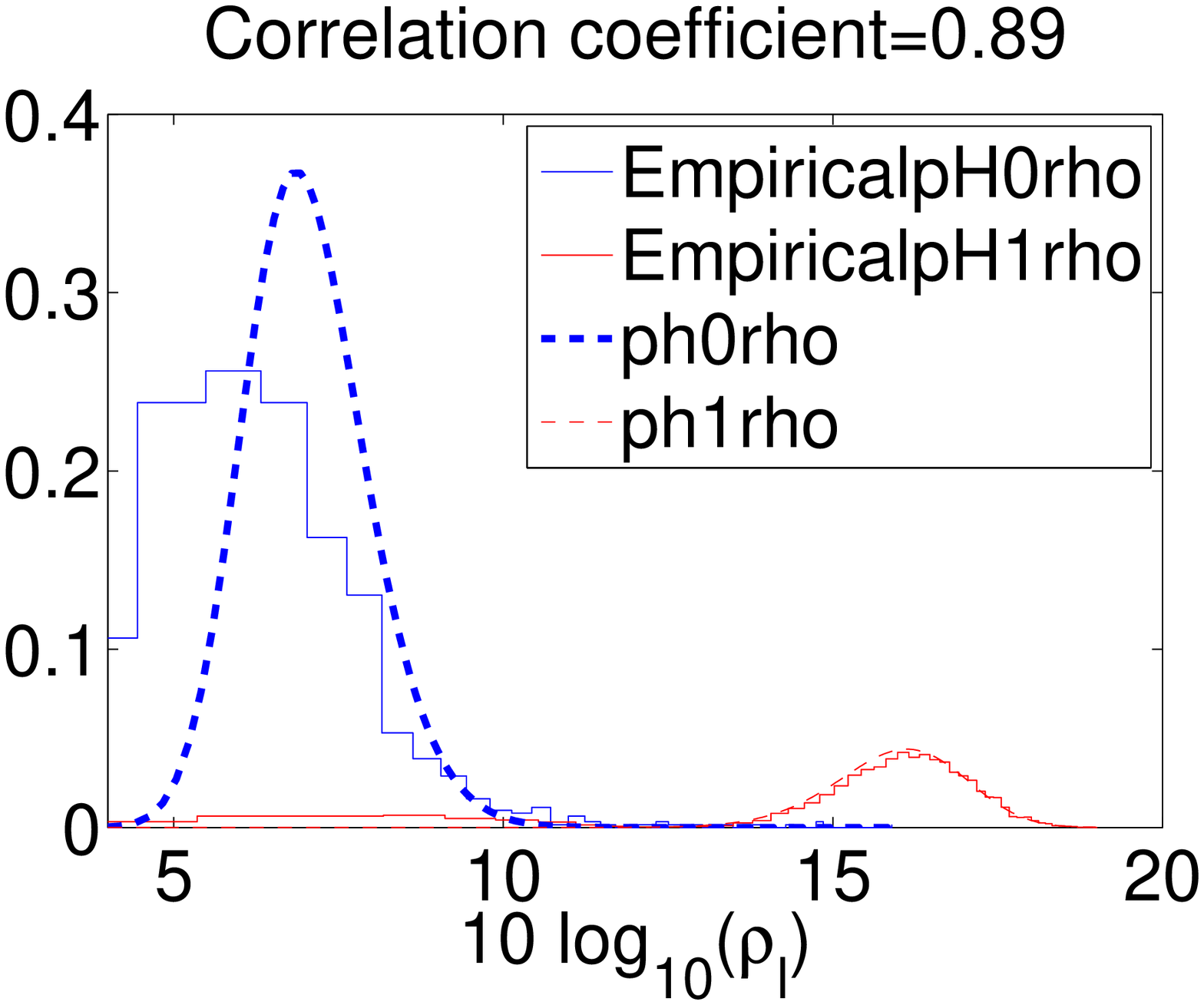}}
}
\centerline{
\subfigure[\label{fig:EmpiricalVSTheoryCorr3}]{\includegraphics[width=0.5\linewidth]{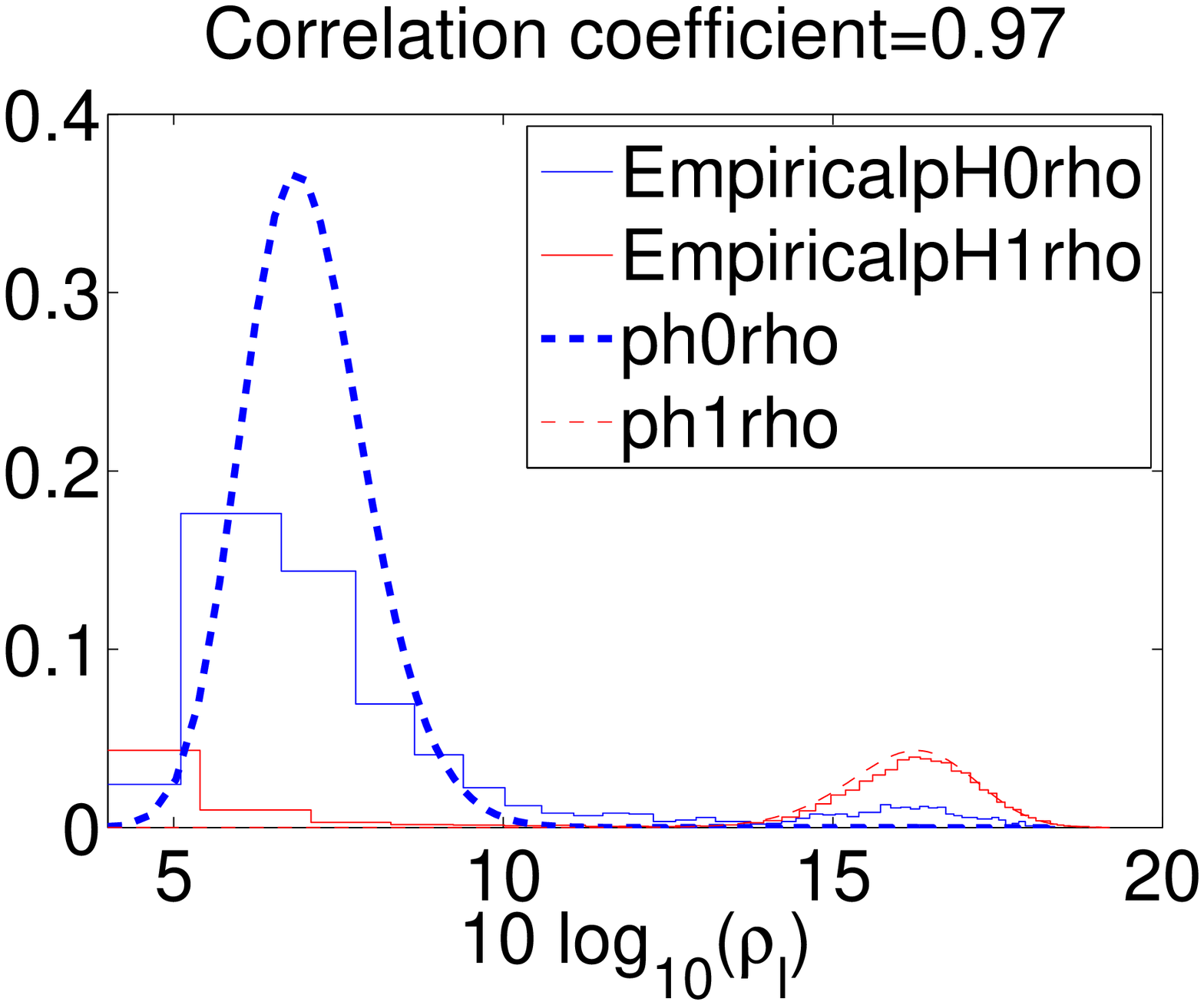}}
\subfigure[\label{fig:EmpiricalVSTheoryCorr4}]{\includegraphics[width=0.5\linewidth]{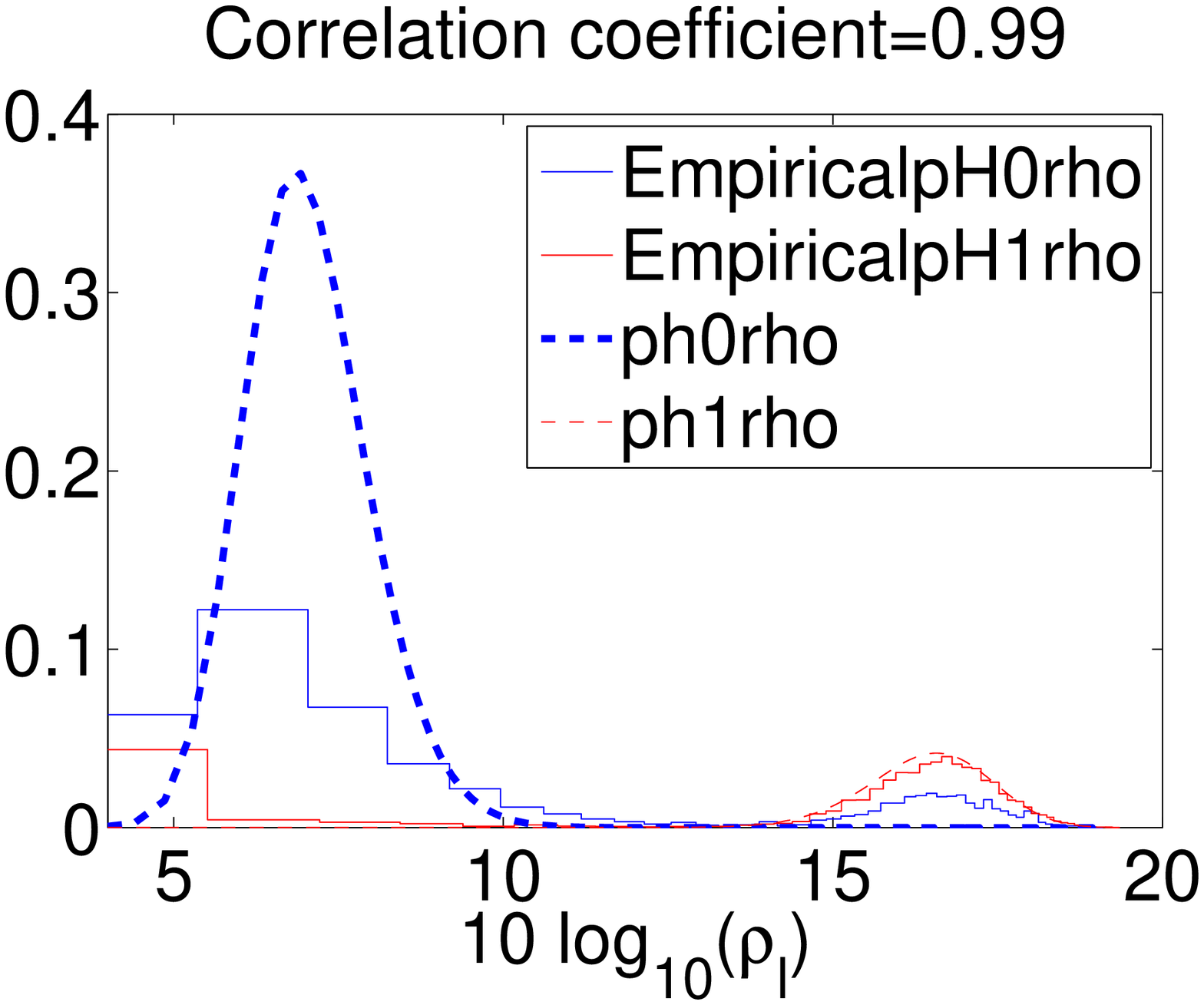}}
}

\caption{Comparison of empirical and derived distributions of the test statistic $\rho_l$ under $H_0$ and $H_1$ hypotheses for SNR$=17$dB and a)$K=N/2$, b)$K=N/4$, c)$K=N/8$, and d)$K=N/16$.}
\label{fig:EmpiricalVSTheoryCorrelated}
\end{figure}
As expected, for low correlations the pdfs derived for the assumption $A1$ provide a close approximation for the $A2$ case, both for $H_0$ and $H_1$ hypotheses.
As the correlation increases, the pdfs of both hypotheses exhibit a second mode at the location of the alternative hypothesis.
This is direct consequence of the high correlation between the components: depending on the noise realization, a component that is ``marked'' as an $H_0$ hypothesis fits the synthetic signal better then the one ``marked'' as an $H_1$.
%The latter is then deemed by the algorithm as noise and either removed by the algorithm, or set to a random parameter, leading to an artifact which appears as if it is generated from $H_0$ hypothesis with the corresponding low statistic $\rho_l$.
Practically, it is, however, not important which component is eventually selected, as long as the artifacts are removed with an appropriately selected threshold $\kappa_l$.
Considering the tails of the pdf $p(\rho_l|H_0)$ we can conclude that in the $A2$ case the threshold $\kappa_l$ computed for the $A1$ assumption seem to be a reasonable practical approximation.

Let us now test the performance of the proposed detector with the adjusted threshold $\kappa_l$.
For that we use the same simulation parameters: we generate a single component with $\tau=0$, $N=128$, and $T_s=1$.
As the performance measure we look at the number of estimated components and the empirical distribution of the estimated delay values versus SNR for the threshold $\kappa_l=1$, i.e., no adjustment, and adjusted threshold 
\begin{equation}  
    \kappa_l=\log(1/\log(1-\epsilon_l)^{-1/N})
    \label{equ:AdjustedKappa}
\end{equation} 
with $\epsilon_l=0.001$.
The latter is selected according to \eqref{equ:TestFunction}.
Also, we will consider cases $K=N$, $K=N/2$, $K=N/4$, and $K=N/16$.
The corresponding plots are summarized in Fig. \ref{fig:EstL}.
\begin{figure*}
\centerline{
\subfigure[\label{fig:LEst}]{\includegraphics[width=0.235\linewidth]{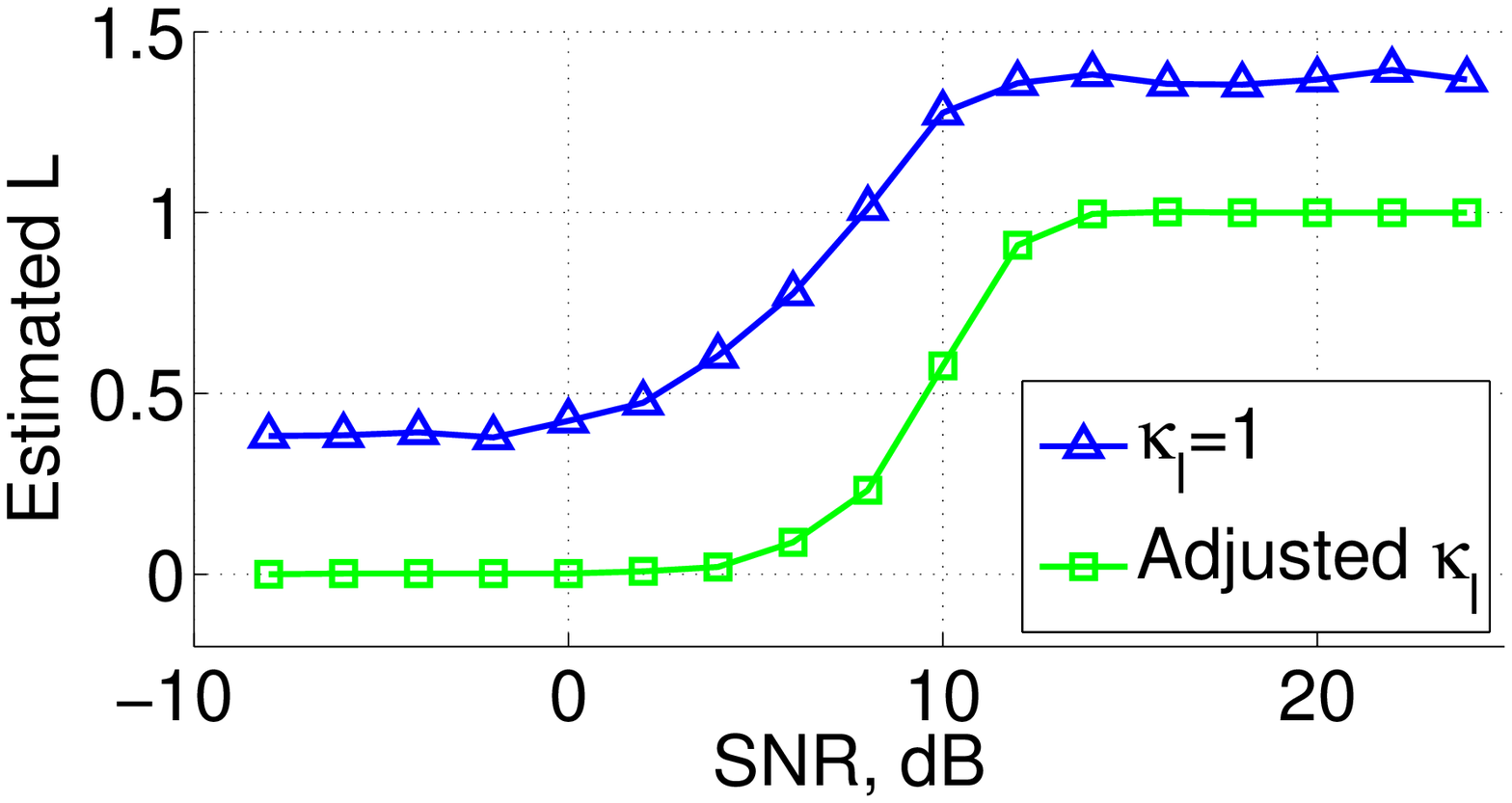}}
\subfigure[\label{fig:LEst_Corr062}]{\includegraphics[width=0.235\linewidth]{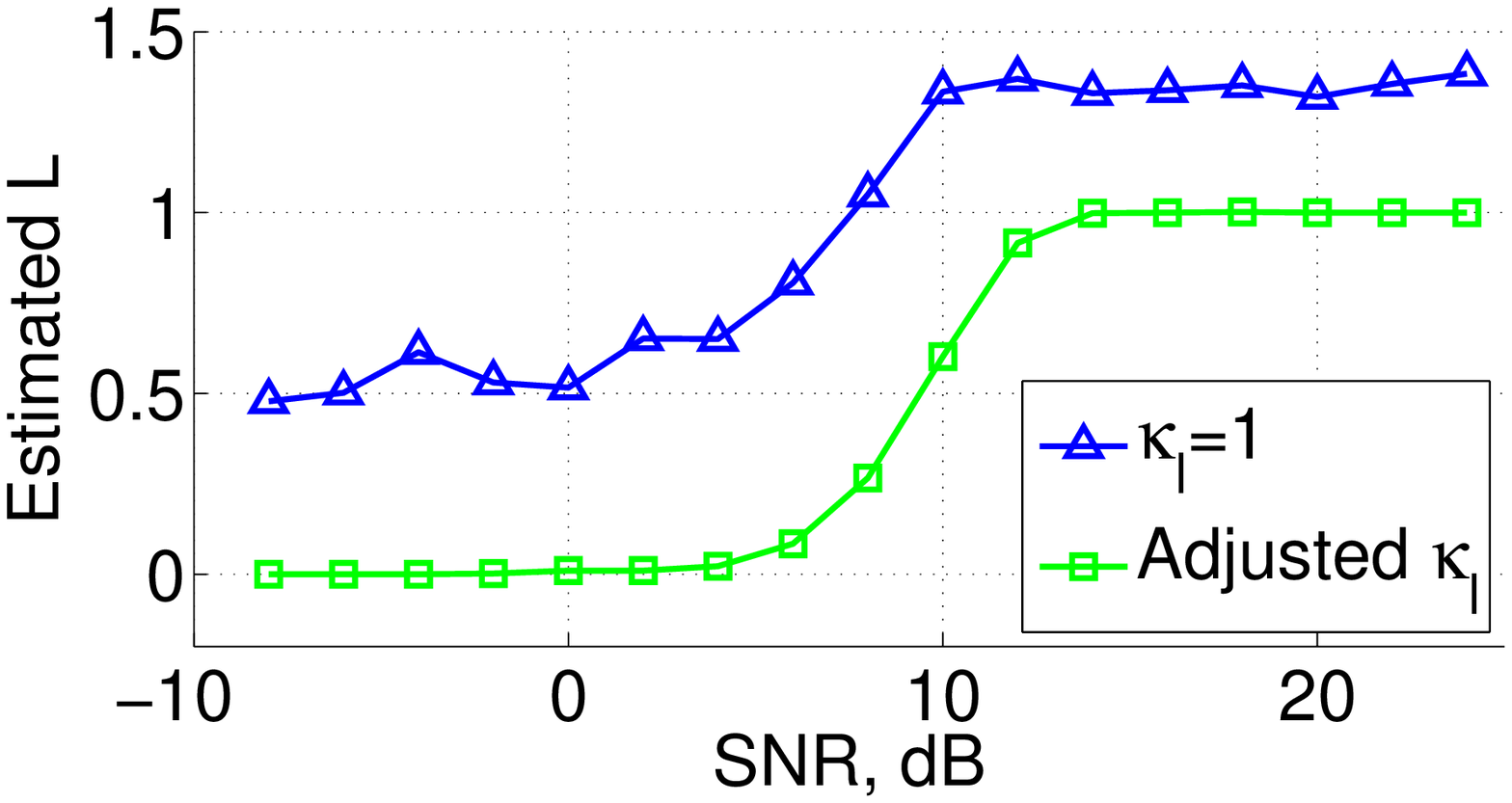}}
\subfigure[\label{fig:LEst_Corr089}]{\includegraphics[width=0.235\linewidth]{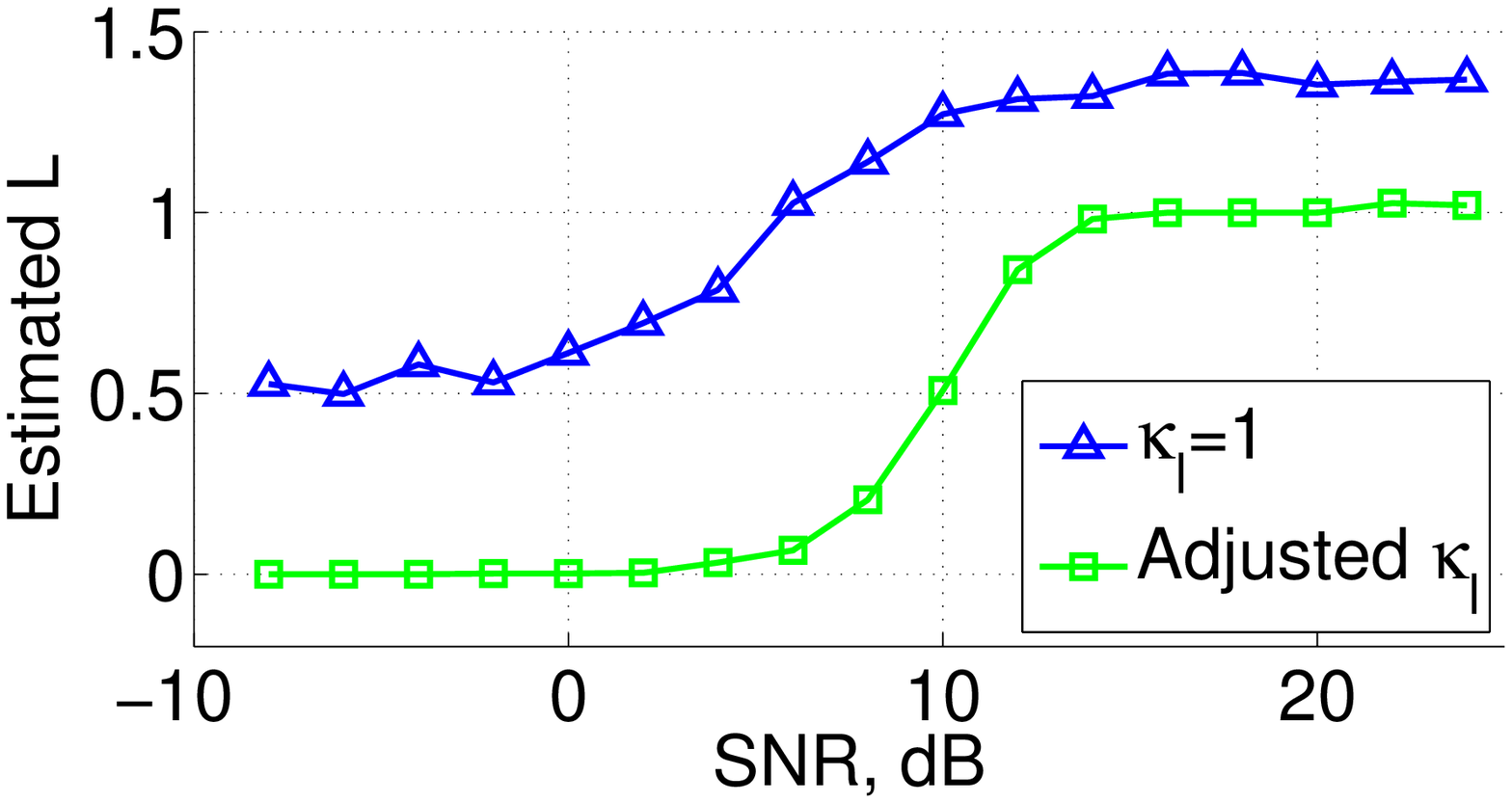}}
\subfigure[\label{fig:LEst_Corr099}]{\includegraphics[width=0.235\linewidth]{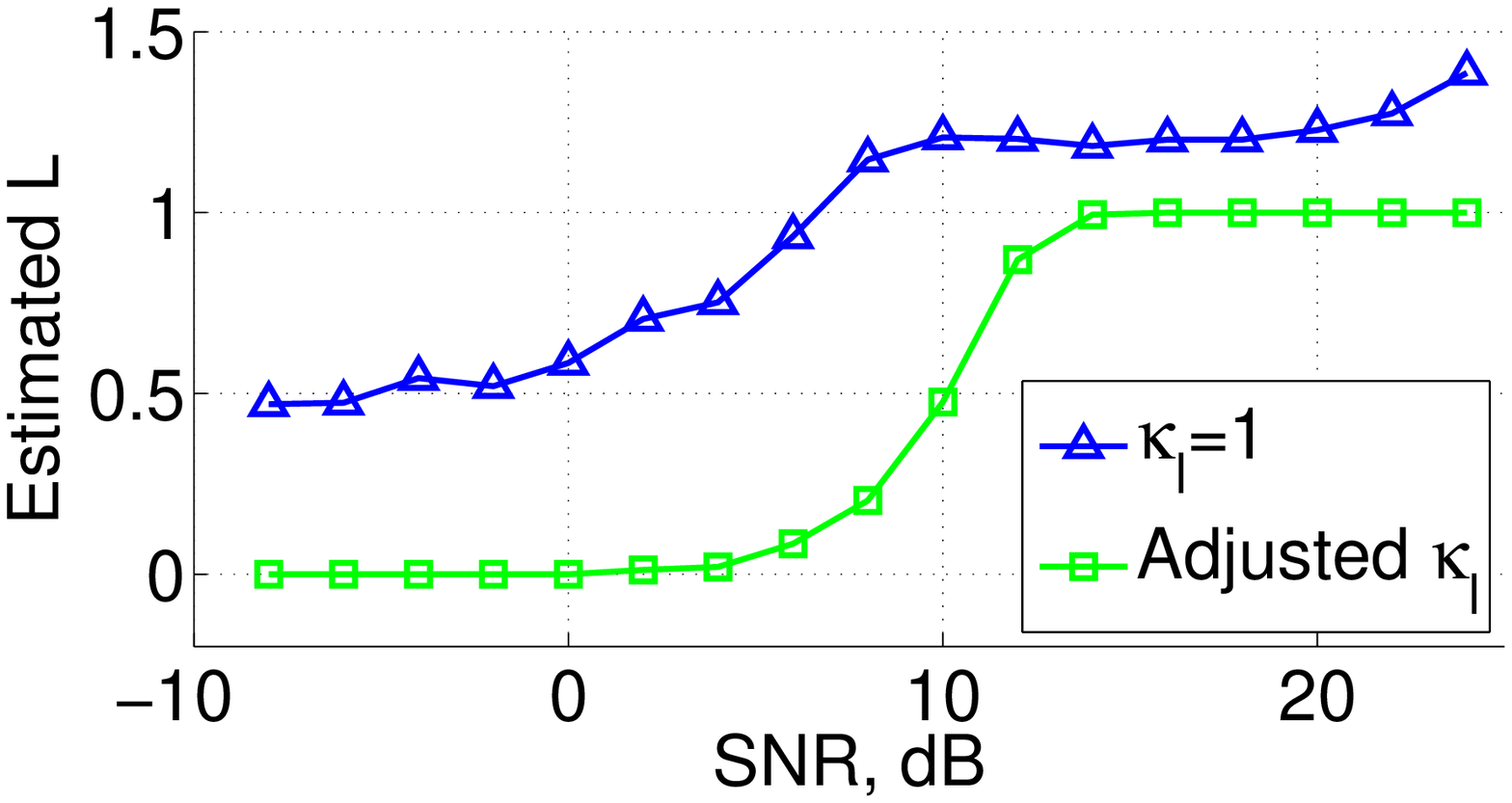}}
}
\centerline{
\subfigure[\label{fig:DelayHist_Kappa1}]{\includegraphics[width=0.235\linewidth]{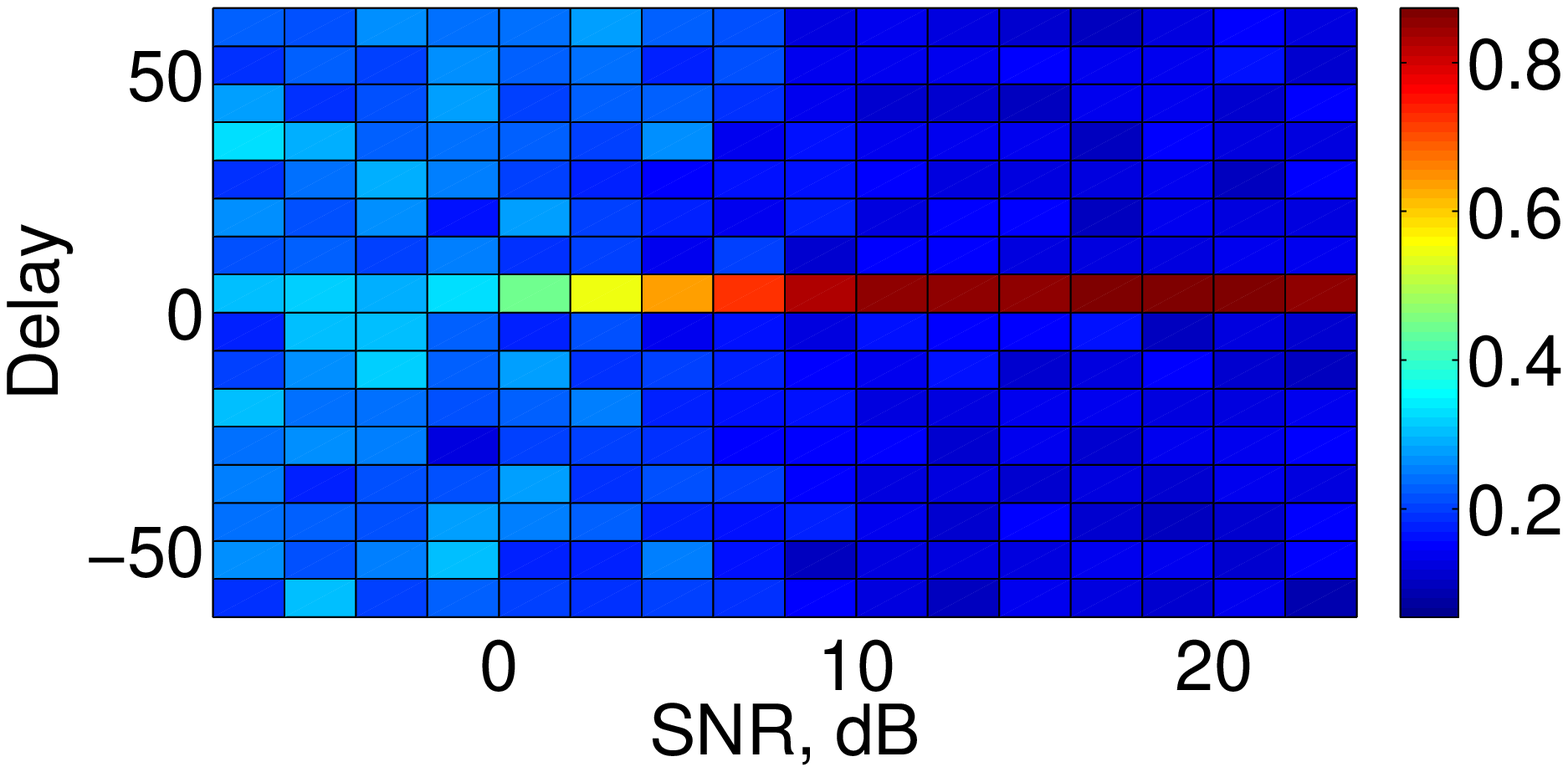}}
\subfigure[\label{fig:DelayHist_Kappa1_Corr062}]{\includegraphics[width=0.235\linewidth]{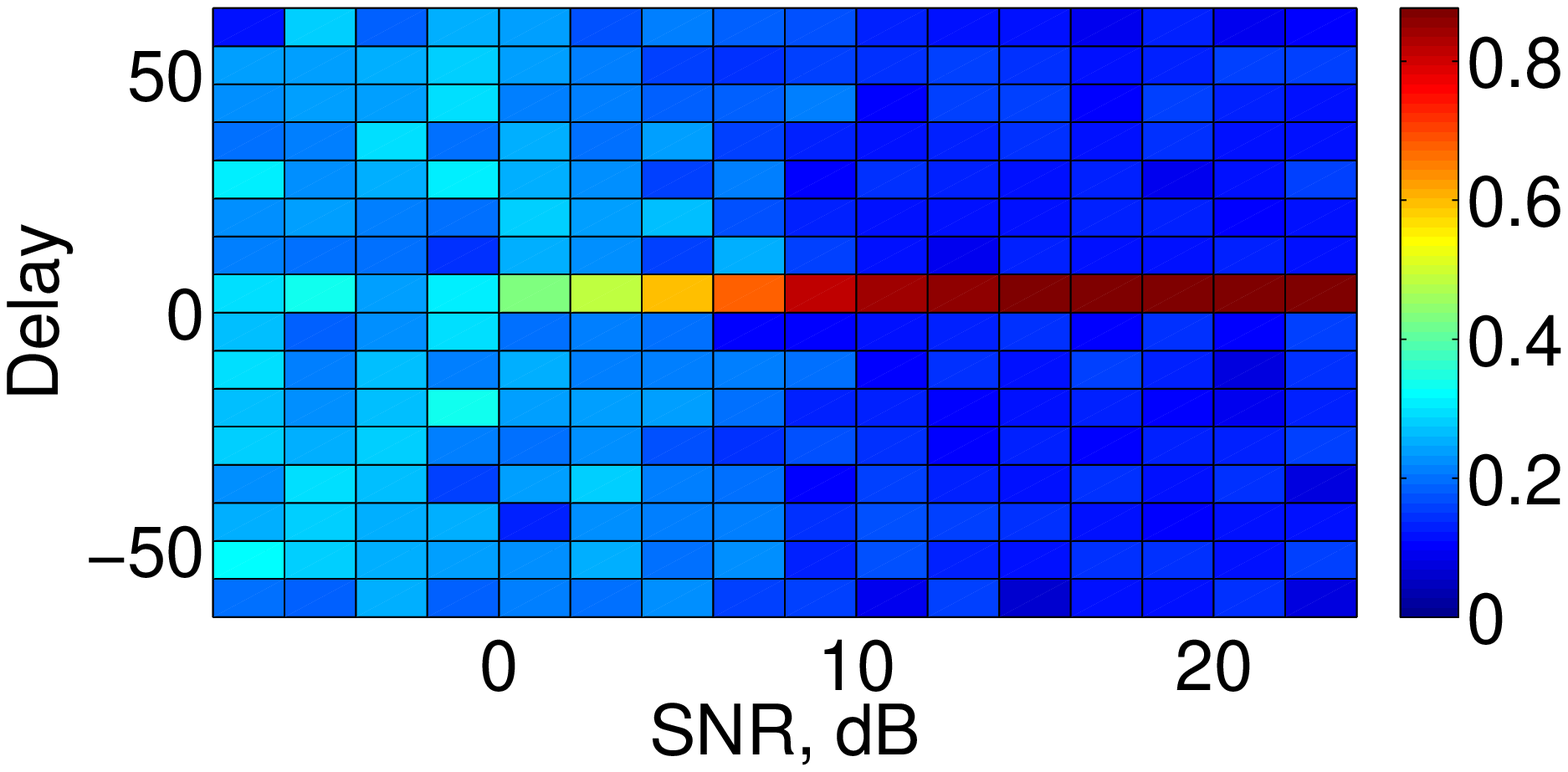}}
\subfigure[\label{fig:DelayHist_Kappa1_Corr089}]{\includegraphics[width=0.235\linewidth]{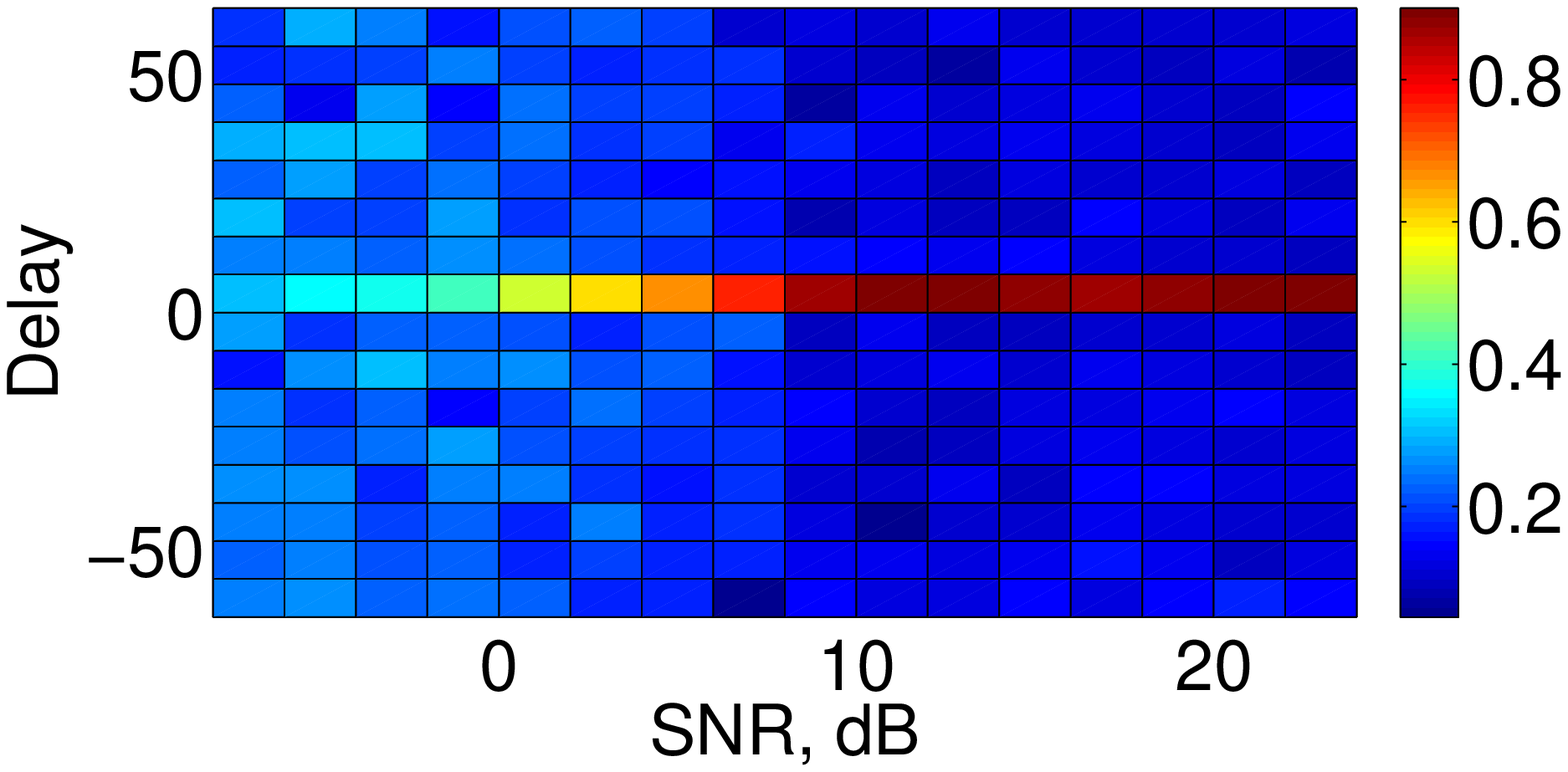}}
\subfigure[\label{fig:DelayHist_Kappa1_Corr099}]{\includegraphics[width=0.235\linewidth]{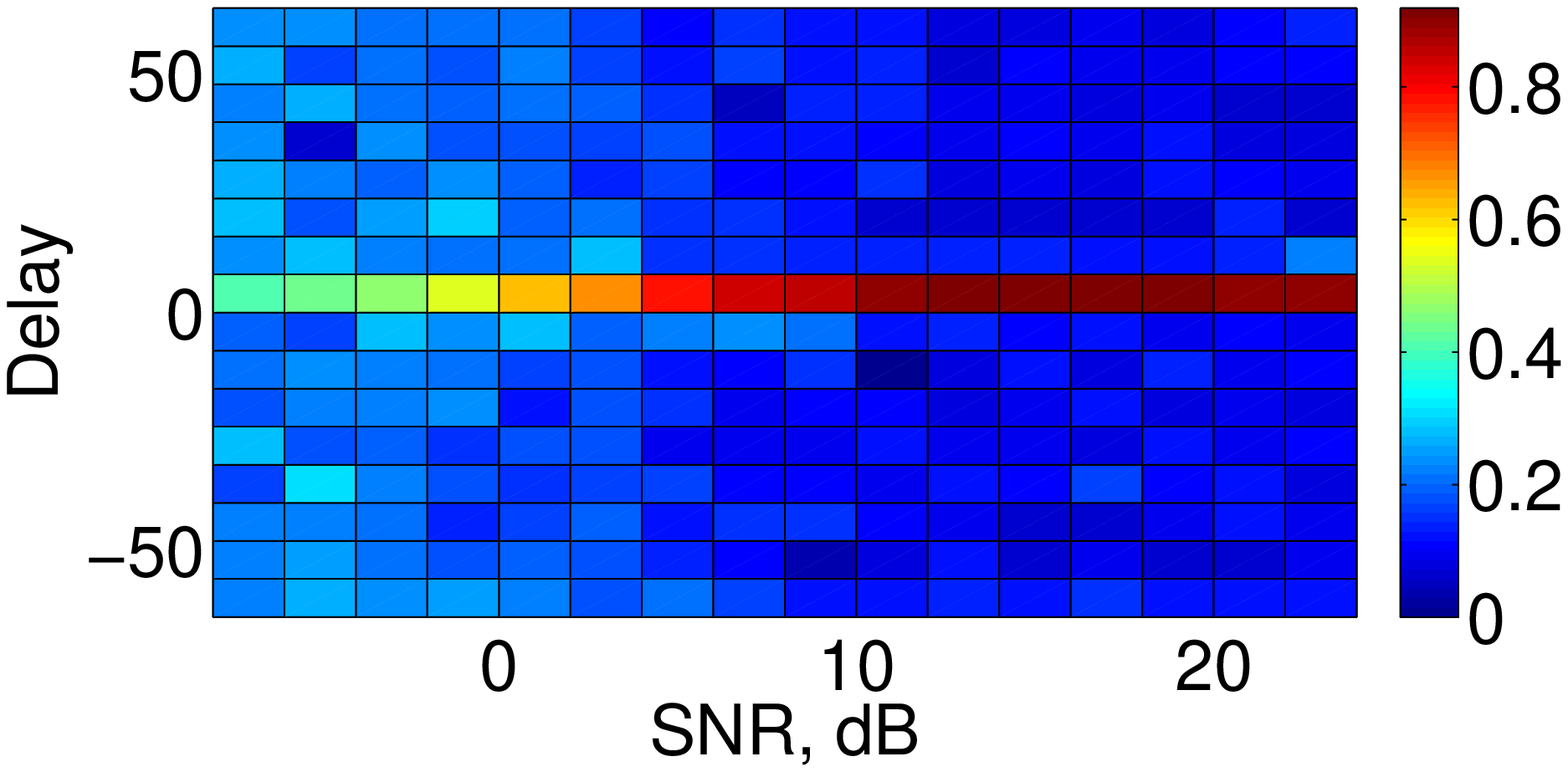}}
}
\centerline{
\subfigure[\label{fig:DelayHist_Gumb}]{\includegraphics[width=0.235\linewidth]{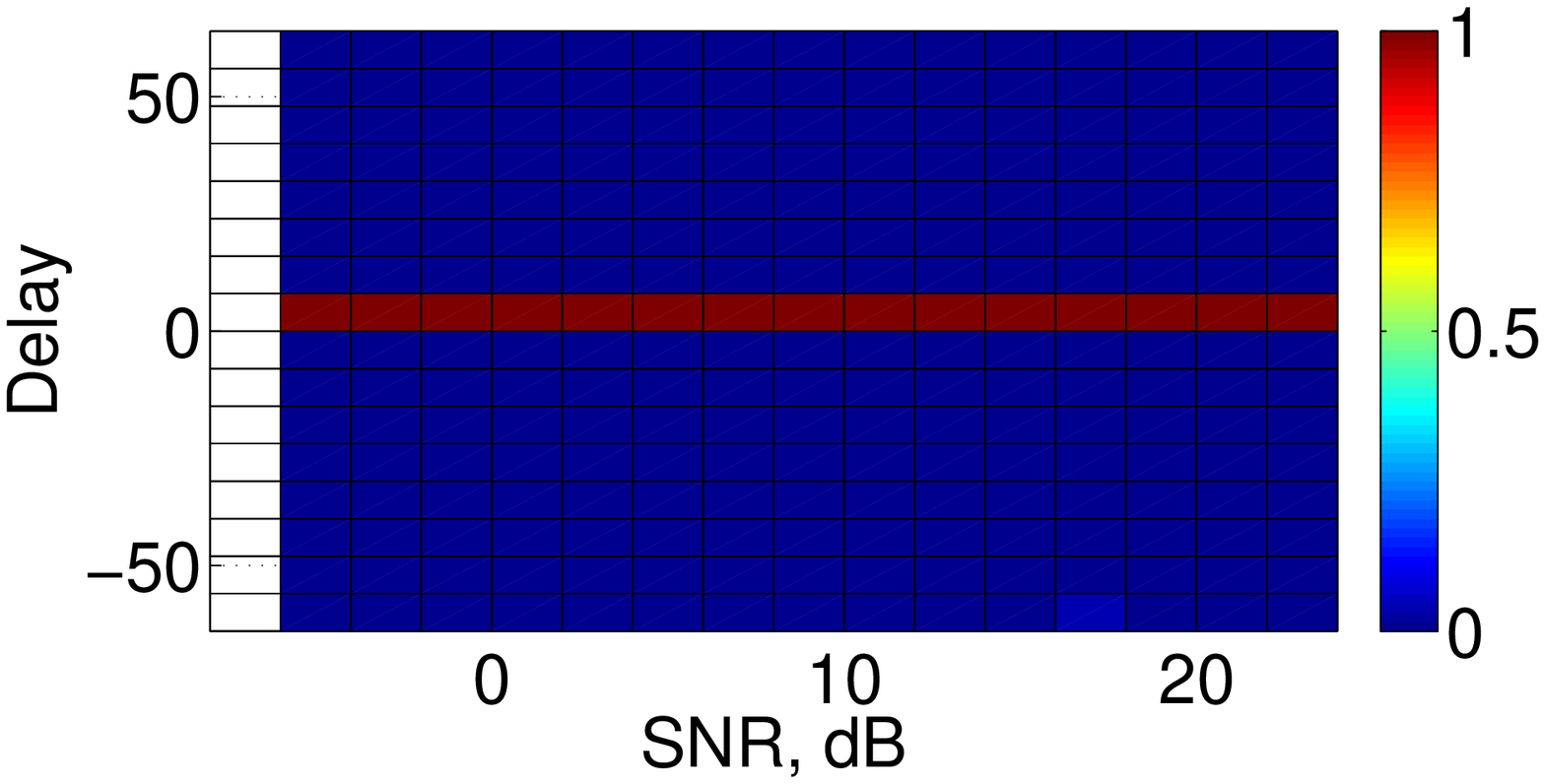}}
\subfigure[\label{fig:DelayHist_Gumb_Corr062}]{\includegraphics[width=0.235\linewidth]{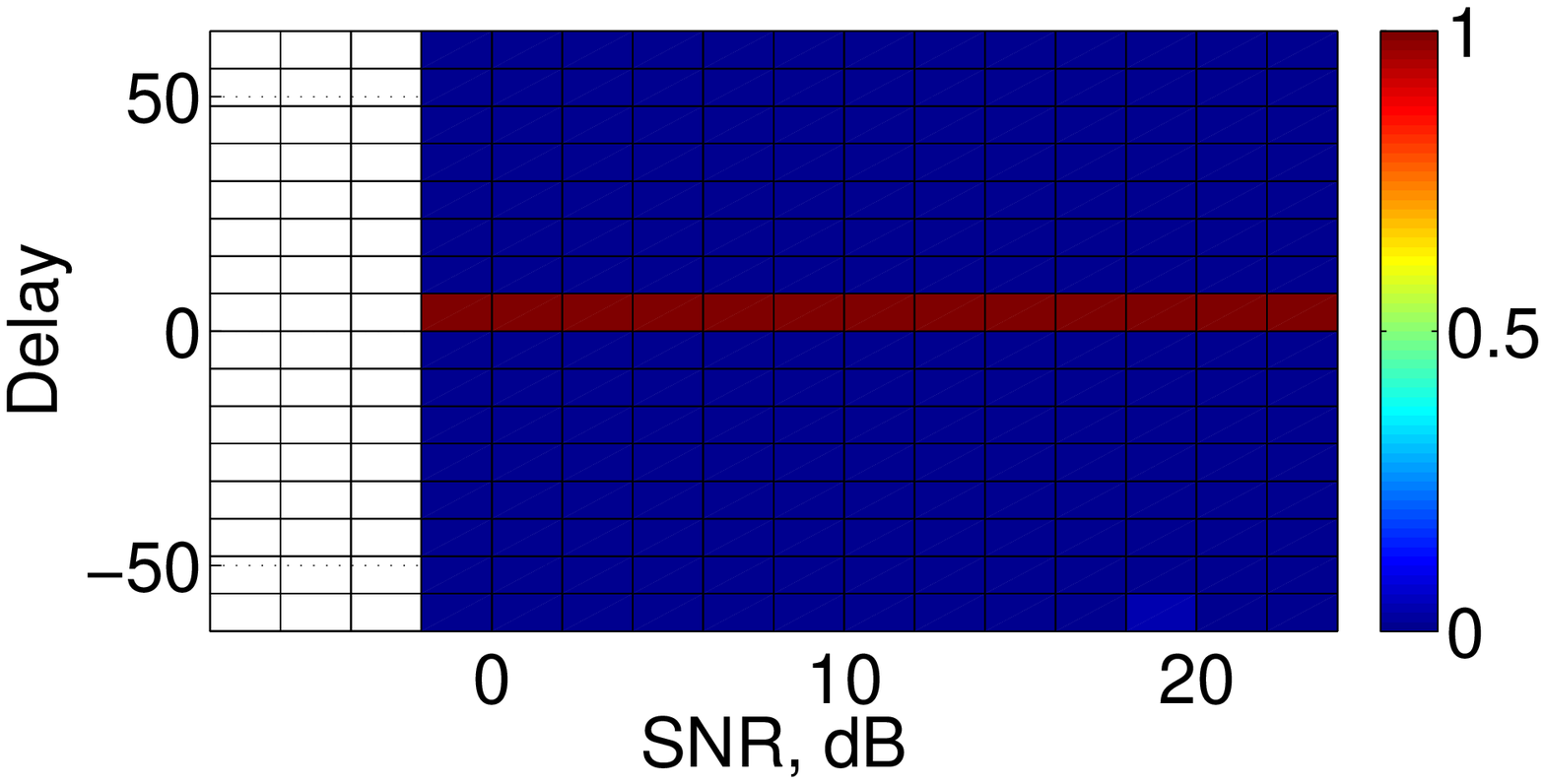}}
\subfigure[\label{fig:DelayHist_Gumb_Corr089}]{\includegraphics[width=0.235\linewidth]{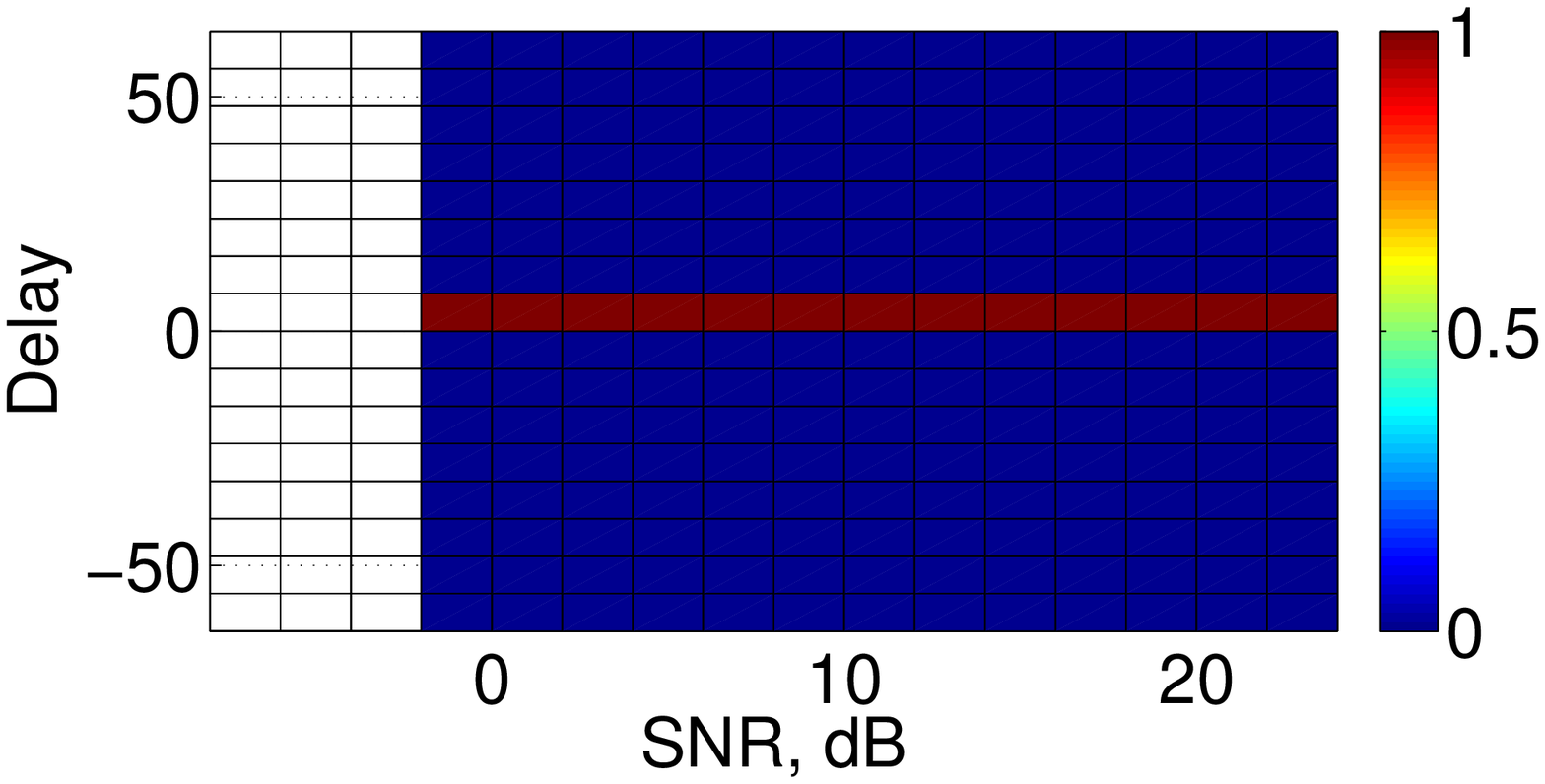}}
\subfigure[\label{fig:DelayHist_Gumb_Corr099}]{\includegraphics[width=0.235\linewidth]{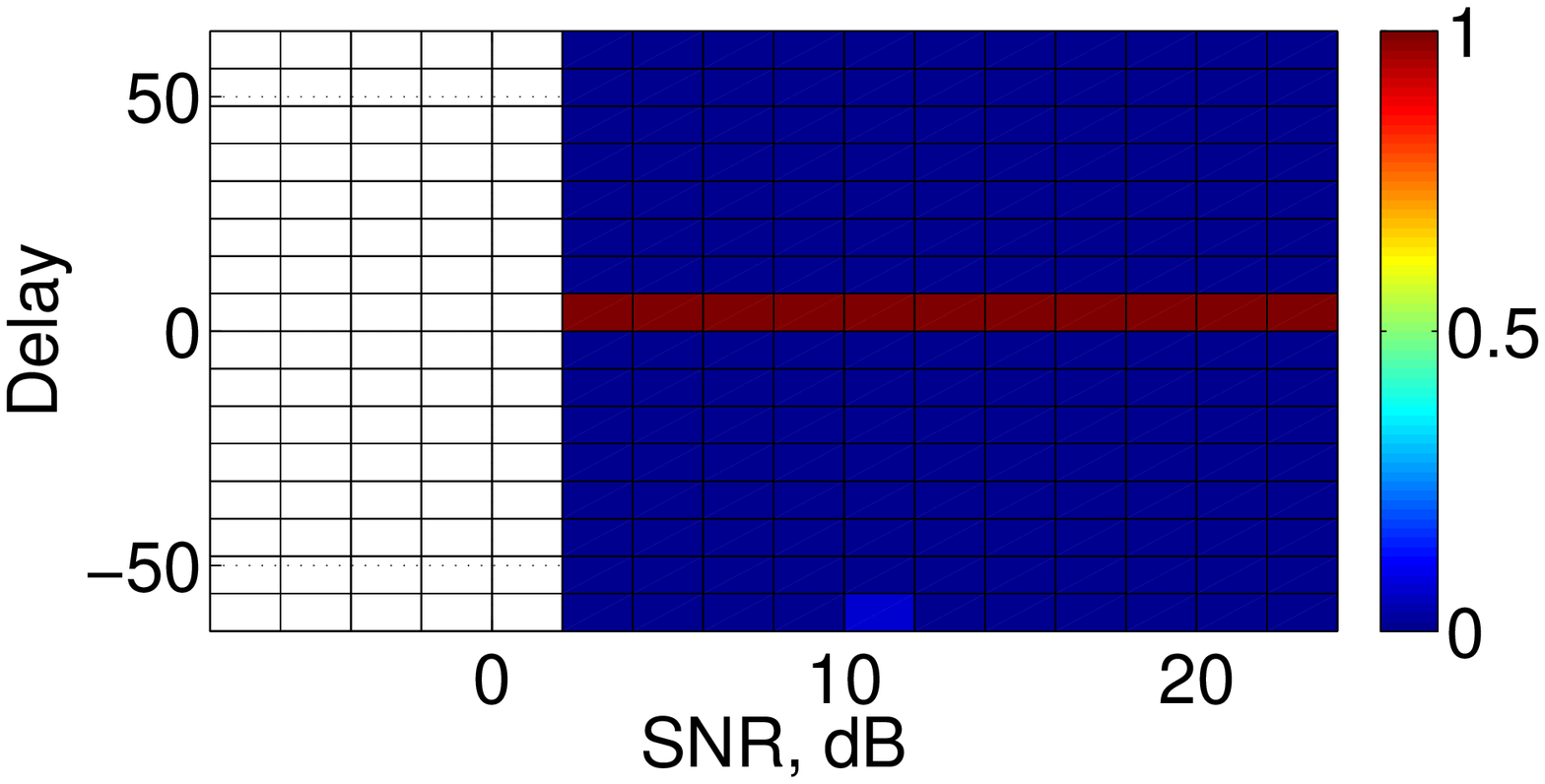}}
}
\caption{Estimated number of components versus SNR for (a,e,i) $K=N$, (b,f,j) $K=N/2$, (c,g,k) $K=N/4$, and (d,h,l) $K=N/16$. In figures (e)-(h) $\kappa_l=1$; in figures (i)-(l) $\kappa_l=\log(1/\log(1-\epsilon_l)^{-1/N})$ with $\epsilon_l=0.001$.}
\label{fig:EstL}
\end{figure*}
As we see, with $\kappa_l=1$ setting, the algorithm mainly detects noise in low SNR (Fig. \ref{fig:DelayHist_Kappa1} \-- \ref{fig:DelayHist_Kappa1_Corr099}) and overestimates the number of components in high SNR regime.
With the adjusted threshold, the number of detections at low SNR is almost zero, yet when a component is detected, it corresponds to the actual multipath component with high probability.

%%%%%%%%%%%%%%%%%%%%%%%%%%%%%%%%%%%%%%%%%%%%%%%%%%%%%%%%%%%%%%%%%%%%%%%%%%%%%%%%%%
\subsection{Superresolution properties of the algorithm}
In the next simulation we investigate the resolution ability of the proposed IARD algorithm for both the assumptions $\mathrm{A1}$ (IARD-A1) and the assumption $\mathrm{A2}$ (IARD-A2).
Here we will consider the case $L=2$, with component delays $\tau_l$ no longer restricted to a sampling grid.
Additionally, we will consider a Doppler shift $\nu_l$ for each component.
This setting will correspond to a time-varying SISO channel model with stationary parameters.
To estimate Doppler frequency we will consider $M=25$ consecutive channel measurements, so that the model of a single component $\vect{s}(\vect{\theta}_l)$ is represented as $\vect{s}(\tau_l,\nu_l)=\mathrm{vec}\{\vect{X}\}$, where $\vect{X}$ is an $R\times M$ matrix and $[\vect{X}]_{r,m}=s[r-\tau_l/T_s]\mathrm{e}^{j2\pi\nu_l r m T_s}$, $r=0,\ldots,R-1$, $m=0,\ldots,M-1$.
The signal $s[n]$ is a downsampled version of the actual $10$MHz-wide calibration signal used in the aeronautical channel measurement campaign \cite{Schneckenburger2014}.
The used sampling period is $T_s=4\mu\mathrm{s}$, which results in $R=128$ samples per single channel measurement.
The synthetic delays $\vect{\tau}=[\tau_1,\tau_2]^T$ of the components are generated as follows: $\tau_1$ is uniformly drawn from the interval $[0, T_s]$ and  $\tau_2=\tau_1+\Delta\cdot T_s$, with $\Delta$ being a simulation parameter.
The Doppler frequency $\nu_1$ of the first component is drawn uniformly from the interval $[-200,200]$Hz; for the second component we select $\nu_2=\nu_1+\epsilon_{\nu}$, where $\epsilon_{\nu}$ is a random jitter in the interval $[-2,2]$Hz.
The weights of both components have unit magnitude and uniformly distributed phase drawn from the interval $[0,2\pi]$.
For both IARD-A1 and IARD-A2 we will select the threshold according to \eqref{equ:AdjustedKappa}.

For comparison purposes we will also consider a classical SAGE algorithm \cite{Fleury_SAGE99} that employs Bayesian Information Criterium (BIC) \cite{lanterman00schwarz} to select the model order.
Two different implementations of the SAGE algorithm with BIC criterion are compared.
The first implementation (SAGE-BIC-1) exploits the signal detection method based on the eigenstructure of the estimated signal covariance matrix \cite{WaxKailath85}.
This algorithm first estimates the correlation matrix of the input signal $\vect{y}$ using $N=R\times M$ data samples;
then, the information-theoretic criterion is applied to the eigenvalues of the correlation matrix following the scheme described in \cite{WaxKailath85}.
This gives an estimate of the number of signals, which is then plugged in the SAGE algorithm to estimate signal parameters.
The second implementation (SAGE-BIC-2) estimates several models with different number of components $L$ using the SAGE algorithm as follows: it starts with the model order $L=0$ and sequentially increases the model order until the minimum of the BIC criterium is achieved, each time fitting the model anew.
The BIC criterion is evaluated as $$\mathrm{BIC}(L)=-\log\left(p(\vect{y}|\hvect{\Theta},\hvect{w})\big|_{L}\right)+\frac{8}{2}L\log(N)$$ for each possible value of $L$.
Here $\log\left(p(\vect{y}|\hvect{\Theta},\hvect{w})\big|_{L}\right)$ is the value of the log-likelihood function evaluated at maximum, under assumption that the model order is $L$.
The penalty factor $\frac{8}{2}L\log(N)$ arises as follows: penalization per single complex amplitude is $\log(N)$, and per additional unknown time/frequency shift is $\frac{3}{2}\log(N)$ (see \cite{Djuric1996} and \cite{Stoica2004} for more details).
Note that in this realisation SAGE-BIC-2 requires fitting multiple models to find the minimum of the BIC criterion.
It is thus computationally very inefficient for realistic channels, where $L$ might range up to several tens of components and number of samples $N$ is on the order $\thicksim 10^{3}\--10^{5}$.

As the performance criteria we compute the averaged number of detected components $\widehat{L}$, the probability of detecting exactly two components $P_D^{(L=2)}$, the averaged delay root median squared error (RMeSE) $\mathrm{RMeSE}(\hvect{\tau})$ normalized by the sampling period $T_s$, and  Doppler RMeSE $\mathrm{RMeSE}(\hvect{\nu})$, normalized by the Doppler resolution $1/NT_s$.
The latter two quantities are computed only for the cases when a correct number of components is detected.
Note that at low SNR the component detection rate will also be low, which is why the median squared error is used instead of mean squared error.
Additionally, we evaluate the averaged computation time per single algorithm run.
\begin{figure*}
\psfrag{Lnum}[bc][bc][0.7][0]{$\widehat{L}$}
\psfrag{Distance}[tc][tc][0.7][0]{$100\%\times\Delta$}
\psfrag{Probability}[bc][bc][0.7][0]{$P_{D}^{(L=2)}$}
\psfrag{DelayError}[bc][bc][0.7][0]{$\mathrm{RMeSE}(\hvect{\tau})/T_s$}
\psfrag{DoppError}[bc][bc][0.7][0]{$\mathrm{RMeSE}(\hvect{\nu})\times NMT_s$}
\centerline{
\subfigure[\label{fig:SNR5_LEst}]{\includegraphics[width=0.245\linewidth]{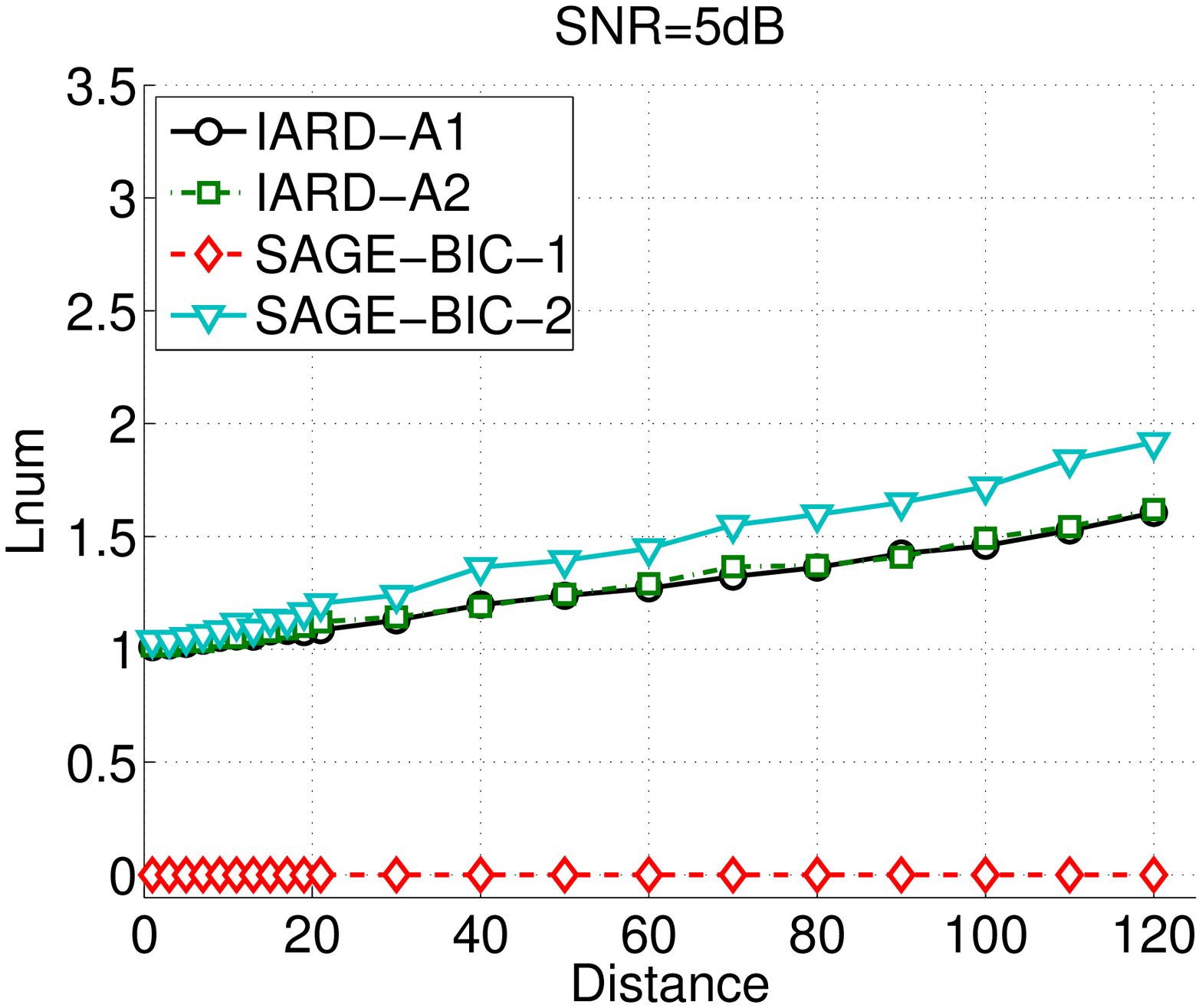}}
\subfigure[\label{fig:SNR10_LEst}]{\includegraphics[width=0.245\linewidth]{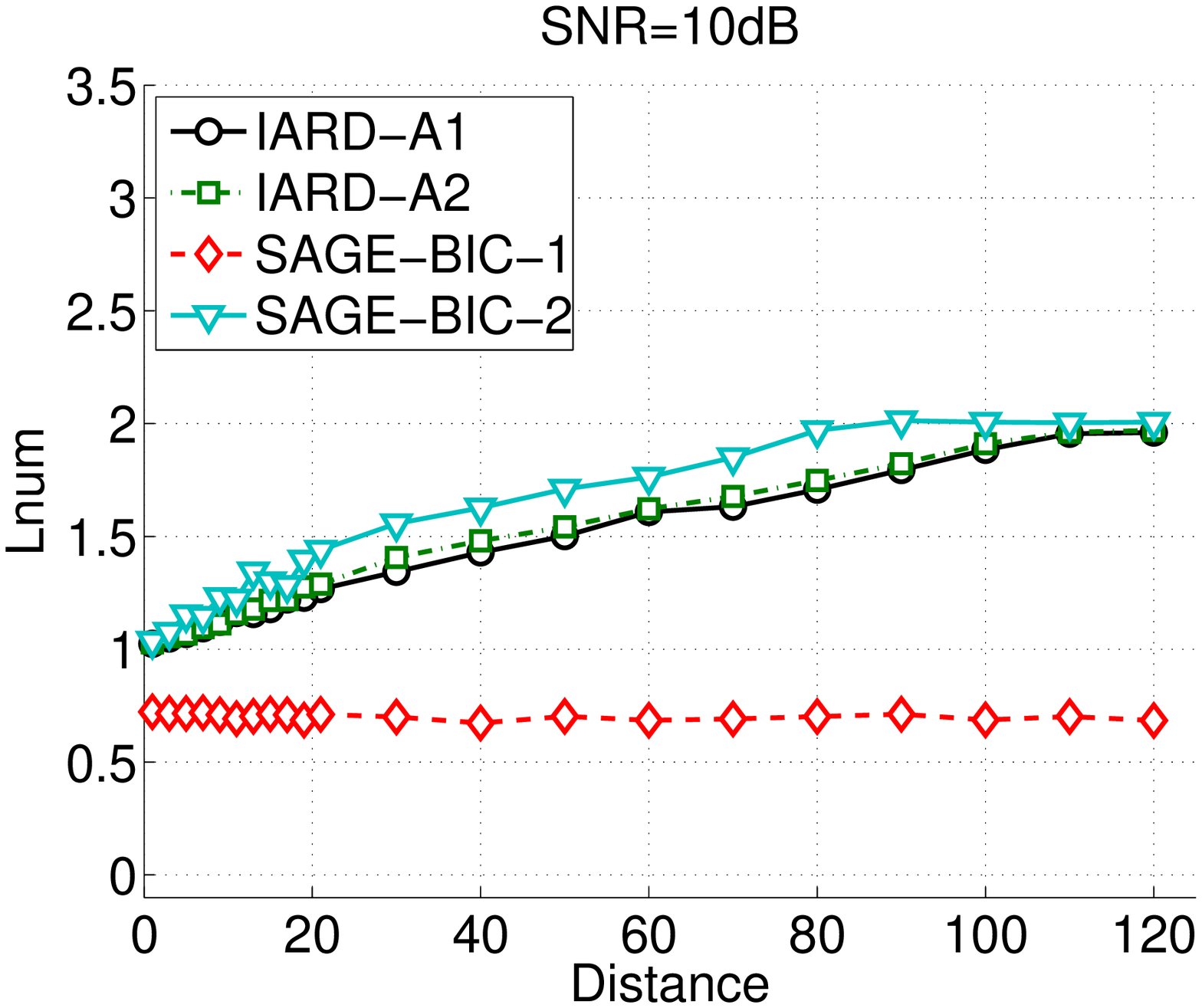}}
\subfigure[\label{fig:SNR20_LEst}]{\includegraphics[width=0.245\linewidth]{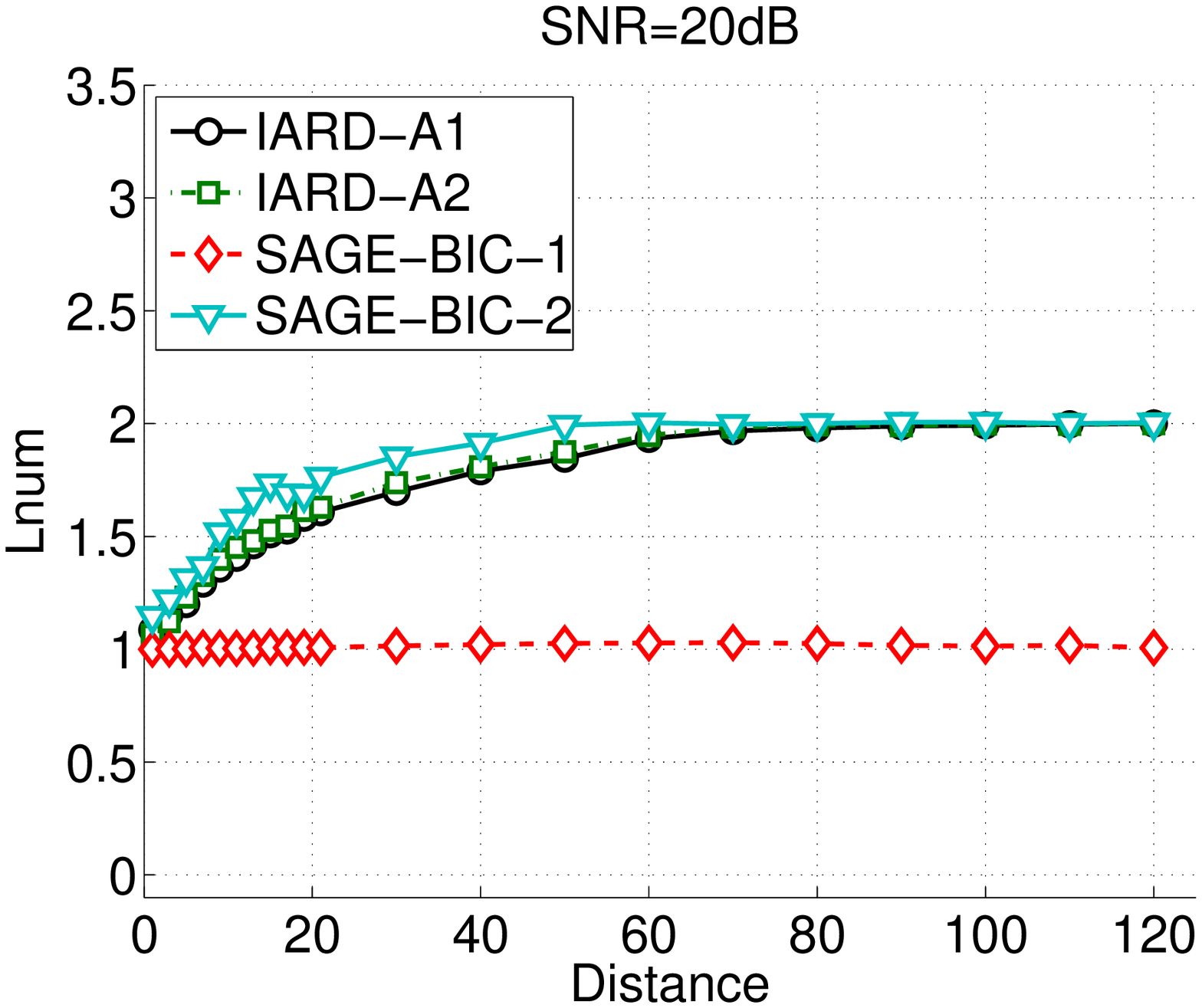}}
\subfigure[\label{fig:SNR30_LEst}]{\includegraphics[width=0.245\linewidth]{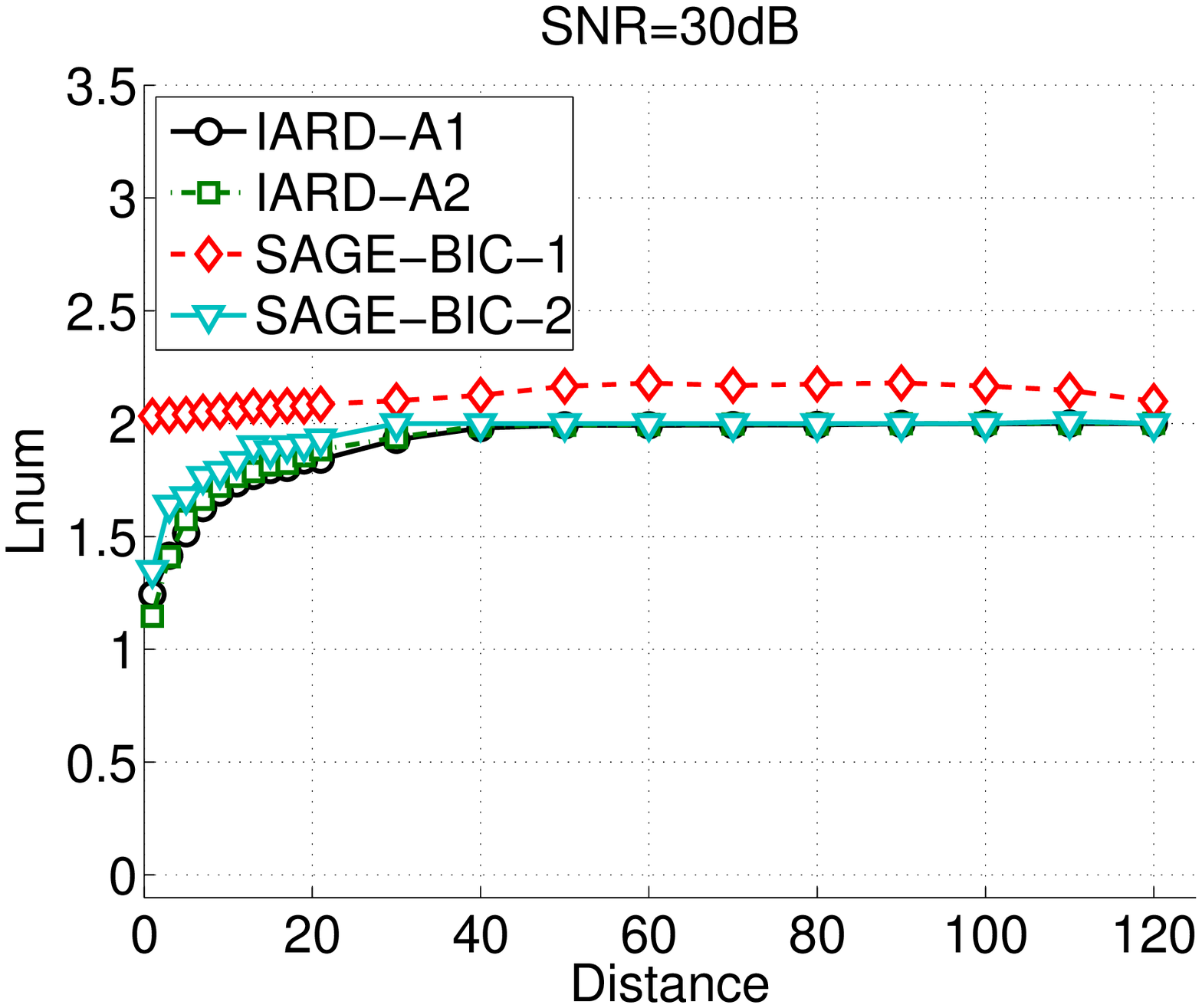}}
}
\centerline{
\subfigure[\label{fig:SNR5_DProb}]{\includegraphics[width=0.245\linewidth]{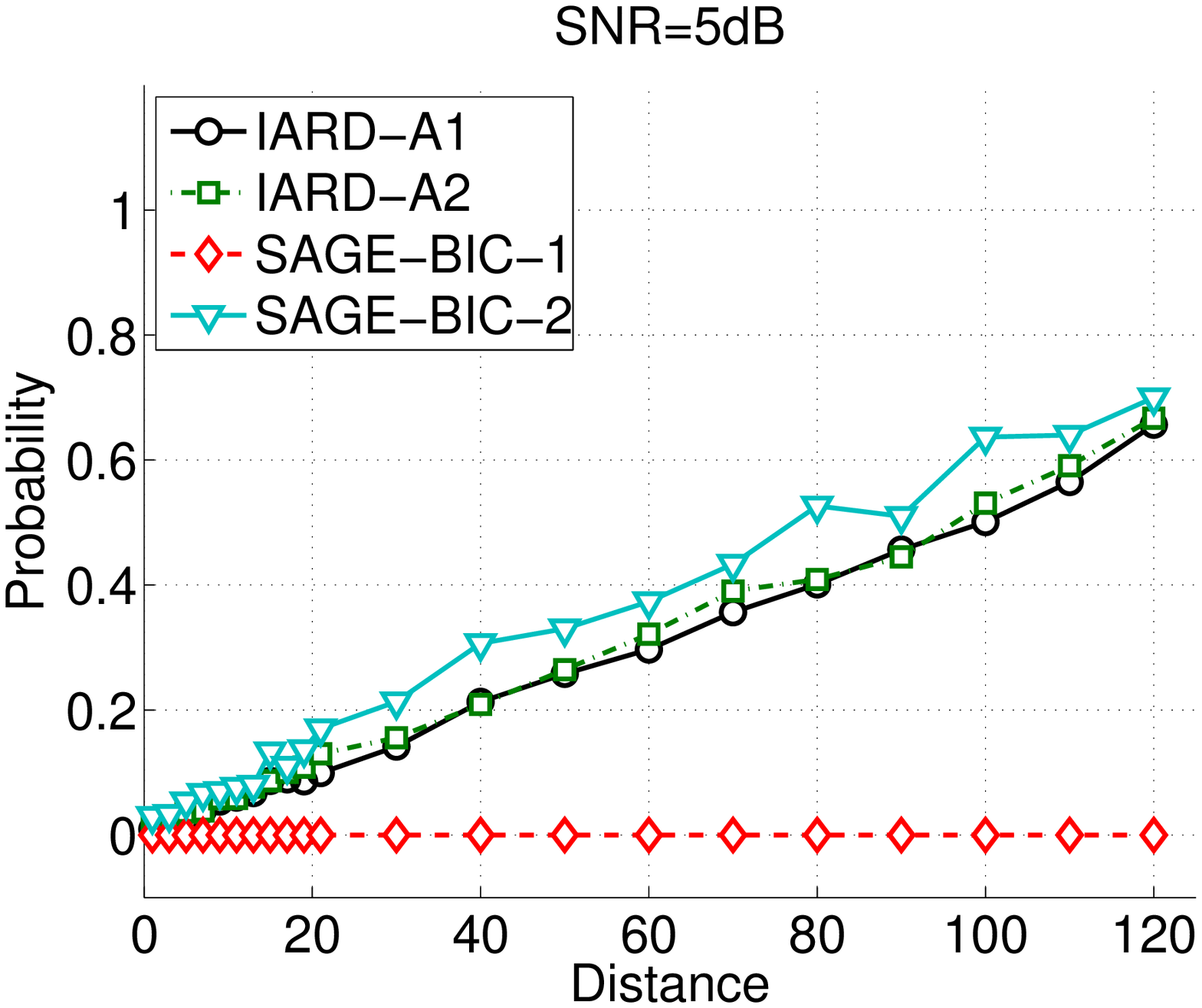}}
\subfigure[\label{fig:SNR10_DProb}]{\includegraphics[width=0.245\linewidth]{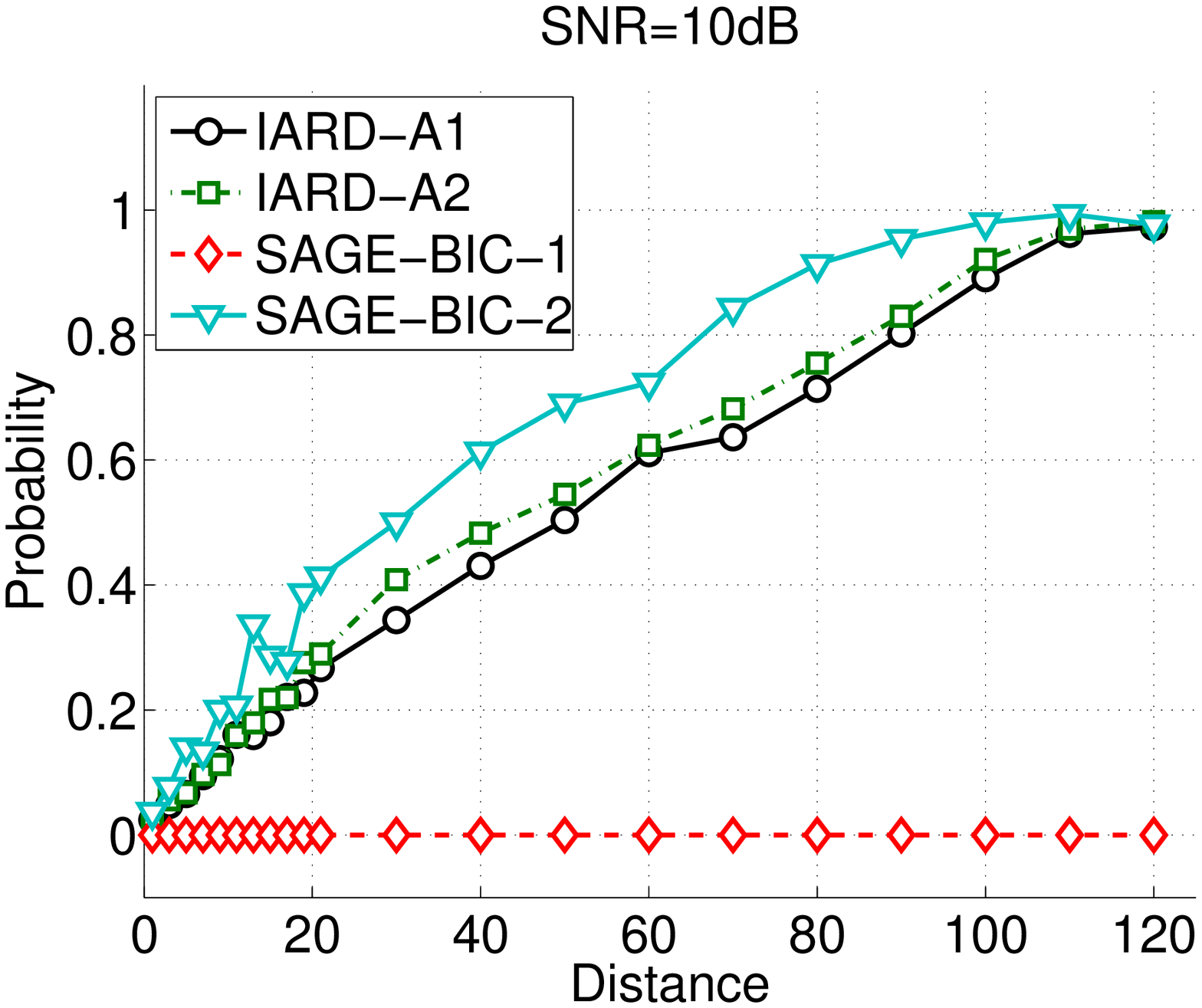}}
\subfigure[\label{fig:SNR20_DProb}]{\includegraphics[width=0.245\linewidth]{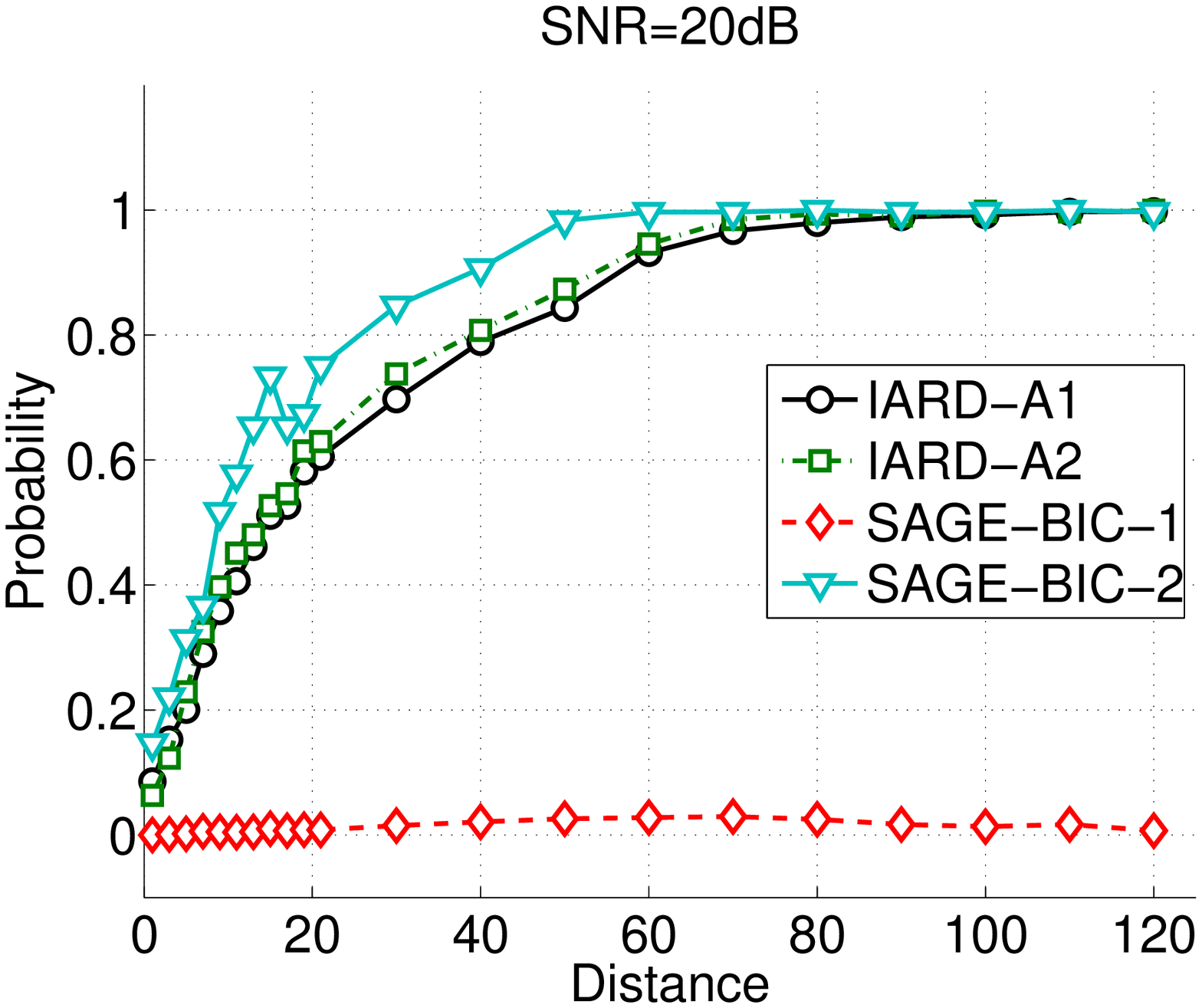}}
\subfigure[\label{fig:SNR30_DProb}]{\includegraphics[width=0.245\linewidth]{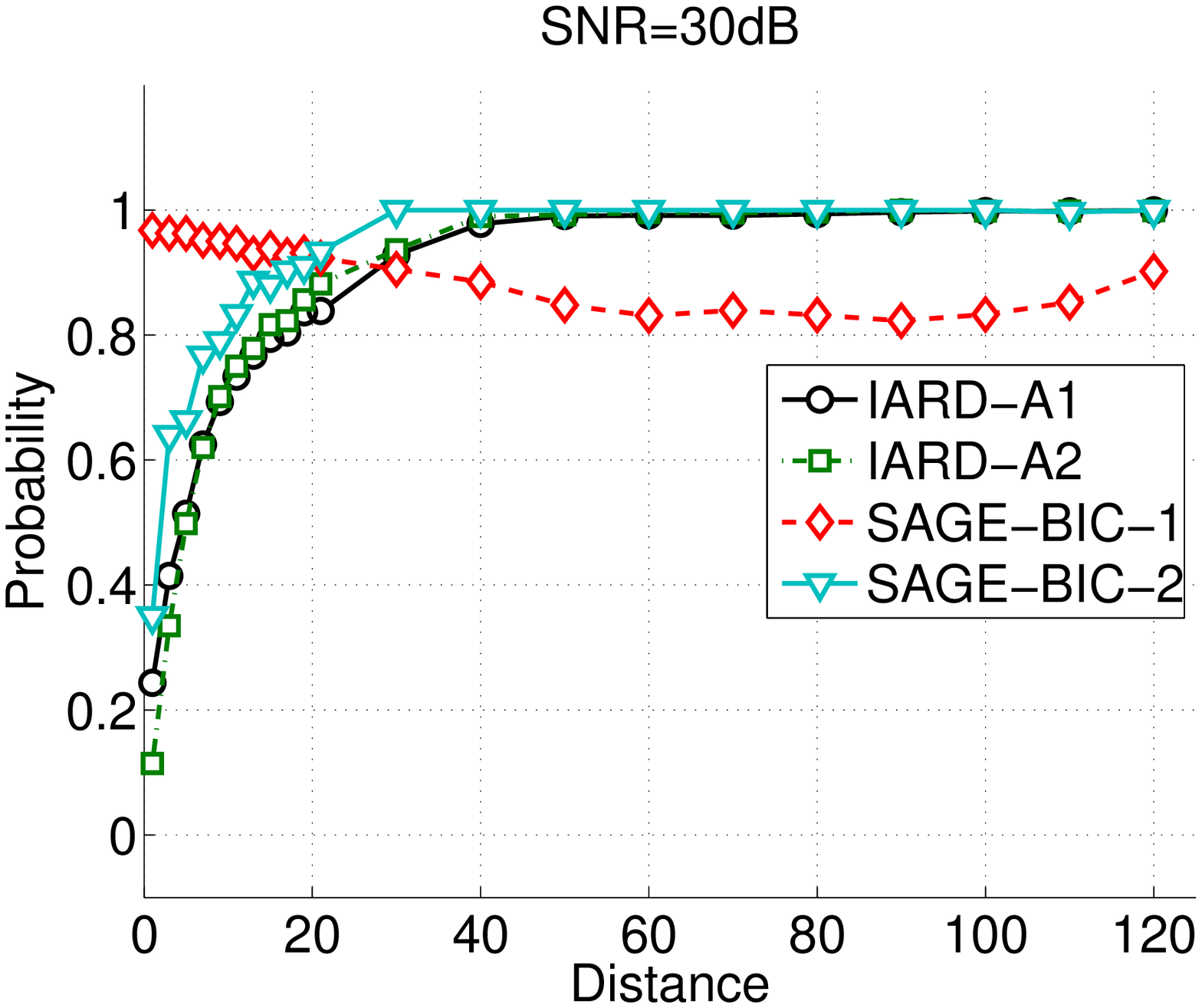}}
}
\centerline{
\subfigure[\label{fig:SNR5_Delay}]{\includegraphics[width=0.245\linewidth]{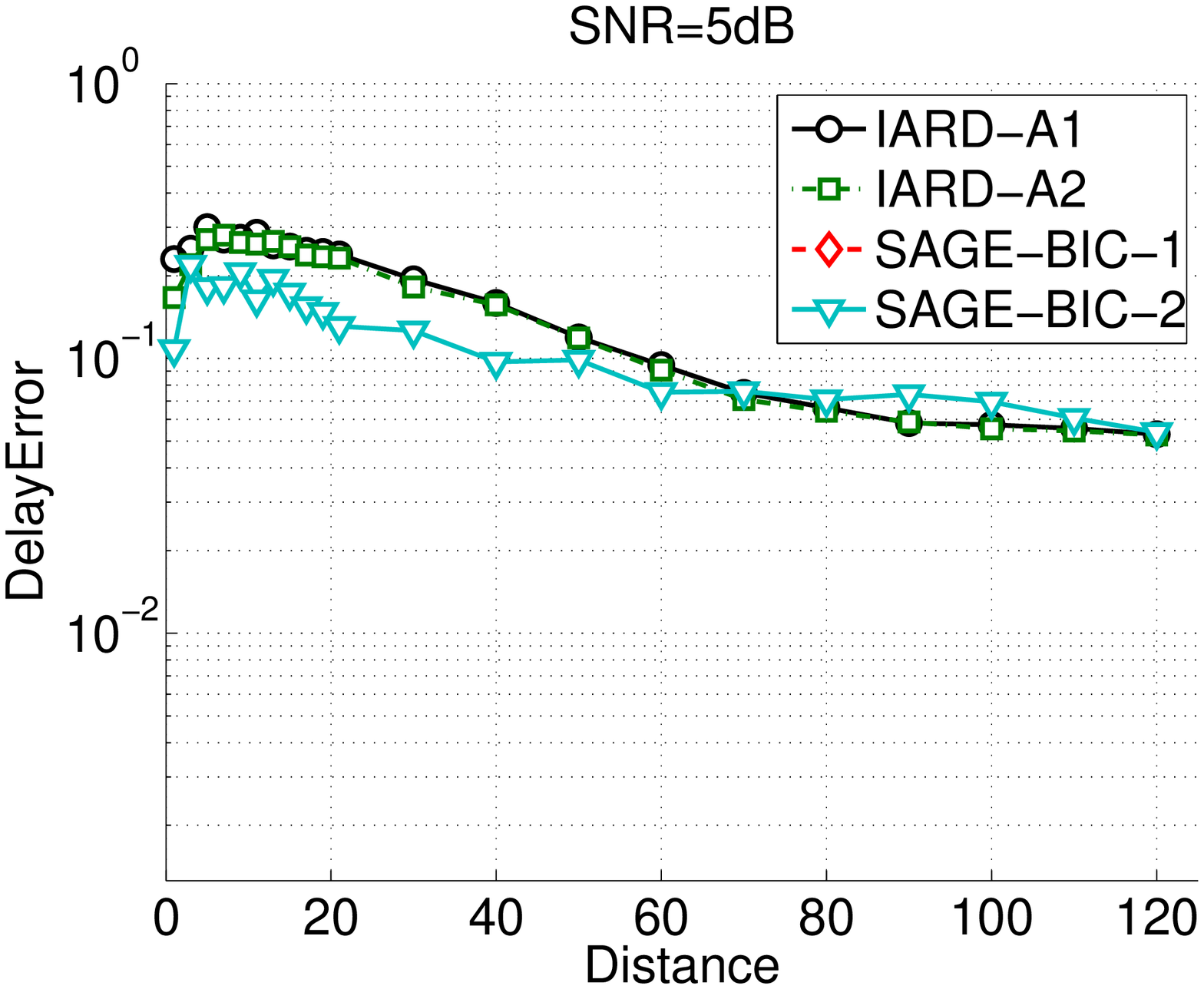}}
\subfigure[\label{fig:SNR10_Delay}]{\includegraphics[width=0.245\linewidth]{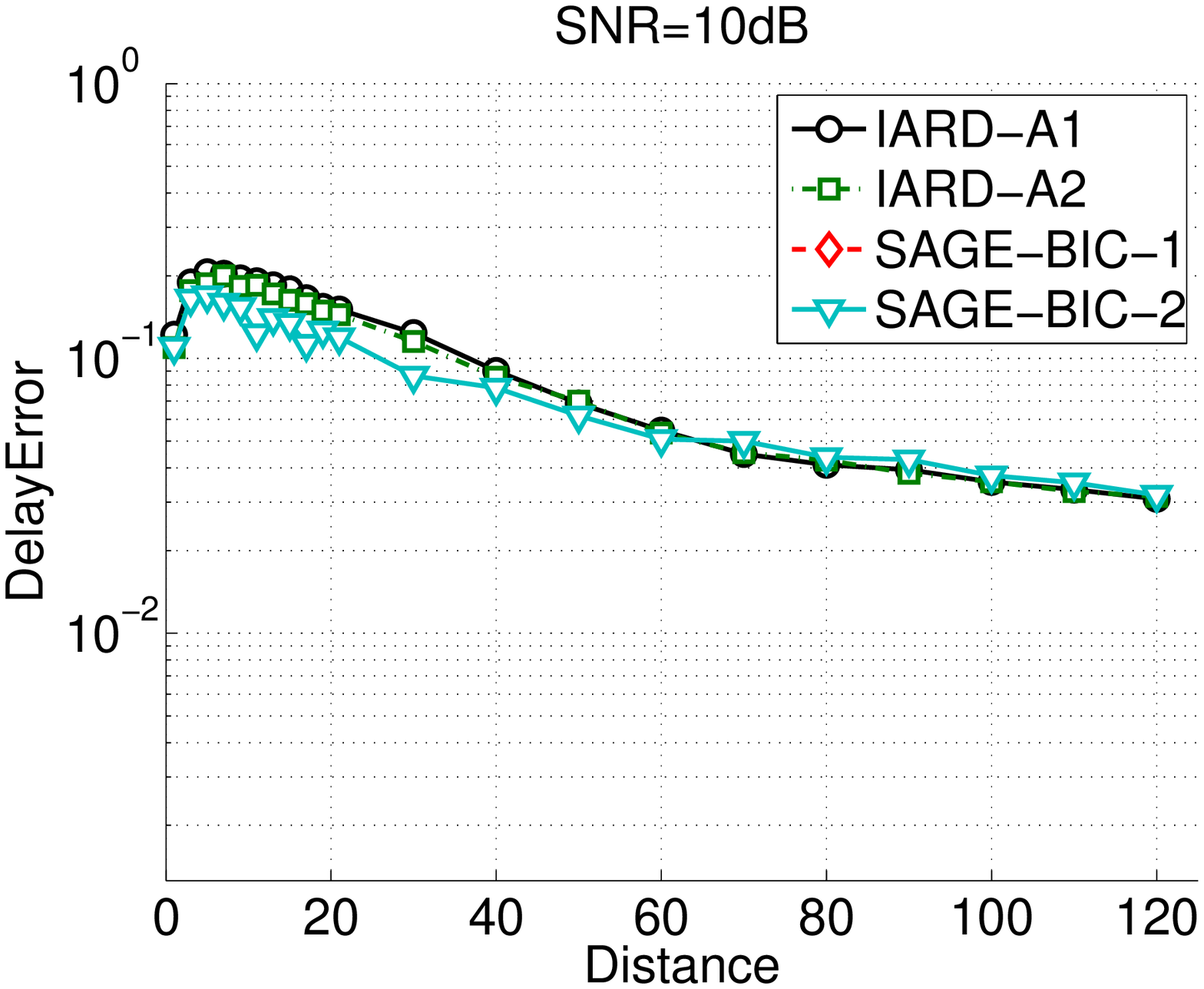}}
\subfigure[\label{fig:SNR20_Delay}]{\includegraphics[width=0.245\linewidth]{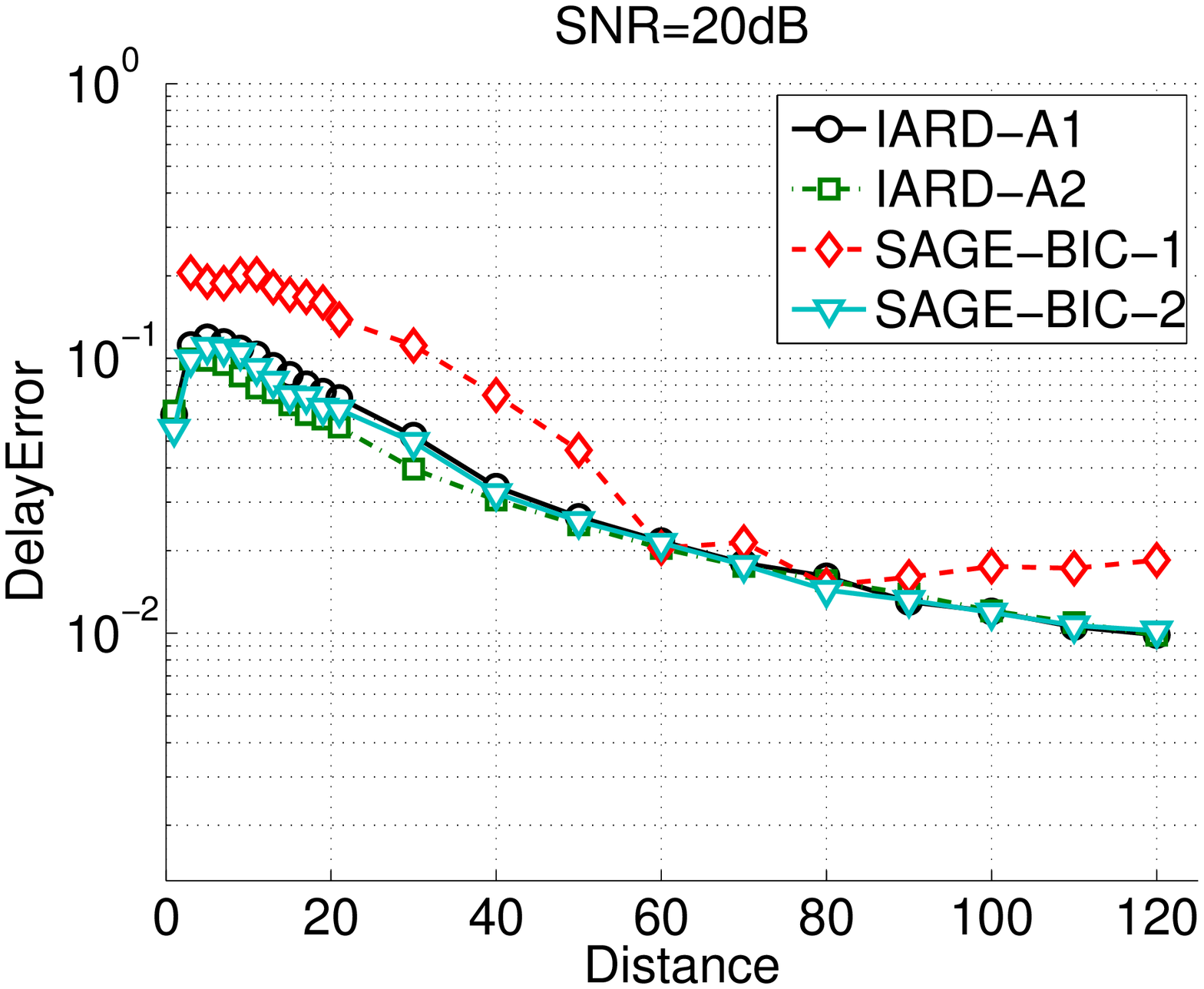}}
\subfigure[\label{fig:SNR30_Delay}]{\includegraphics[width=0.245\linewidth]{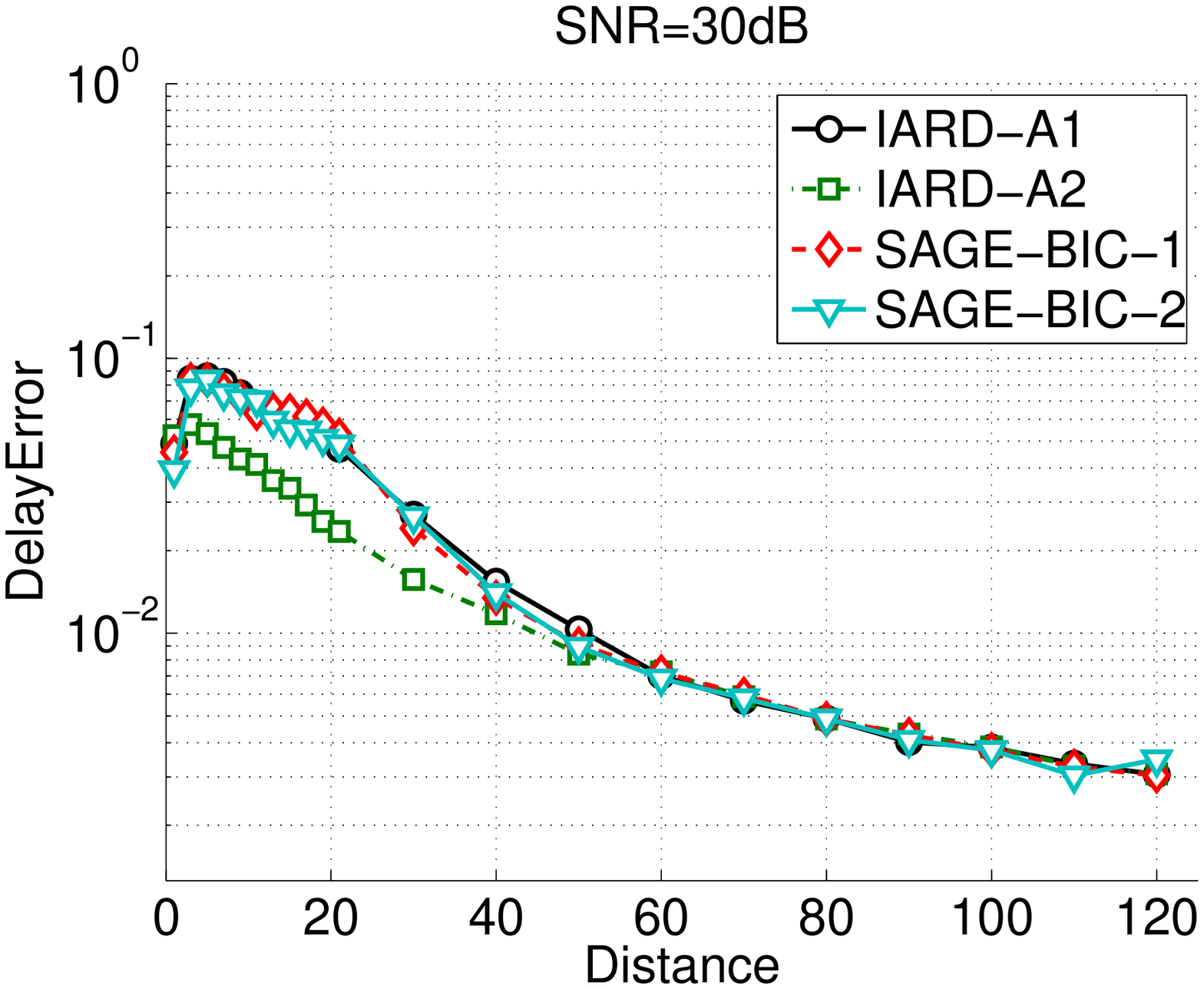}}
}
\centerline{
\subfigure[\label{fig:SNR5_Doppler}]{\includegraphics[width=0.245\linewidth]{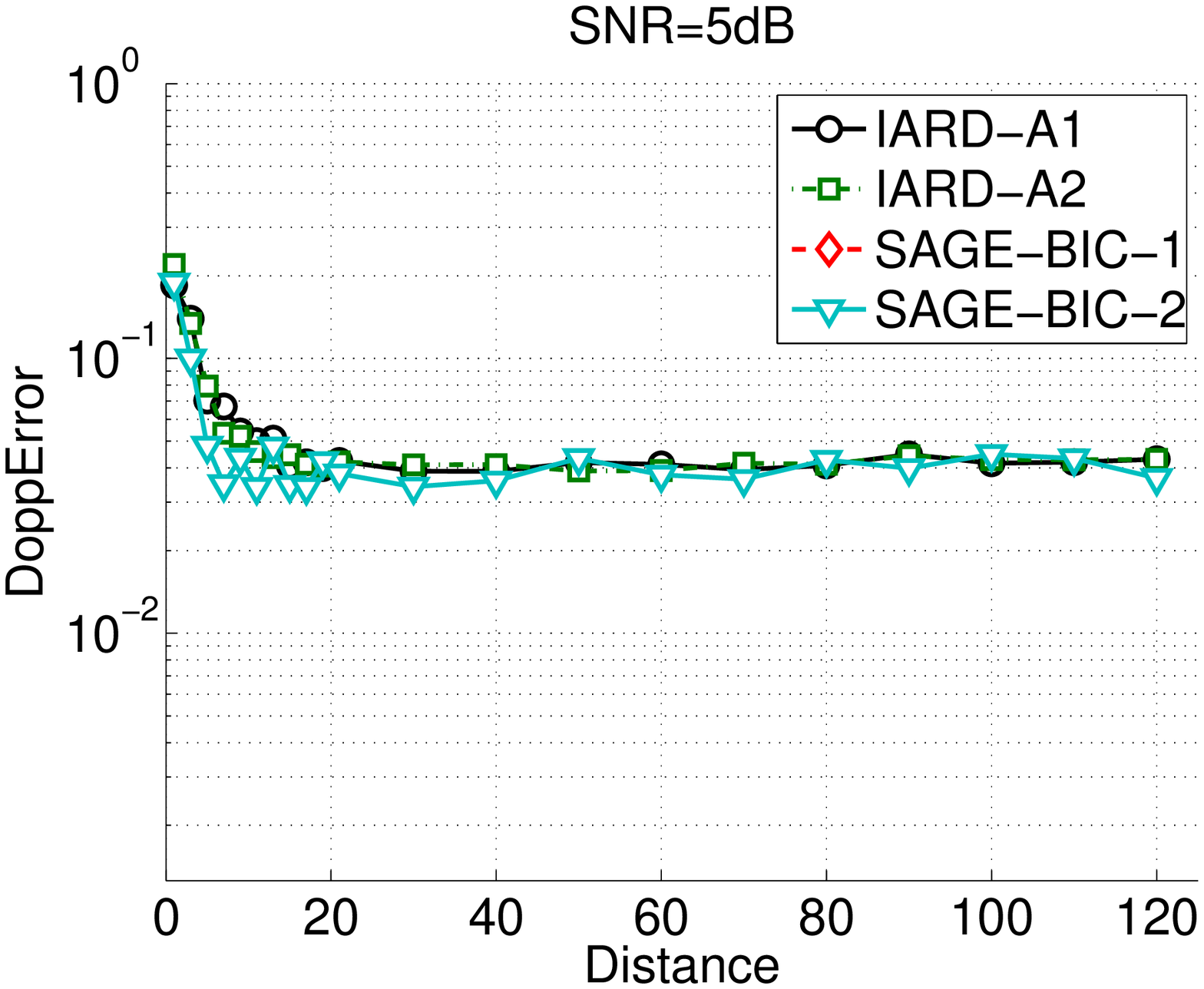}}
\subfigure[\label{fig:SNR10_Doppler}]{\includegraphics[width=0.245\linewidth]{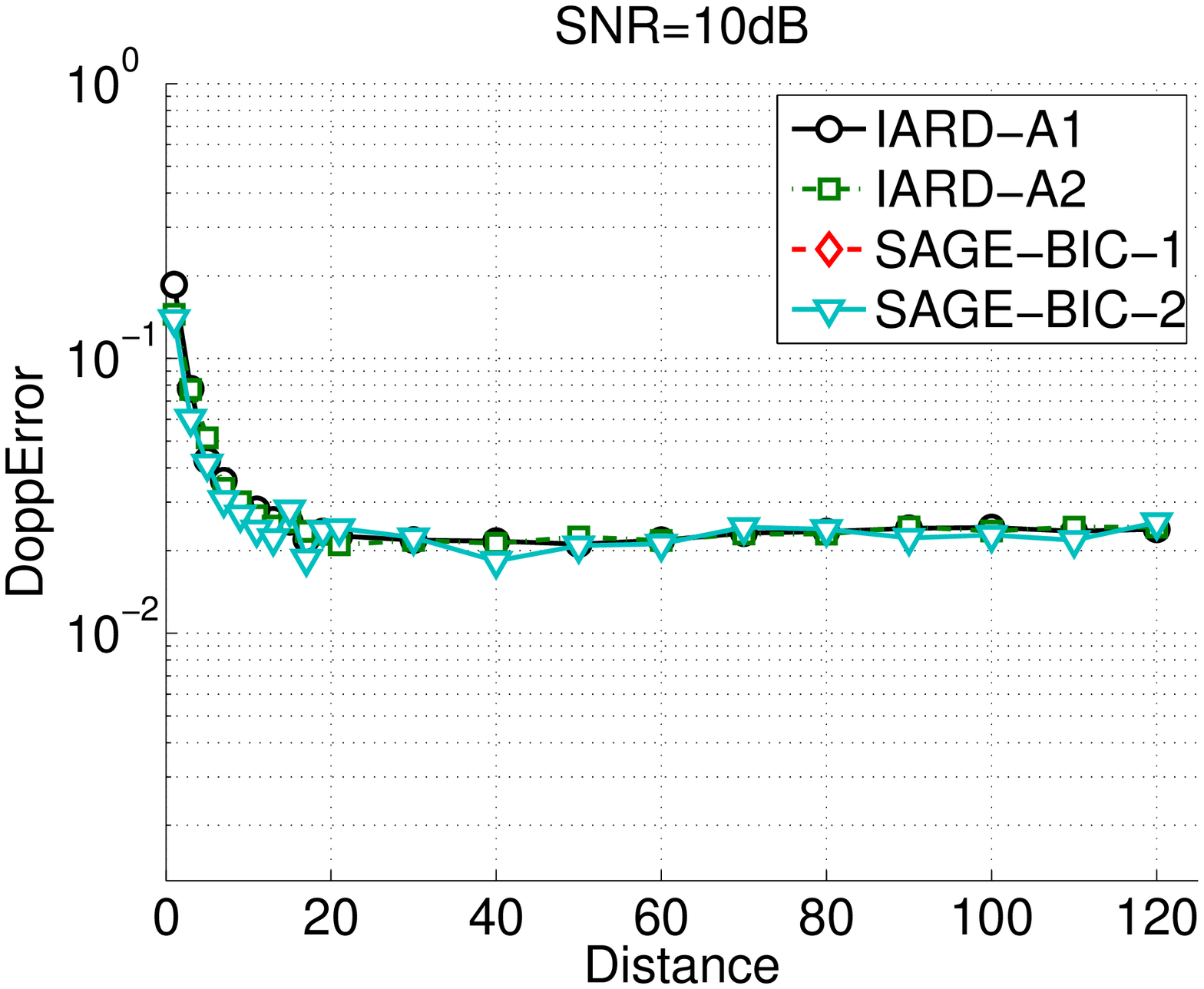}}
\subfigure[\label{fig:SNR20_Doppler}]{\includegraphics[width=0.245\linewidth]{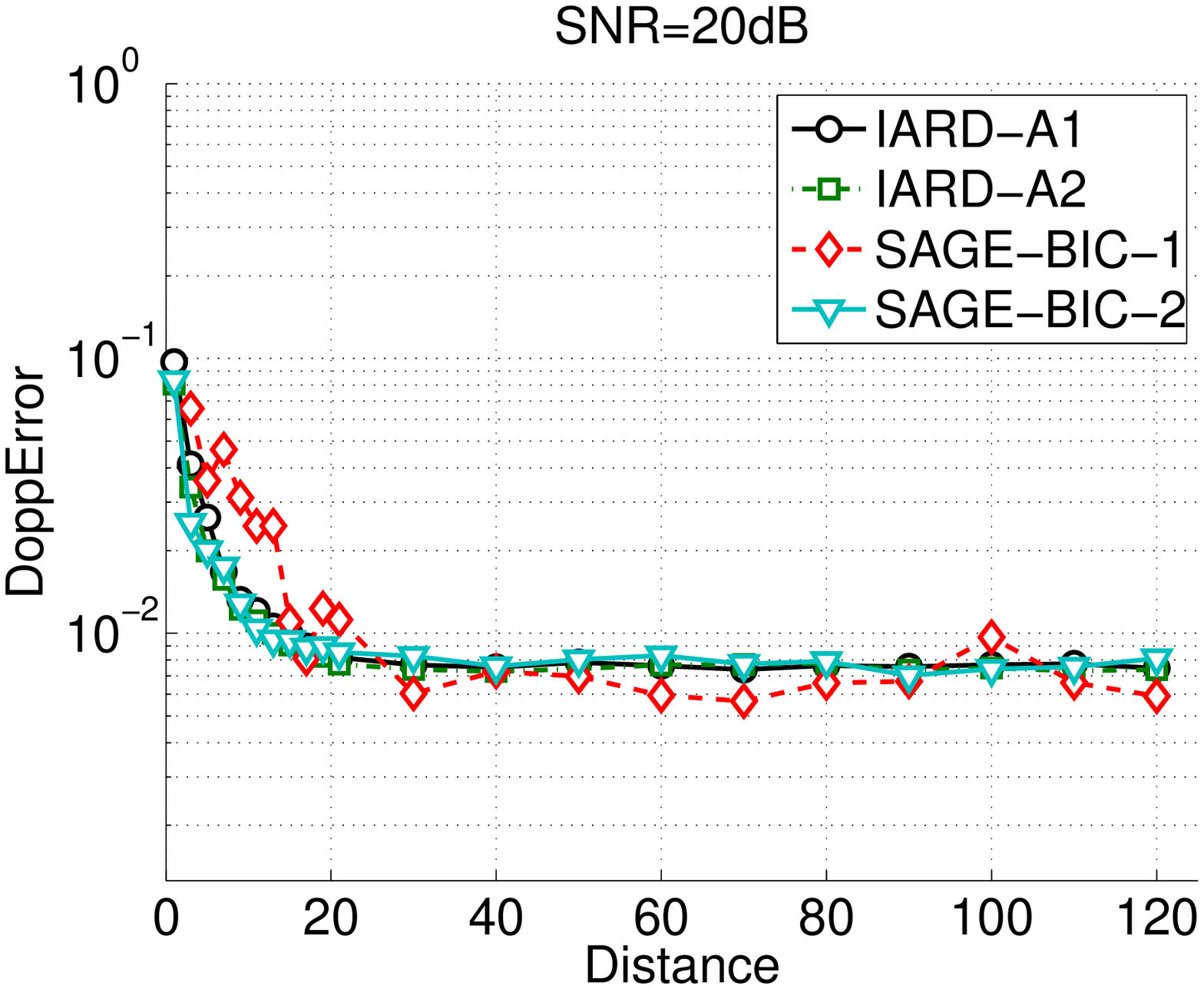}}
\subfigure[\label{fig:SNR30_Doppler}]{\includegraphics[width=0.245\linewidth]{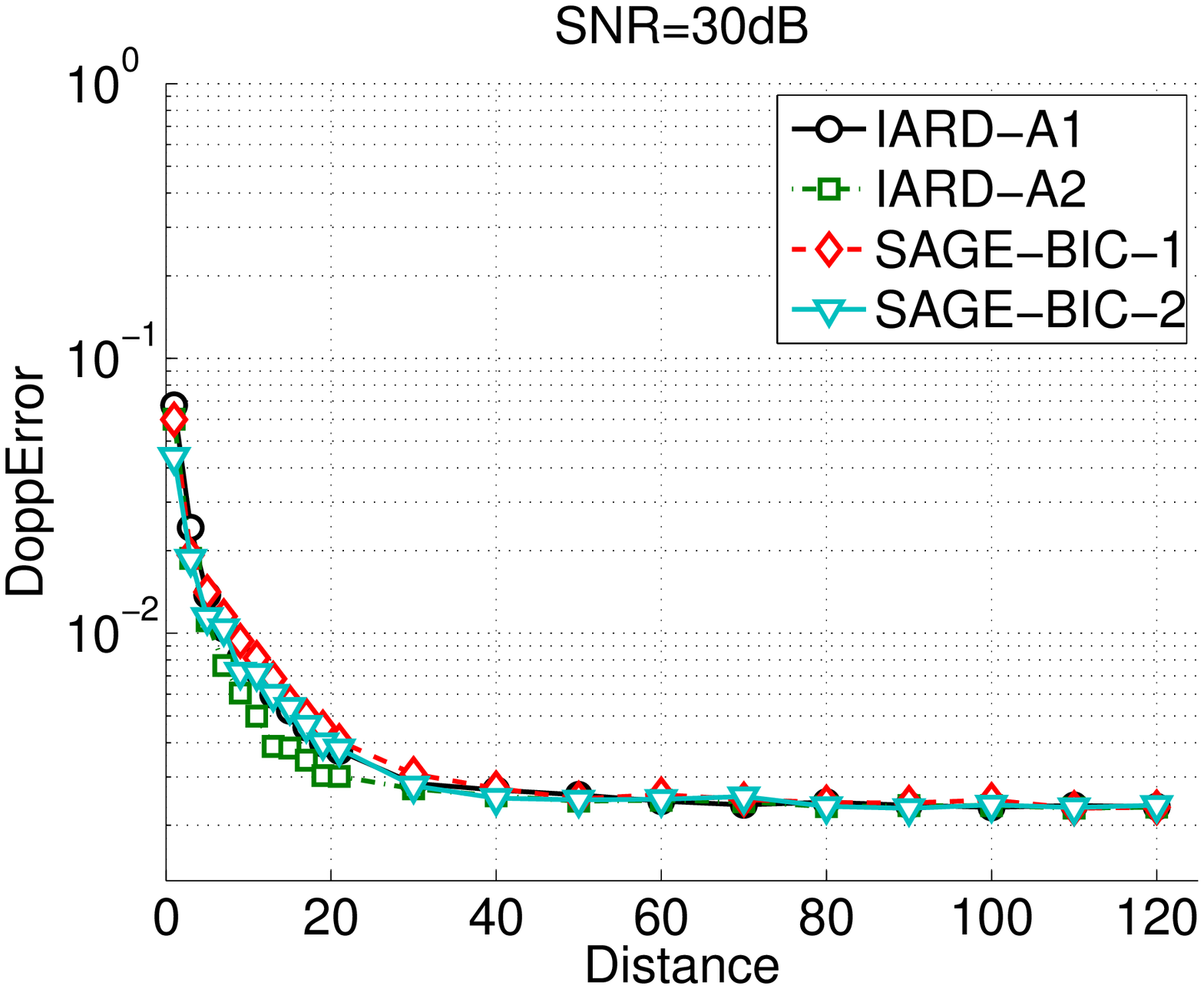}}
}
\centerline{
\subfigure[\label{fig:SNR5_Time}]{\includegraphics[width=0.245\linewidth]{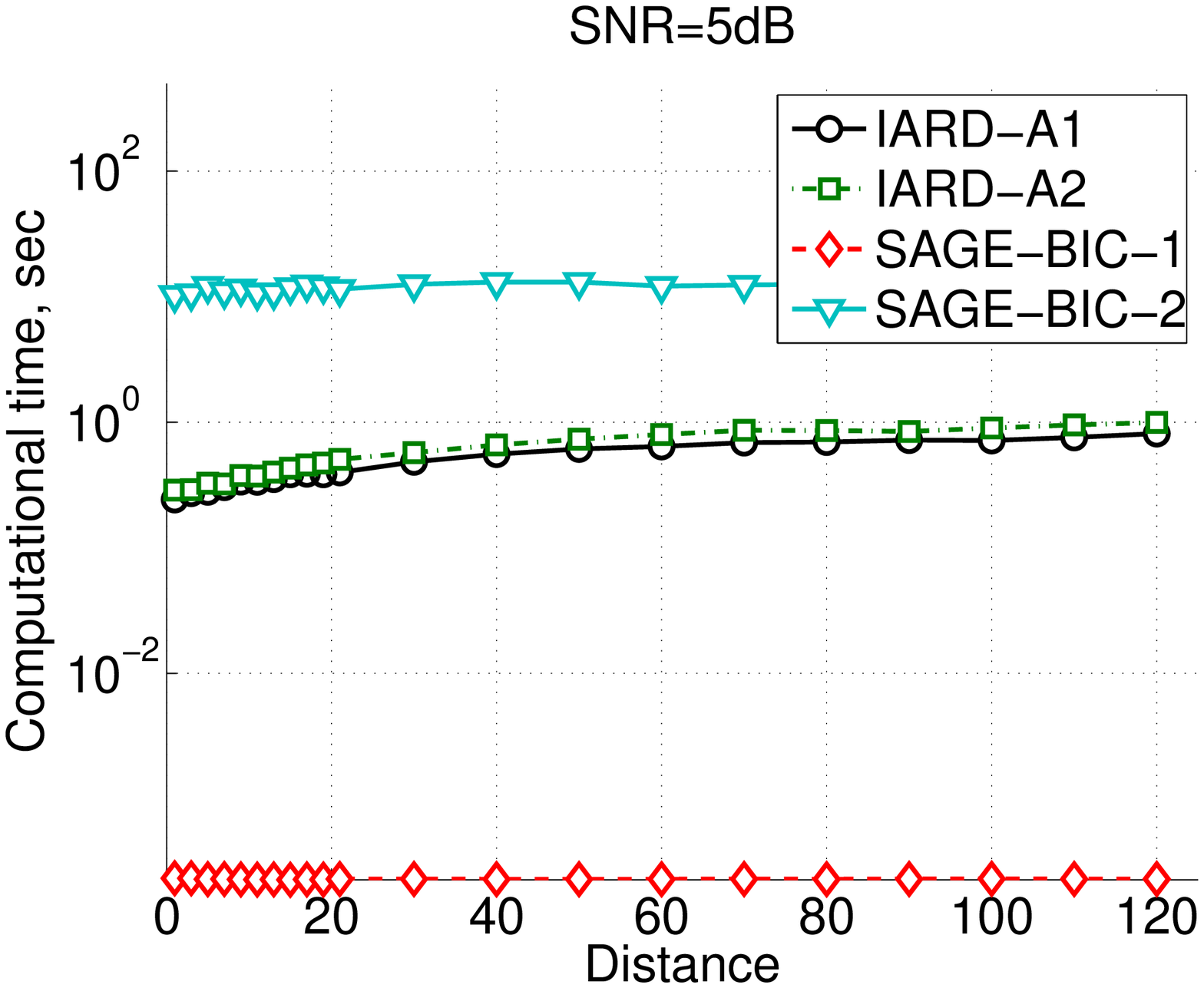}}
\subfigure[\label{fig:SNR10_Time}]{\includegraphics[width=0.245\linewidth]{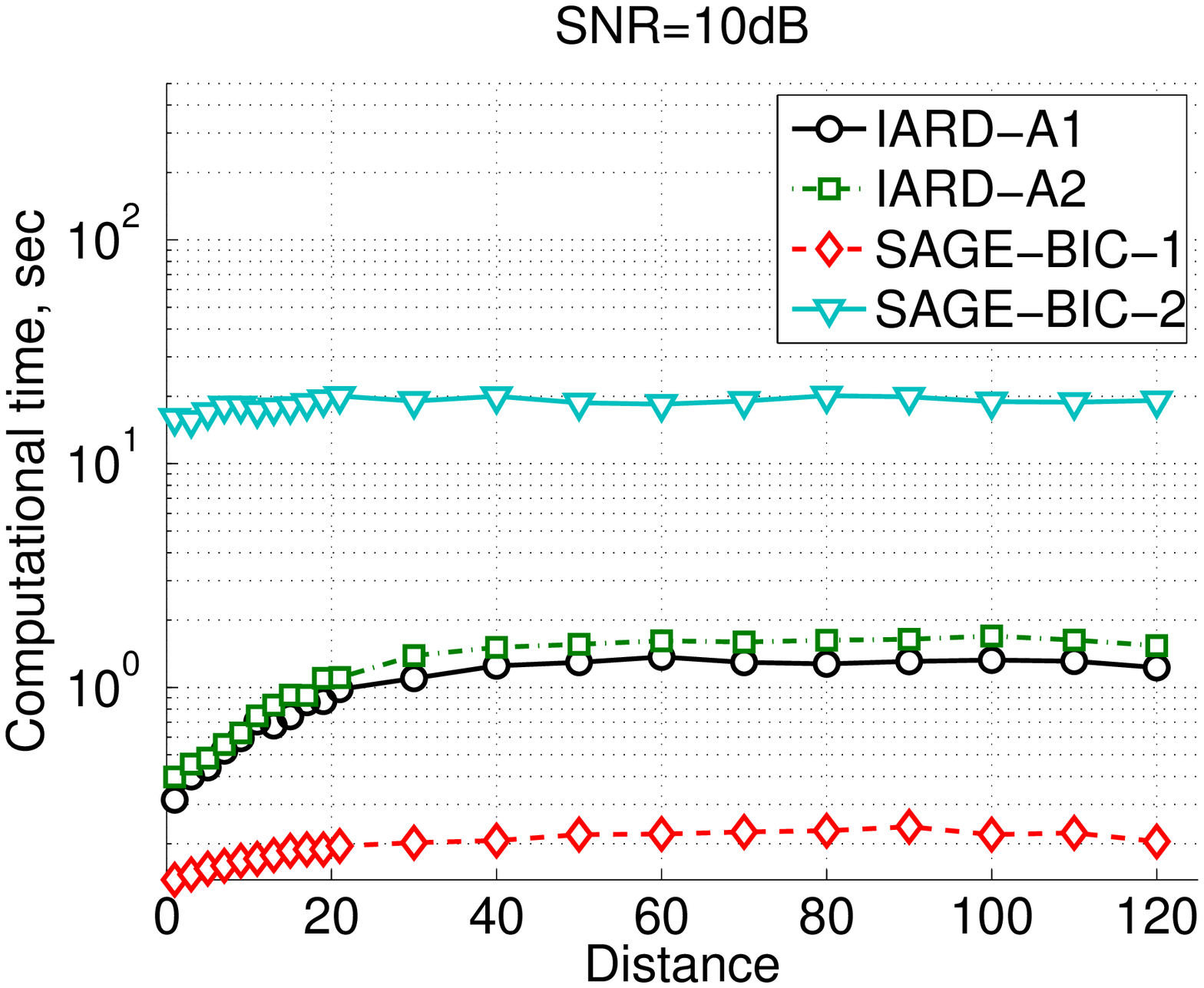}}
\subfigure[\label{fig:SNR20_Time}]{\includegraphics[width=0.245\linewidth]{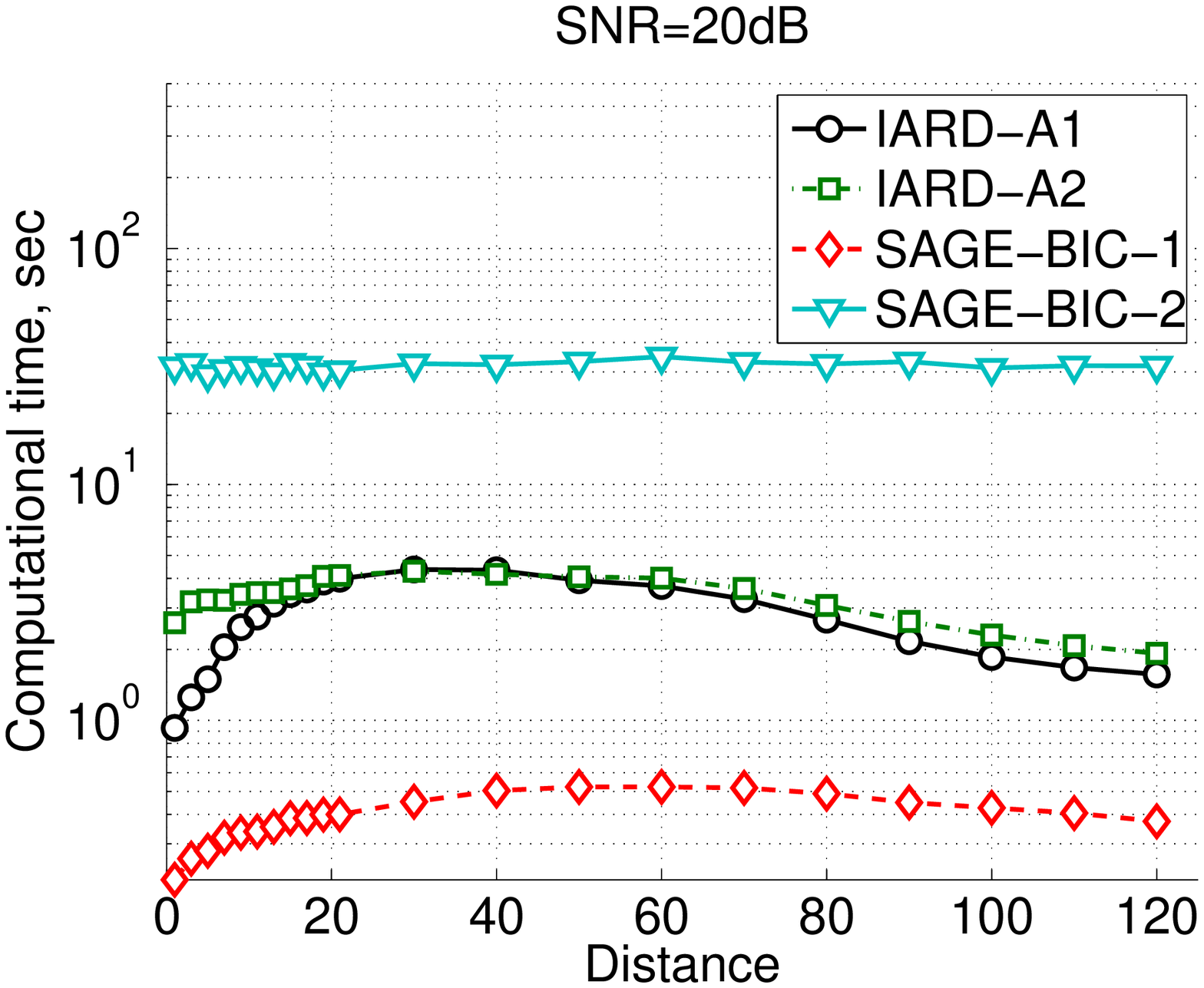}}
\subfigure[\label{fig:SNR30_Time}]{\includegraphics[width=0.245\linewidth]{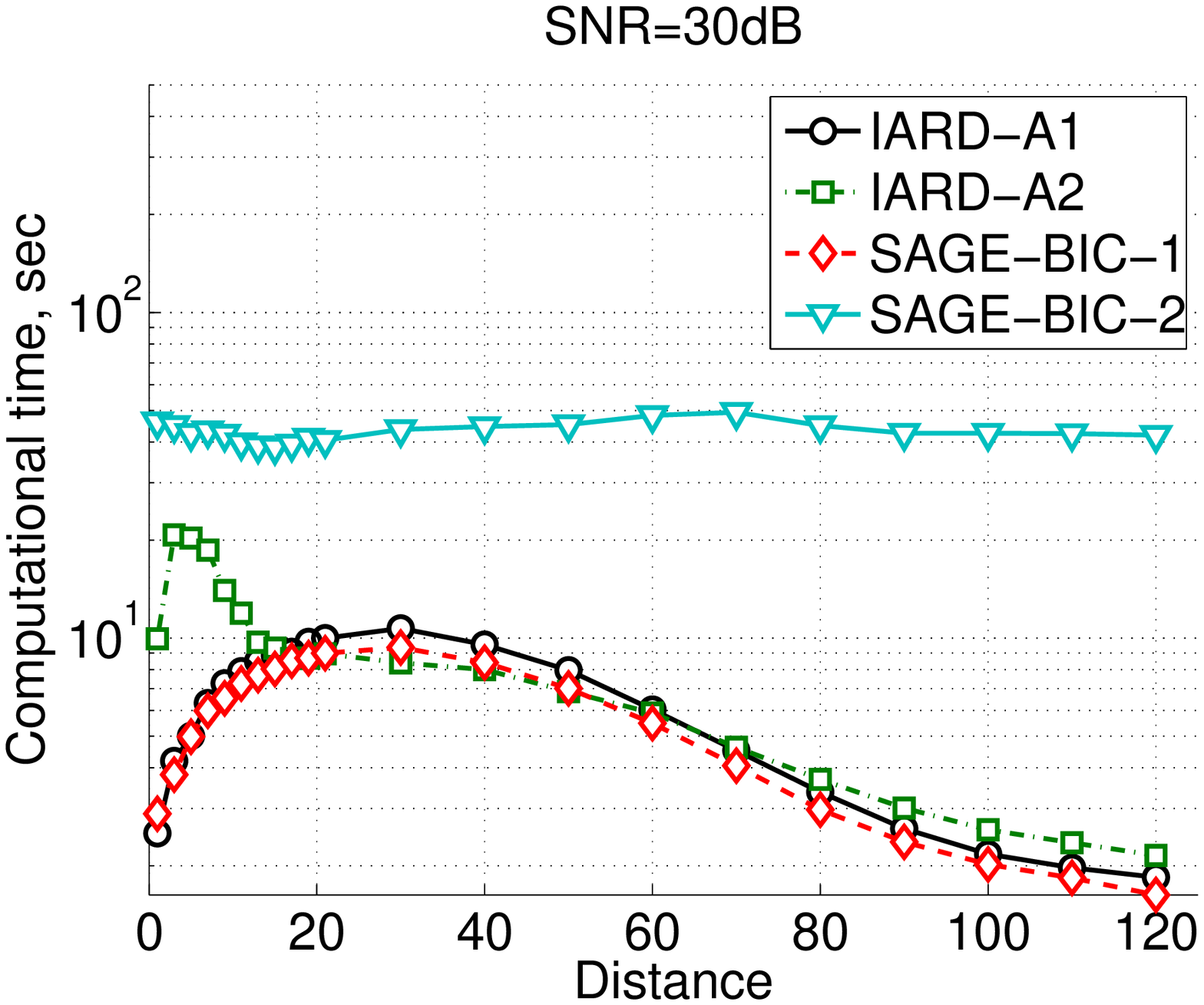}}
}
\caption{Estimation performance of the algorithm in the superresolution regime for (a-e) $5$dB, (f-j) $10$dB, (k-o) $20$dB, and (p-t) $30$dB SNR. Shown are (a,f,k,p) the averaged number of detected components $\widehat{L}$; (b,g,l,q) the probability $P_D^{L=2}$ of the detecting exactly two components;  (c,h,m,r) the normalized delay estimation RMSE, (d,i,n,s) the normalized Doppler estimation RMSE, and (e,j,o,t) the averaged computational time in seconds per single algorithm run.}
\label{fig:Resolution}
\end{figure*}
The corresponding plots for SNR $5$dB, $10$dB, $20$dB, and $30$dB are summarized in Fig.~\ref{fig:Resolution}.
The results are obtained by averaging over $2000$ independent Monte Carlo runs for the IARD-A1, IARD-A2, and SAGE-BIC-1 algorithms.
The statistics for the SAGE-BIC-2 algorithm are averaged over 300 Monte Carlo runs.

In terms of the estimated number of components $\widehat{L}$ (Fig. \ref{fig:SNR5_LEst}\--\ref{fig:SNR30_LEst}), and probabilities of detection $P_D^{(L=2)}$ (Fig. \ref{fig:SNR5_DProb}\--\ref{fig:SNR30_DProb}), the IARD-A1, IARD-A2, and SAGE-BIC-2 algorithms perform quite well, with the latter offering a slightly better performance. 
The SAGE-BIC-1 algorithm performs in contrast quite poorly: it either underestimates the number of components in low SNR regime, or consistently  overestimates the model order in the high SNR regime.
Its performance also seems to be insensitive to the component spacing $\Delta$.
In contrast, the number of correct detection for the other algorithms grows as $\Delta$ and SNR increases.

In terms of accuracy of parameter estimation (Fig. \ref{fig:SNR5_Delay}\--\ref{fig:SNR30_Delay} and \ref{fig:SNR5_Doppler}\--\ref{fig:SNR30_Doppler}) we see that in low SNR regime, SAGE-BIC-2 performs slightly better than the other algorithms.
In the high SNR regime, SAGE-BIC-2 and IARD-A1 perform identically well, with
IARD-A2 outperforming them for small component spacing $\Delta$ \--- the advantage of the assumption $\mathrm{A2}$ over a ``simpler'' assumption $\mathrm{A1}$.
For larger spacing $\Delta$, i.e., when the correlation between the components decreases, this advantage, however, disappears, and SAGE-BIC-2, IARD-A1 and IARD-A2 deliver similar performance.

Finally, let us consider the computational time of the algorithms (\ref{fig:SNR5_Time} - \ref{fig:SNR30_Time}).
It is interesting to note that although SAGE-BIC-2 has better component detection capabilities, its computational time is significantly higher, since multiple models with different number of components have to be estimated.
SAGE-BIC-1 algorithm is the fastest, since the model order selection is done prior to multipath parameter estimation \--- the most time-consuming part of the algorithm.
The IARD-A1 and IARD-A2 algorithms are much faster than SAGE-BIC-2, yet they offer a compatible performance both in terms of component detection probabilities, as well as in the parameter estimation accuracy.
For a higher number of components $L$ the inefficiency of the SAGE-BIC-2 algorithm will constitute itself quite significant.

The difference between the $\mathrm{A1}$ and $\mathrm{A2}$ assumptions exhibits itself only for component spacing $\Delta$ below approx. $60\%$, i.e., in a super-resolution regime.
In terms of the detection rate, both assumptions perform quite similarly.
As expected, the parameter estimation accuracy is better for the $\mathrm{A2}$ assumption, yet at the expense of slightly higher computational time.

%%%%%%%%%%%%%%%%%%%%%%%%%%%%%%%%%%%%%%%%%%%%%%%%%%%%%%%%%%%%%%%%%%%
\section{Conclusion}
This work discusses a joint sparse estimation and detection of multipath components within variational Bayesian framework.
The approach is based on a variational realization of incremental automatic relevance determination (IARD) algorithm \--- a Bayesian sparse signal reconstruction technique.
The variational Bayesian formulation of the algorithm permits extending the standard IARD algorithm for linear models to a problem of parameters estimation of superimposed signals, which requires nonlinear optimizations.
The sparsity is used to estimate the number of active signals in the model.

However, for the problem of super-resolution multipath component estimation, where an accurate model order selection is of a particular interest, it has been observed that IARD generally overestimates the number of components.
Here we have demonstrated that this can be explained by the model fitting step at which dispersion parameters of propagation paths are estimated.
This steps performs a nonlinear optimization  that adapts the dictionary matrix of the IARD algorithm.
As a consequence, the model overfits the measured signal and artifacts are inserted into the model.                                                                                                                                        %The algorithm has been investigated under two assumptions: $\mathrm{A1}$ uncorrelated components, also known as the mean-field approximation, and $\mathrm{A2}$ correlated components, where the weights of individual component are assumed to be correlated. In both cases IARD provides a computationally efficient way to test whether a sparsity parameter for a particular component is finite, which indicates a non-zero weight of the corresponding propagation path.

To overcome this we proposed a hypothesis test that exploits statistical structure of the IARD inference expressions.
We have shown that due to the optimization of multipath dispersion parameters, the corresponding sparsity parameters will follow an extreme value distribution under additive Gaussian noise assumption.
This interpretation permits a correction of sparsity-driven model order selection within IARD using binary hypotheses testing.
We have shown that the standard IARD approach is equivalent to a hypothesis test with a very high probability of false alarm, which explains model order overestimation.
By adjusting the IARD pruning conditions to guarantee the desired false alarm probability, the model order selection can be improved and correct order can be estimated even in challenging super-resolution regime.
Simulation studies have demonstrated that this adjustment allows extraction of the true signal sparsity in simulated scenarios and further acceleration of the convergence rate of the algorithm as compared to the classical information-theoretic model order selection schemes.

\bibliographystyle{IEEEtran}
\bibliography{./IEEEabrv,./detection}






% \end{thebibliography}

\end{document}